\documentclass{easychair}

\usepackage{doc}
\usepackage{graphicx}
\usepackage{xspace}
\usepackage{bm}
\usepackage{amsmath}
\usepackage{amsfonts}
\usepackage{amsthm}
\usepackage{empheq}
\usepackage{amssymb}
\usepackage{mathtools}
\usepackage{stmaryrd}
\usepackage{xcolor}
\usepackage{subcaption}
\usepackage{float}
\usepackage{hyperref}
\usepackage{dlogic}
\usepackage{wrapfig}
\usepackage{empheq}
\usepackage{geometry}
\usepackage{footmisc}
\usepackage{wrapfig}
\usepackage{cancel}

\usepackage{cite}

\theoremstyle{plain}
\newtheorem{lemma}{Lemma}

\newtheorem{theorem}{Theorem}
\newtheorem{corollary}{Corollary}

\theoremstyle{remark}
\newtheorem{remark}{Remark}
\newtheorem{example}{Example}
\newtheorem{counterexample}{Counterexample}

\theoremstyle{definition}
\newtheorem{definition}{Definition}

\newif\ifprintappendix 
\printappendixtrue % version with appendix

\title{Comparison Invariants for Verifying Control Invariance}

\author{
Promit Panja
\and
Andr\'{e} Platzer
}

\institute{
  Department of Informatics,
  Karlsruhe Institute of Technology, Karlsruhe, Germany\\
  \email{\{promit.panja, andre.platzer\}@kit.edu}
 }

\authorrunning{P. Panja, A. Platzer}

\titlerunning{Comparison Invariants for Verifying Control Invariance}

\begin{document}

\maketitle

\begin{abstract}
  Control invariance validates that dynamical systems have a control input that 
  preserves a given property at all times. This paper introduces a set of sound axioms 
  and proof rules in differential dynamic logic (\dL) that enable 
  verification of control invariance. First, the \emph{scalar} and \emph{vector comparison principles}, 
  relating a system of differential equations to a comparison system such that 
  invariance properties can be established more easily, are axiomatized in \dL. This 
  axiomatization primarily utilizes \emph{differential ghosts}, which are 
  proof-theoretic generalizations of comparison systems. Next, with the comparison principles serving as the basis, \emph{comparison invariants} 
  are introduced, and sound axioms and proof rules are derived. Comparison 
  invariants reduce the question of control invariance to a functional 
  inequality on its Lie derivative for a suitable class of functions, moreover, 
  the right choice of function can result in decidable arithmetic. Furthermore, the perennially 
  popular control barrier functions (CBFs) used in safety-critical control are shown to 
  be a special instance of comparison invariants. This yields an axiomatization of 
  CBFs that leads to a dedicated set of proof rules. The rules allow for the verification 
  of CBFs, which are traditionally used for synthesizing safe controllers without 
  verification. Lastly, comparison invariants are shown to unify several other safety verification techniques, 
  including Darboux invariants and differential invariants, further cementing their 
  versatility. 
\end{abstract}

\section{Introduction}

Control systems and cyber-physical systems are an integral part of modern society \cite{alur2015principles,Platzer18}, 
which are deployed in various domains such as aerospace, medical devices, autonomous vehicles, 
industrial robots, and power grids \cite{sastry2013nonlinear,Platzer10}. Control systems 
are often termed \emph{safety-critical} when safety is a primary concern, breach of 
which can cause catastrophic damage including loss of life \cite{Platzer18,Platzer10,alur2015principles}. Consequently, 
safety-critical control systems require formal safety guarantees. Making sure that 
a control system behaves safely under all admissible control inputs has been a central 
challenge in control theory and formal verification \cite{sastry2013nonlinear,khalil2002nonlinear,Platzer18}. 
At the heart of this challenge lies control invariance, the principle that a system's differential
equations can admit control inputs that preserve a desired property.

This paper studies deductive verification of control invariance by developing an 
axiomatic foundation in \emph{differential dynamic logic} (\dL) \cite{DBLP:journals/jar/Platzer08,Platzer10,DBLP:conf/lics/Platzer12b,DBLP:journals/jar/Platzer17,Platzer18,DBLP:conf/lics/PlatzerT18,DBLP:journals/jacm/PlatzerT20,DBLP:journals/jacm/PlatzerQ25}, a logic for modeling 
and proving hybrid systems. In order to axiomatically prove control invariance we utilize 
the \emph{comparison principles} \cite{walter2013ordinary} as the underlying proof technique. 
The comparison principles enable relating a system of differential equations with a 
comparison system aiding invariance proofs by establishing properties that were not possible 
before. Based on the comparison principles we introduce a new general class of invariants 
called \emph{comparison invariants}, which reduce the question of control invariance 
to questions about real arithmetic, specifically a functional inequality on the Lie derivative 
of the invariants, proving sublevel sets as control invariant. Comparison invariants are used to further generalize various safety 
verification techniques including the ever popular control barrier functions \cite{ames2019control}.

Control barrier functions are a go-to choice in safety-critical control theory for 
synthesizing efficient safety enforcing controllers \cite{ames2014control,ames2016control,ames2019control,xiao2021high}. 
Due to their relative simplicity and their formal safety guarantees control barrier functions 
have attracted much success, especially in fields like safety-critical robotics. However, 
these formal guarantees are only enforced by the synthesized controllers if the control barrier function 
used to synthesize it was valid in the first place. Verifying whether a candidate 
control barrier function is valid or not is a nontrivial task \cite{clark2024semi}. 
Recently, there have been efforts in verifying control barrier functions, primarily 
utilizing sum-of-squares optimization using semidefinite programming \cite{clark2024semi}, however, 
in-practice these methods can suffer from potential soundness issues due to round-off errors in floating-point solvers.

Due to a lack of formal logical embeddings for verifying control barrier functions, this work 
was motivated to bridge this gap by providing an axiomatic foundation for proving 
control invariance in general, and as a consequence, also control barrier functions. With this 
work our aim is to provide a unified proof framework in \dL, which allows proving various 
control invariance properties via a general set of axioms and proof rules. Moreover, we propose a logical foundation for verifying control barrier functions for the safety-critical control community for synthesizing safe controllers, in which unverified candidate control barrier functions abound.

\paragraph{Contributions.} The main contributions are summarized as follows.
\begin{enumerate}
  \item We first axiomatize the \emph{scalar} and \emph{vector comparison principles}, 
  and show they are derivable in \dL resulting in axioms \lref{ax:SCP} and \lref{ax:VCP}, respectively.
  \item Next, we introduce \emph{scalar} and \emph{vector comparison invariants} 
  along with their respective axioms and proof rules, 
  enabling sound verification of control invariance properties in \dL.
  \item Further, we give an axiomatization of \emph{control barrier functions} (CBFs), scalar 
  and vector, and in the process show how they are a special case of comparison invariants.
  \item Additionally, we show the existing \dL axioms and proof rules for Darboux invariants and 
  differential invariants, respectively, are derivable from comparison invariants.
\end{enumerate}

\section{Background: Differential Dynamic Logic}
 
Differential dynamic logic (\dL), a logic for modeling and reasoning about hybrid 
systems \cite{Platzer18}, forms the basis of this paper. We will briefly go over the syntax, semantics, 
and the proof calculus, followed by control systems modeling principles in \dL. 
A detailed exposition of the logic can be found elsewhere 
\cite{DBLP:journals/jar/Platzer17,Platzer18}.

\subsection{Syntax}
Here we give an overview of the \dL syntax, which comprises \emph{terms} and \emph{formulas} 
that are then combined to form \emph{hybrid programs} modeling hybrid systems.

\subsubsection{Terms}\label{syntax:terms}
The \emph{term language} of \dL follows the following grammar, where 
$x \in \vars$ is a variable in the set of all variables $\vars$, and $c \in \Q$ is 
a rational constant:
\[
  p, q \hspace{0.25em} ::= \hspace{0.25em} x \hspace{0.25em} | \hspace{0.25em} c \hspace{0.25em}
  | \hspace{0.25em} p + q \hspace{0.25em} | \hspace{0.25em} p \cdot q \hspace{0.25em} | 
  \hspace{0.25em} \eta(p_{1}, \dots, p_{n}) \hspace{0.25em} 
  | \hspace{0.25em} (p)'.
\]
Also, $x = (x_{1}, \dots, x_{n})^{\top}$ 
is a vector of variables. It is worth noting that this paper frequently 
utilizes (fixed) function symbols like $\eta(\cdot)$, which are scalar functions, 
and $\boldsymbol{\eta}(\cdot) = (\eta_{1}(\cdot), \dots, \eta_{n}(\cdot))^{\top}$ 
are vector-valued functions, where $\eta_{1}(\cdot), \dots, \eta_{n}(\cdot)$ are scalar functions. 
All the function symbols are assumed to be continuous in their respective arguments. For further 
distinction, vector of terms are written in bold $\mathbf{p} = (p_{1}, \dots, p_{n})^{\top}$.
The term language also includes \emph{differential terms} 
$(p)'$ that allow reasoning about derivatives locally, along the evolution of associated 
\emph{ordinary differential equations} (ODEs) \cite{DBLP:journals/jacm/PlatzerT20}, and are defined as follows:
\begin{equation*}
    (p)' \defeq \sum_{i=1}^{n} \frac{\partial p}{\partial x_{i}} \cdot x'_{i} \hspace{3em}
    \Lie{\mathbf{f}(x)}{p} \defeq \sum_{i=1}^{n} \frac{\partial p}{\partial x_{i}} \cdot f_{i}(x).
\end{equation*}
The differential of a term $p$ coincides with its syntactic Lie derivative along 
the ODE $x' = \mathbf{f}(x)$, with $x' = (x_{1}', \dots, x_{n}')^{\top}$ and $\mathbf{f}(x) = (f_{1}(x), \dots, f_{n}(x))^{\top}$, 
when following the chain rule. Here we use $\Lie{\mathbf{f}(x)}{p}$ as 
an operator and use $\dot{p}$ as the syntactic 
representation of the Lie derivative for brevity. For two vectors $\mathbf{v}$ and $\mathbf{w}$ the 
inner product is denoted as $\mathbf{v} \cdot \mathbf{w}$ and the Euclidean norm is 
$\norm{\mathbf{v}} = \sqrt{\mathbf{v} \cdot \mathbf{v}}$. Since standard norms are not part of \dL term syntax, 
squared Euclidean norms are used whenever necessary, as they are definable in first-order logic over 
real arithmetic ($\text{FOL}_{\mathbb{R}}$) as $ \norm{\mathbf{v}}^{2} = \sum_{i=1}^{n} v_{i}^{2}$.

\subsubsection{Formulas}
The \emph{formulas} of \dL are defined by the following grammar, where $\sim \hspace{0.3em} \in \{>, \geq, =\}$ 
is a comparison operator and $\alpha$ is a \emph{hyrbid program} \cite{Platzer18}:
\[
  \phi, \psi \hspace{0.25em} ::= \hspace{0.25em} p \sim q \hspace{0.25em} | \hspace{0.25em} 
  \neg \phi \hspace{0.25em} | \hspace{0.25em} \phi \Land \psi \hspace{0.25em} | \hspace{0.25em} 
  \phi \Lor \psi \hspace{0.25em} | \hspace{0.25em} \forall x \phi \hspace{0.25em} | 
  \hspace{0.25em} \exists x \phi \hspace{0.25em} | \hspace{0.25em} \fBox[\alpha][\phi] 
  \hspace{0.25em} | \hspace{0.25em} \fDia[\alpha][\phi].
\]
This grammar extends the standard first-order logic over real arithmetic 
($\text{FOL}_{\mathbb{R}}$), with the box 
$\fBox[\alpha][\phi]$ and the diamond $\fDia[\alpha][\phi]$ modalities, which express 
that \emph{all} or \emph{some} runs of the hybrid program $\alpha$ satisfy 
$\phi$, respectively.

\subsubsection{Hybrid Programs}
The language of \textit{hybrid programs} is defined by the following grammar, 
where $x$ is a variable, $p$ is a \dL term, and $Q$ is a formula of first-order 
real arithmetic \cite{DBLP:journals/jar/Platzer17,Platzer18}.
\[
  \alpha, \beta \hspace{0.25em} ::= \hspace{0.25em} x' = \mathbf{f}(x) \hspace{0.1em} \& 
  \hspace{0.1em} Q \hspace{0.25em} | \hspace{0.25em} \alpha ; \beta \hspace{0.25em} 
  | \hspace{0.25em} \alpha^*
\]
Continuous dynamics are modeled by the $n$-dimensional system of autonomous\footnote{An ODE of the form $x' = f(x)$ is called \emph{autonomous}, if the right-hand side of the ODE, i.e., $f(x)$ only depends on the state variable $x$ and \textbf{not} explicitly on time.}
ODEs $\ode[x]{\mathbf{f}(x)}[Q]$ whose evolution is constrained by 
the domain $Q$; the ODE is modeled as $x' = \mathbf{f}(x)$ when 
$Q \equiv true$. The continuous and discrete fragments can be 
combined to model hybrid dynamics using the following program combinators: 
\textit{sequential composition} $\alpha ; \beta$, which executes program $\alpha$ 
followed by program $\beta$, and (bounded)
\textit{nondeterministic repetition} $\alpha^*$, which repeats the program $\alpha$ 
zero or arbitrarily more number of times. The syntax of the discrete fragment of \dL 
and other program combinators can be found elsewhere \cite{Platzer18}.

\begin{remark}\label{rem:time}
  This paper utilizes systems of ODEs of the form 
  $\odes{x' = \mathbf{f}(x,\tau), \tau' = 1}[Q]$,
  where the variable $\tau$ is specifically meant for tracking time. Due to the  
  ODE $\tau' = 1$, the system of ODEs as a whole is still autonomous in the extended 
  state space \cite[Ch. 1, \S1.3]{chicone2006ordinary}. This choice is made so that 
  once can show ODE solutions remain 
  bounded in a compact time interval, which would not have been possible to show 
  syntactically in \dL's proof calculus without the addition of the time variable 
  $\tau$.
\end{remark}

\subsection{Semantics}
The semantics of \dL are defined such that states $\nu: \vars \to \real$ assign 
a real value to each variable in $\vars$. For some state 
$\nu \in \States$ and some term $p$, where $\States \subseteq \R^{n}$ denotes 
the set of all states, the semantics $\nu\sem{p}$ is the real value of $p$ 
when evaluated at $\nu$. For a formula $\phi$, the semantics $\sem{\phi} \subseteq \States$
is the set of states in which $\phi$ holds. First-order connectives have their usual 
semantics, e.g., $\sem{\phi \Land \psi} = \sem{\phi} \cap \sem{\psi}$.

Now we present the semantics of formulas most relevant to this paper, i.e., 
semantics of ODEs under the box modality.\footnote{The complete \dL semantics including semantics of ODEs under the diamond modality 
can be found elsewhere \cite{DBLP:journals/jar/Platzer17,Platzer18,DBLP:journals/jacm/PlatzerT20}.}
Let $\varphi: [0, T) \to \States$ (for some $0 < T \leq \infty$), 
be the unique solution of the ODE $x' = \mathbf{f}(x)$ extended maximally to the right \cite{walter2013ordinary,chicone2006ordinary} 
with initial value $\varphi(0) = \nu$, where $\nu \in \States$ is some state:
\begin{align*}
  \nu \in \sem{\fBox[\ode[x]{\mathbf{f}(x)}[Q]][\psi]} &\text{ iff for all } 0 \leq t < T, \varphi(\xi) \in \sem{Q} \text{ for all } 0 \leq \xi \leq t \text{ such that } \varphi(t) \in \sem{\psi}.
\end{align*}
Informally it means, the formula $\fBox[\ode[x]{\mathbf{f}(x)}[Q]][\psi]$ is true in the initial 
state $\nu$ if \emph{all} states reached by the solution of the ODE from $\nu$ 
satisfy the postcondition $\psi$, while staying in the domain constraint $Q$.

\subsection{Proof Calculus}
\begin{wrapfigure}{r}{0.28\textwidth}
  \centering
  \vspace{-25pt}
  \begin{sequentproof}
    \AxiomC{$\Gamma_{1} \vdash \phi_{1}$}
    \AxiomC{$\dots$}
    \AxiomC{$\Gamma_{n} \vdash \phi_{n}$}
    \TrinaryInfC{$\vdots \ \begin{pmatrix}
      \text{\small hybrid program} \\
      \text{\small reasoning}
    \end{pmatrix}$}
    \un[]{\Gamma}{[\alpha]\phi}
  \end{sequentproof}
  \vspace{-10pt}
  \caption{Proofs proceed upwards by deduction, where each reasoning step is justified by sound \dL axioms and proof rules.}
  \vspace{-10pt}
\end{wrapfigure}

The \dL proof calculus enables deductive verification of hybrid 
dynamical systems by utilizing sound axioms and proof rules of \dL via compositional 
reasoning principles of hybrid programs \cite{DBLP:journals/jar/Platzer17,Platzer18}. 
All derivations use classical sequent calculus with standard rules for logical 
connectives and sequents. A sequent is of the form $\Gamma \vdash \Delta$ with 
semantics $(\bigwedge_{\psi \in \Gamma}\psi) \to \Delta$, and  is valid if and only 
if the formula is valid. Completed proof branches are marked with $*$. Since 
first-order real arithmetic is decidable \cite{bochnak2013real}, we assume access 
to a decision procedure and mark steps justified by real arithmetic with the rule \lref{Real}. 
An axiom is sound if all its instances are valid. A proof rule is sound  if validity 
of all its premises implies the validity of its conclusion. Axioms and proof rules 
are derivable if they can be obtained from sound \dL axioms and proof rules. The 
soundness of the base \dL axiomatization then guarantees the soundness of all derived 
axioms and proof rules \cite{DBLP:journals/jar/Platzer17,Platzer18,DBLP:journals/jacm/PlatzerT20}. 
An overview of the base \dL calculus is provided in Appendix~\ref{appendix:base-calculus}.

\subsection{Control Systems in \dL}

In general, control systems are modeled by parameterized ODEs with an explicit control parameter, 
which represents the control input. For a (compact) control admissible set $\mathcal{U} \subseteq \R^{m}$, 
following is a control system, where $x \in \mathcal{X} \subseteq \R^{n}$ is 
the state variable and $u \in \mathcal{U}$ is the control input:
\[x' = \mathbf{f}(x, u)\]

In \dL, however, control systems are modeled as hybrid programs; discrete programs 
and continuous dynamics paired together along with nondeterministic repetition. 
\begin{equation}\label{eq:control-sys-dl}
  \Gamma \to [\underbrace{(ctrl;dyn)^{*}}_{\alpha^{*}}]P
  \hspace{6em}
  \vcenter{\hbox{%
  \AxiomC{$\Gamma \vdash Inv$}
  \AxiomC{$\textcolor{RPTHMediumBlue}{Inv \vdash [\alpha]Inv}$}
  \AxiomC{$Inv \vdash P$}
  \LeftLabel{\axtag{loop} }
  \TrinaryInfC{$\Gamma \vdash [\alpha^{*}]P$}
  \DisplayProof
  }}
\end{equation}

In the above formula, $ctrl$ refers to the \emph{(discrete) controller} and $dyn$ refers 
to the \emph{continuous dynamics}, i.e., a system of ODEs. The central reasoning principle for proving 
hybrid programs is the \lref{loop} rule \cite[Ch. 7, \S7.3, Lemma 7.3]{Platzer18}.
This proof rule decomposes the hybrid program $\alpha^{*}$ in an inductive fashion. 
The most crucial step is showing that formula $Inv$ is an invariant for the hybrid 
program, more importantly the ODEs in $\alpha$. Because of this reason the rest of 
this paper focuses on proving invariants of ODEs, specifically, control invariance 
of ODEs, which will be defined and looked into in detail in the later sections.

\section{Comparison Principle}\label{sec:comparison-priciple}

We begin our endeavor by first taking a look at \emph{comparison systems} and 
the \emph{comparison principle}, which are fundamental tools in the theory of 
dynamical systems enabling differential inequality reasoning using differential 
equations \cite{walter2013ordinary,bellman1962vector,chatterjee2006stability,nersesov2006stability}.

\paragraph{Comparison System.} Informally, a \emph{comparison system} is an auxiliary 
system of differential equations, which establishes a relation with the original system 
of differential equations enabling the analysis of properties that were 
not possible before. For example, for a 
system of ODEs $x' = \mathbf{f}(x)$, if one wants to show a property $R(x)$ holds true given 
some property $P(x)$ was true initially, instead of proving this implication directly, which often times can 
be difficult, one might construct a comparison system $y' = \mathbf{g}(x,y)$ that enables 
proving from $P(x)$ a new suitable intermediate property $Q(x,y)$, which implies the intended $R(x)$.
\begin{equation}\label{eq:comparison-formula}
   [x' = \mathbf{f}(x)]P(x) \to [x' = \mathbf{f}(x), y' = \mathbf{g}(x,y)]Q(x,y) \to [x' = \mathbf{f}(x)]R(x)
\end{equation}
At a high-level, the above \dL formula shows how this relation with a comparison 
system can be established for a suitable $Q(x,y)$ and $y' = \mathbf{g}(x,y)$.
\begin{remark}
  Of course, formula (\ref{eq:comparison-formula}) requires handling side conditions on 
  the solutions of $y' = \mathbf{g}(x,y)$ to maintain soundness. In \dL this sound  
  reasoning is achieved by \emph{differential ghosts} 
  \cite{DBLP:journals/jar/Platzer17,DBLP:journals/fac/TanP21,DBLP:journals/jacm/PlatzerT20}, 
  which are proof-theoretic generalizations of comparison systems. The \dL axioms \lref{ax:DG} 
  and \lref{ax:BDG} allow introducing auxiliary dynamics that abstract from
  the existing dynamics to simplify proofs.
\end{remark}

\begin{lemma}[Differential Ghosts {\cite{DBLP:journals/jacm/PlatzerT20,DBLP:journals/fac/TanP21}}]
  The axioms \emph{\lref{ax:DG}} and \emph{\lref{ax:BDG}} enable sound additions of \emph{linear} 
  and \emph{nonlinear} ghost ODEs, respectively, to the existing systems.
  \begin{align*}
    \axtag{ax:DG} \quad & \fBox[\ode[x]{\mathbf{f}(x)}[Q]][\psi] \Liff \exists y \fBox[\odes{x' = \mathbf{f}(x), y' = a(x)y + b(x)}[Q]][\psi] \\
    \axtag{ax:BDG} \quad & \fBox[\ode[x]{\mathbf{f}(x)}, \ode[y]{\mathbf{g}(x, y)}[Q]][\norm{y}^{2} \leq p(x)] \to \\ 
    & (\fBox[\ode[x]{\mathbf{f}(x)}[Q]][\psi] \Liff \fBox[\ode[x]{\mathbf{f}(x)}, \ode[y]{\mathbf{g}(x, y)}[Q]][\psi])
  \end{align*}
\end{lemma}

A cleverly constructed comparison system can enable reasoning about properties of 
ODEs with (relative) ease. Once a suitable comparison system has been constructed, 
there also needs to be a way to relate this system to the original system of ODEs. 

The \emph{comparison principle} emerged as a powerful tool in the qualitative theory of 
dynamical systems \cite{walter2013ordinary,chatterjee2006stability,DBLP:conf/fm/SogokonGTP18}, 
and provides a simple mechanism for relating the solutions of two different differential systems. 
It plays an integral part in the analysis of \emph{partial differential equations} (PDEs) 
\cite{kawohl2000comparison}. In deductive verification the comparison principle 
has been successfully leveraged \cite{DBLP:journals/tocl/Platzer17} to 
prove invariance via Hamilton-Jacobi-Isaacs PDEs for differential hybrid games.
For ODEs, the comparison principle has been widely used for stability 
analysis, generalizing the celebrated Lyapunov's direct method, and certifying invariance 
\cite{bellman1962vector,chatterjee2006stability,nersesov2006stability,DBLP:conf/fm/SogokonGTP18,walter2013ordinary}. 
Next we will go over the scalar and vector comparison principles and axiomatize 
them in \dL.

\subsection{Scalar Comparison}

The scalar comparison principle proves properties about a system of ODEs by relating 
a scalar differential inequality with another scalar differential equation (comparison system). 
To demonstrate the scalar comparison principle, consider an $n$-dimensional system 
of ODEs $x' = \mathbf{f}(x,t)$ whose solution $x(t)$ exists in some compact time interval $I = [0,T]$ 
for some $T > 0$. For a scalar function $h: \R^{n} \to \R$ satisfying the following differential 
inequality, the corresponding one-dimensional comparison system in $z$ can be constructed
\begin{equation}\label{eq:scalar-comparison}
  h' \geq \eta(h(x),t), \qquad z' = \eta(z(t),t)
\end{equation}
where $z$ is a fresh scalar variable and 
$\eta: \R \times I \to \R$ is an appropriately chosen function (the exact criteria are investigated below). The comparison principle 
enables one to use the solutions $h(x)$, i.e., $h(x(t))$ to relate the solutions 
$x(t)$ of the original system to the solution $z(t)$ of the comparison system. This 
sound reasoning is possible due an accompanying comparison theorem.

\begin{definition}[Local Lipschitz {\cite[\S10.IV]{walter2013ordinary}}]\label{def:local-lip}
  A function $\eta(\cdot)$ is said to be \emph{locally Lipschitz} if the following formula 
  is valid, $B_{r}(a) \equiv t \geq 0 \Land \norm{x - a}^{2} < r^{2} \Land \norm{y - a}^{2} < r^{2}$ is an open ball of radius $r$.
  \[\LocLip{\eta} \defequiv \forall \delta \forall a \exists r \exists k \forall t \forall x \forall y (t \leq \delta^{2} \Land B_{r}(a) \to \norm{\eta(x,t) - \eta(y,t)}^{2} \leq k^{2}\norm{x - y}^{2}).\]
\end{definition}

\begin{theorem}[Scalar Comparison Theorem {\cite[\S9.IX]{walter2013ordinary}}]\label{thm:scalar-comparison}
  Let $g$ and $h$ be scalar functions differentiable in $I = [\xi, \xi + \delta] \subseteq \R$ and 
  $\eta: D \to \R$ locally Lipschitz for some real domain $D$. If $h(\xi) \geq g(\xi)$ and $h' - \eta(h,t) \geq g' - \eta(g,t)$ 
  in $I$, then $h \geq g$ in $I$.
\end{theorem}

This theorem is a general version, which relates two differentiable functions 
on some (compact) interval. Following is a corollary based on Theorem~\ref{thm:scalar-comparison} 
that directly relates to the system in (\ref{eq:scalar-comparison}) and is more 
useful in our case.

\begin{corollary}\label{cor:scalar-comparison}
  Let $h$ and $z$ be differentiable scalar functions defined on the interval 
  $I = [0, T] \subset \R$ where $T > 0$. For some locally Lipschitz function $\eta: \R \times I \to \R$, 
  if $h(0) \geq z(0)$, and if $h' \geq \eta(h,t)$ and $z' = \eta(z,t)$ holds in $I$, 
  then $h(t) \geq z(t)$ for all $t \in I$.
\end{corollary}
\begin{proof}
  Since $z' - \eta(z,t) = 0$ and $h' - \eta(h,t) \geq 0$, also 
  $h' - \eta(h,t) \geq z' - \eta(z,t)$ in $I$. The rest follows directly from Theorem~\ref{thm:scalar-comparison}.
\end{proof}

\begin{figure}[htbp]
    \centering
    % first Image
    \begin{subfigure}[b]{0.3\textwidth}
        \centering
        \includegraphics[width=\textwidth]{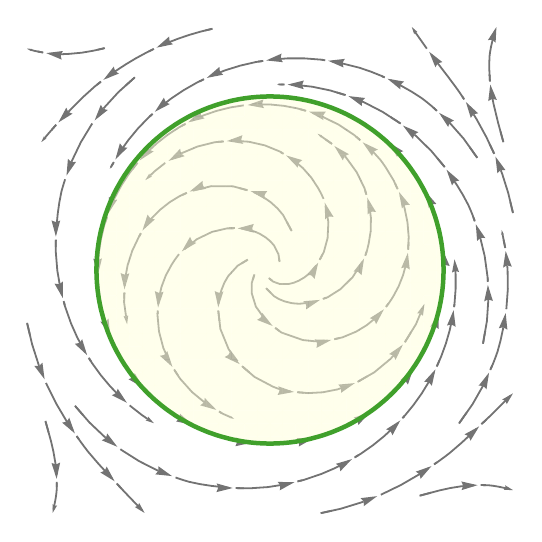}
        \caption{}
        \label{fig:phase-plot}
    \end{subfigure}
    % second Image
    \begin{subfigure}[b]{0.3\textwidth}
        \centering
        \includegraphics[width=\textwidth]{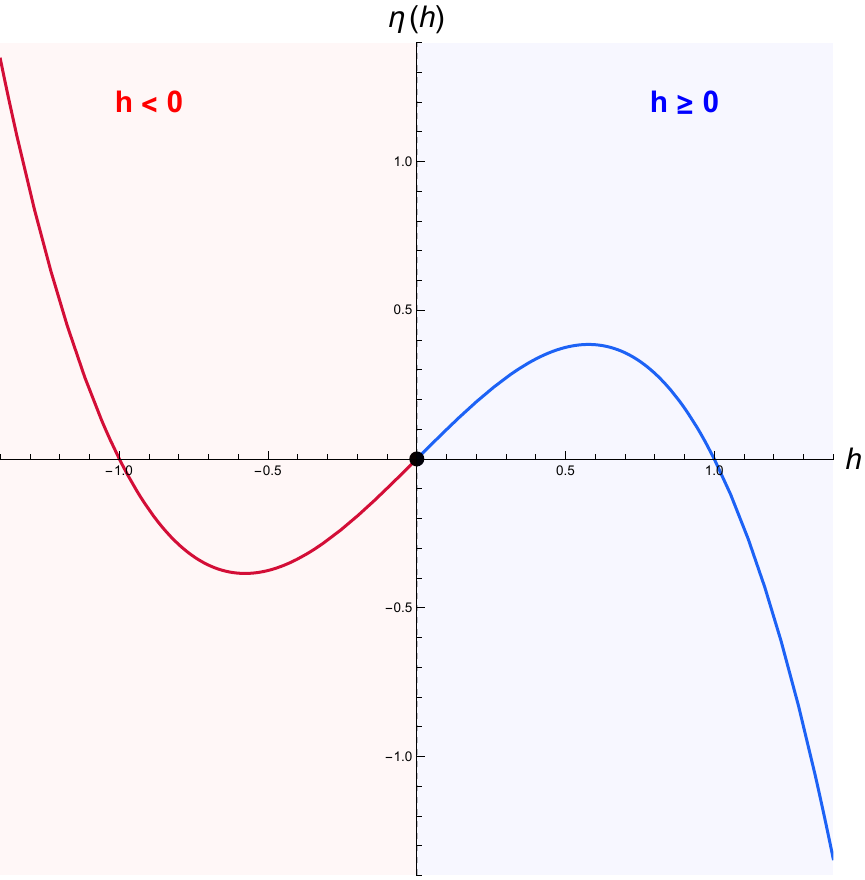}
        \caption{}
        \label{fig:eta-h}
    \end{subfigure}
    % third Image
    \begin{subfigure}[b]{0.3\textwidth}
        \centering
        \includegraphics[width=\textwidth]{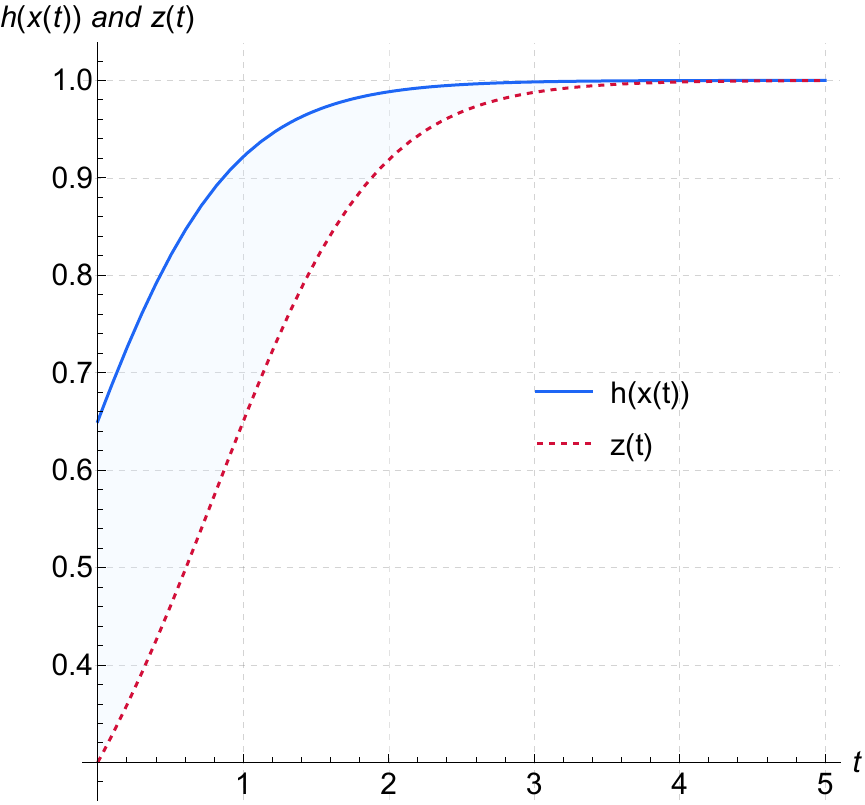}
        \caption{}
        \label{fig:lower-bound}
    \end{subfigure}
    \vspace{-5pt}
    \caption{The phase portrait of the system $\alpha_{e}$ with the invariant 
    $h \geq 0$ from Example~\ref{ex:invariance-comparison} is shown in (a). Figure (b) shows how $\eta(h)$ changes 
    as $h$ changes. The right figure (c) shows that $h(x(t))$ is lower bounded by 
    the solution $z(t)$ of the comparison system $z' = \eta(z)$.}
    \label{fig:comparison}
\end{figure}

An important ingredient of Theorem~\ref{thm:scalar-comparison} and Corollary~\ref{cor:scalar-comparison} 
is that the function $\eta(\cdot)$ must be locally Lipschitz, making sure that there 
are no pathological behaviors in the solution of the differential inequality and 
the comparison system \cite{walter2013ordinary}.
The validity of the formula \LocLip{$\eta$} from Definition~\ref{def:local-lip} 
ensures that the function $\eta(\cdot)$ is locally Lipschitz, 
and can be soundly utilized in \dL to carry out proofs. Moreover, if $\eta(\cdot)$ is a 
polynomial, then $\LocLip{\eta}$ is a first-order logic formula over the reals (FOL$_\R$), and hence 
decidable \cite{bochnak2013real}.

\begin{remark}
  The function $\eta(\cdot)$, which acts as the bridge relating the differential inequality 
  and the comparison system is often called the \emph{comparison function} \cite{kellett2014compendium}. 
  Different comparison functions with different properties can be chosen depending on the proof requirement \cite{kellett2014compendium}.
\end{remark}

\begin{example}[Invariance via Comparison]\label{ex:invariance-comparison}
  Consider the following 2D dynamical system
  \[\alpha_{e} \equiv x_{1}' = -\frac{1}{2}(1 - x_ {1}^{2} - x_ {2}^{2})(x_ {1}^{2} + x_{2}^{2} - 2)x_{1} - x_{2}, \ x_{2}' = -\frac{1}{2}(1 - x_ {1}^{2} - x_ {2}^{2})(x_ {1}^{2} + x_{2}^{2} - 2)x_{2} + x_{1}\]
  here $x = (x_{1}, x_{2})^{\top} \in \R^{2}$. For the system $\alpha_{e}$ and the 
  function $h(x) = 1 - x_{1}^{2} - x_{2}^{2}$, $h(x) \geq 0$ is an invariant, as 
  illustrated\footnote{Figure~\ref{fig:comparison} is only meant for illustration purposes and is not a formal proof. The invariance of $h \geq 0$ can be formally proven in \dL by utilizing axiom \lref{ax:SCP}.} 
  in Figure~\ref{fig:phase-plot}. The invariance of $h \geq 0$ can 
  be proved using the comparison principle from Corollary~\ref{cor:scalar-comparison} 
  by showing $\dot{h} \geq \eta(h)$, where $\eta(h) := h - h^{3}$, holds true. As we 
  can see $\eta(\cdot)$ is a nonlinear locally Lipschitz function, Figure~\ref{fig:eta-h} 
  shows how $\eta(h)$ changes when $h \geq 0$ or $h < 0$. Now we can construct the 
  comparison system $z' = \eta(z)$ whose solutions $z(t)$ lower bounds $h(x(t))$ 
  whenever $h(x(0)) \geq z(0) \geq 0$, making $h \geq 0$ an invariant, as can be 
  seen in Figure~\ref{fig:lower-bound}.
\end{example}

\begin{lemma}[Scalar Comparison Principle]\label{lem:scp}
  The axiom \lref{ax:SCP} internalizes the scalar comparison principle and is derived syntactically in \emph{\dL}, where the term $h$ represents $h(x)$, 
  and $h$, $z$, $\eta(\cdot)$ are scalars.
  \begin{align*}
        \axtag{ax:SCP} \quad & h \geq z \Land \LocLip{\eta} \to \\
    & (\fBox[\odes{x' = \mathbf{f}(x), \tau' = 1}[Q]][\dot{h} \geq \eta(h,\tau)] \to \fBox[\odes{x' = \mathbf{f}(x), \tau' = 1, z' = \eta(z,\tau)}[Q]][h \geq z])
  \end{align*}
\end{lemma}
\begin{proof}[Proof (Sketch)]
  The axiom is proven syntactically in \dL. The full proof can be found in Appendix~\ref{appendix:comparison}. 
  At a high level, the proof utilizes axiom \lref{ax:BDG} to introduce the comparison 
  system $z' = \eta(z,\tau)$ in the antecedent to match the succedent. This 
  enables proving the new invariant $\exp(k\tau)\cdot(h-z) \geq 0$ using 
  axioms \lref{ax:DC} and \lref{ax:DI_inequal}. 
  Where $\exp(\cdot)$ is the real exponential function, $k$ is the Lipschitz constant, and 
  $\tau$ is a variable for tracking time with the associated ODE $\tau' = 1$. This new invariant paired with 
  the fact $\eta(\cdot)$ being locally Lipschitz, i.e., \LocLip{$\eta$}, basically says 
  $\exp(k\tau)\cdot(h-z)$ is monotonically increasing, hence proving the original invariant 
  $h \geq z$.
\end{proof}

Axiom \lref{ax:SCP} in the above lemma internalizes Corollary~\ref{cor:scalar-comparison} 
in \dL, which enables sound utilization of the comparison principle in various invariance 
proofs. At a high level, axiom \lref{ax:SCP} says for all runs of $\odes{x' = \mathbf{f}(x), \tau' = 1, z' = \eta(z,\tau)}[Q]$, 
$h$ is lower-bounded by $z$, if $h \geq z$ was true initially, and for all runs of $\odes{x' = \mathbf{f}(x), \tau' = 1}[Q]$, 
the inequality $\dot{h} \geq \eta(h,\tau)$ holds true, for locally Lipschitz $\eta(\cdot)$.
Next we will generalize to vectorial case.

\subsection{Vector Comparison}

Similar to the scalar comparison principle, the \emph{vector comparison principle} 
relates the vectorial differential inequality $\mathbf{h}' \geq \boldsymbol{\eta}(\mathbf{h})$ 
to the comparison system $z' = \boldsymbol{\eta}(z)$. Here $z = (z_{1}, \dots, z_{n})^{\top}$ is fresh vector of 
variables and $\boldsymbol{\eta}(\cdot)$ is a 
locally Lipschitz vector-valued function. Consequently, one might think Theorem~\ref{thm:scalar-comparison} 
and Corollary~\ref{cor:scalar-comparison} can be na\"{\i}vely carried over to the 
vectorial case to prove $\mathbf{h} \geq z$ with component-wise ordering, i.e., $\bigwedge^{n}_{i=1} h_{i} \geq z_{i}$. Unfortunately, 
this is not true in general and more regularity assumptions are required on $\boldsymbol{\eta}(\cdot)$ \cite[\S10, Supplement I]{walter2013ordinary}. 
In order to better understand why that is the case consider the following trivial counterexample.

\begin{counterexample}
  For the system $x' = \mathbf{f}(x)$, where $x \in \R^{2}$ and $\mathbf{f}(x) = (0, 0)^{\top}$, 
  let $\mathbf{h}(x) = (0, 0)^{\top}$ and $\boldsymbol{\eta}(y_{1}, y_{2}) = (-y_{2}, 0)^{\top}$, 
  where $\mathbf{h}: \R^{2} \to \R^{2}$ and $\boldsymbol{\eta}: \R^{2} \to \R^{2}$. 
  Clearly $\boldsymbol{\eta}(\cdot)$ is a vector of polynomials, hence, locally 
  Lipschitz, also the following inequality holds
  \[\dot{\mathbf{h}}(x) = (0, 0)^{\top}, \quad \boldsymbol{\eta}(\mathbf{h}) = (0, 0)^{\top}, \quad \text{hence,} \quad \dot{\mathbf{h}} \geq \boldsymbol{\eta}(\mathbf{h}).\]
  Now consider the comparison system $z' = \boldsymbol{\eta}(z)$, where $z \in \R^{2}$ is 
  fresh, with the initial condition $z(0) = (0, -1)^{\top}$, this implies 
  $\mathbf{h}(x(0)) \geq z(0)$ holds. Next compute the solutions of the comparison 
  system
  \[z' = \boldsymbol{\eta}(z), \ \text{i.e.,} \ z_{1}' = -z_{2}, \ z_{2}' = 0, \quad \text{whose solutions are,} \quad z(t) = (t, -1)^{\top}.\]
  From this we can clearly see for all $t > 0$ and any initial value $x(0)$ that 
  $\mathbf{h}(x(t)) \ngeq z(t)$, which shows $\boldsymbol{\eta}(\cdot)$ just being 
  locally Lipschitz is not enough.
\end{counterexample}

The missing property is, \emph{if and only if} the locally Lipschitz vector-valued function $\boldsymbol{\eta}(\cdot)$ 
is \emph{quasimonotone increasing} then Theorem~\ref{thm:scalar-comparison} and 
Corollary~\ref{cor:scalar-comparison} can be lifted to the vectorial case \cite[\S10.XII]{walter2013ordinary}. 
This property can be encoded as a logical formula and is defined as follows.

\begin{definition}[Quasimonotone Increasing {\cite[\S10.XII]{walter2013ordinary}}]\label{def:quasi-mon}
  A vector-valued function $\boldsymbol{\eta}(\cdot)$ is called \emph{quasimonotone increasing} 
  if the following formula is true.
  \begin{equation*}
    \QuasiMon{\boldsymbol{\eta}} \defequiv \forall t \forall x \forall y \left(x \geq y \to \bigwedge^{n}_{i=1}\left(x_{i} = y_{i} \to \eta_{i}(x,t) \geq \eta_{i}(y,t)\right)\right)
  \end{equation*}
  Where $x \geq y$ is component-wise ordering i.e., $x_{i} \geq y_{i}$ for $i = 1, \dots, n$.
\end{definition}

This quasimonotonicity fact combined with the local Lipschitz condition can now be 
utilized to lift the scalar comparison theorem to the vectorial case as stated in the following 
theorem.

\begin{theorem}[Vector Comparison Theorem {\cite[\S10.XII]{walter2013ordinary}}]\label{thm:vector-comparison}
  Let $\boldsymbol{\eta}: D \to \R^{n}$ be quasimonotone 
  increasing and locally Lipschitz for some real domain $D$, and let $\mathbf{v}$ 
  and $\mathbf{w}$ be differentiable on $I = [\xi, \xi + \delta] \subseteq \R$. If 
  $\mathbf{v}(\xi) \geq \mathbf{w}(\xi)$ and $\mathbf{v}' - \boldsymbol{\eta}(\mathbf{v},t) \geq \mathbf{w}' - \boldsymbol{\eta}(\mathbf{w},t)$ 
  in $I$, then $\mathbf{v} \geq \mathbf{w}$ in $I$.
\end{theorem}

Similar to the scalar case, the above theorem is a general version for differentiable 
functions on some compact interval. Following is a direct corollary to 
Theorem~\ref{thm:vector-comparison}, which is more suited to relate properties 
of ODEs.

\begin{corollary}\label{cor:vector-comparison}
  For some differentiable vector-valued functions $\mathbf{h}$ and $z$, let $\boldsymbol{\eta}(\cdot)$ 
  be quasimonotone increasing and locally Lipschitz. If $\mathbf{h}' \geq \boldsymbol{\eta}(\mathbf{h},t)$ 
  and $z' = \boldsymbol{\eta}(z,t)$ hold, along with $\mathbf{h}(0) \geq z(0)$ on 
  $I = [0, T] \subset \R$ for some $T > 0$. Then for all $t \in I$ we have $\mathbf{h} \geq z$.
\end{corollary}
\begin{proof}
  $\mathbf{h}' \geq \boldsymbol{\eta}(\mathbf{h},t)$ and $z' = \boldsymbol{\eta}(z,t)$ 
  imply $\mathbf{h}' - \boldsymbol{\eta}(\mathbf{h},t) \geq z' - \boldsymbol{\eta}(z,t)$. 
  Rest by Theorem~\ref{thm:vector-comparison}.
\end{proof}

This observation termed as the \emph{vector comparison principle} is axiomatized 
in \dL, which helps in invariance proofs for vectorial postconditions. The following 
lemma shows the axiom \lref{ax:VCP} that internalizes Corollary~\ref{cor:vector-comparison} 
and is derived syntactically in \dL's proof calculus.

\begin{lemma}[Vector Comparison Principle]\label{lem:vcp}
  The following axiom derives syntactically in \dL and enables reasoning about ODE properties using the vector comparison 
  principle. Where $\mathbf{h}$ represents $\mathbf{h}(x)$, and $\mathbf{h}$, 
  $z$, and $\boldsymbol{\eta}(\cdot)$ are vectors.
  \begin{align*}
    \axtag{ax:VCP} \quad & \mathbf{h} \geq z \Land \QuasiMon{\boldsymbol{\eta}} \Land \LocLip{\boldsymbol{\eta}} \to \\
    & (\fBox[\odes{x' = \mathbf{f}(x), \tau' = 1}[Q]][\dot{\mathbf{h}} \geq \boldsymbol{\eta}(\mathbf{h},\tau)] \to \fBox[\odes{x' = \mathbf{f}(x), \tau' = 1, z' = \boldsymbol{\eta}(z,\tau)}[Q]][\mathbf{h} \geq z])
  \end{align*}
\end{lemma}
\begin{proof}[Proof (Sketch)]
  The proof follows similar to the scalar case by introducing the comparison system $z' = \boldsymbol{\eta}(z,\tau)$ 
  using axiom \lref{ax:BDG}, which enables proving the new invariant $\exp(k\tau)\cdot (\mathbf{h} - z) \geq 0$ 
  using \lref{ax:DC} and \lref{ax:DI_inequal}. But unlike the scalar case, this time it has to be shown 
  that each component of $z$ is upper-bounded by the respective 
  component of $\mathbf{h}$. In order to show this first $\exp(k\tau)\cdot (\mathbf{h} - z) \geq 0$ 
  is equivalently rewritten as $\bigwedge^{n}_{i=1} (\exp(k\tau)\cdot (h_{i} - z_{i}) \geq 0)$, 
  then we construct an intermediate 
  vector $\tilde{z} = (h_{1}, \dots, h_{i-1}, z_{i}, h_{i+1}, \dots, h_{n})^{\top}$ that 
  agrees with every component of $\mathbf{h}$ except for the $i$-th, where it agrees with 
  $z$, this enables showing $\norm{\mathbf{h} - \tilde{z}}^{2} = (h_{i} - z_{i})^{2}$. This paired with the fact that $\boldsymbol{\eta}(\cdot)$ is locally Lipschitz and 
  more importantly quasimionotone increasing enables proving $\exp(k\tau)\cdot (\mathbf{h} - z) \geq 0$.
  The full proof can be found in Appendix~\ref{proof:vcp}.
\end{proof}

\section{Control Invariance}\label{sec:control-invariance}

\begin{wrapfigure}{r}{0.45\textwidth}
  \centering
  \vspace{-25pt}
  \centering
  \includegraphics[width=0.45\textwidth]{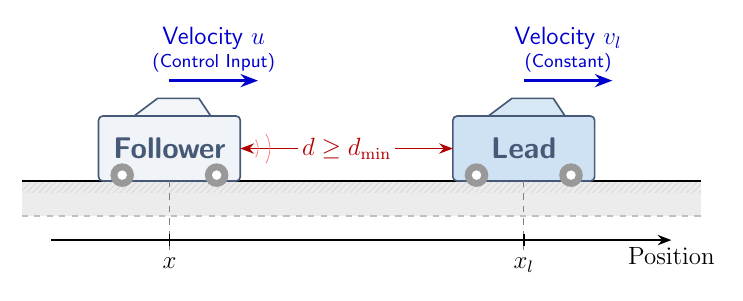}
  \vspace{-25pt}
  \caption{Adaptive cruise control system as modeled in Example~\ref{ex:acc}.}
  \label{fig:acc}
  \vspace{-10pt}
\end{wrapfigure}

With the two comparison principle axioms \lref{ax:SCP} and \lref{ax:VCP} in place 
we can move forward with our axiomatic analysis of control invariance.

\begin{example}[Adaptive Cruise Control (ACC)]\label{ex:acc}
  The following ODE models the dynamics of a vehicle following another vehicle in 
  straight line with the explicit time variable $\tau \in \R$
  \[\alpha_{\text{acc}} \equiv \odes{x' = u, x_{l}' = v_{l}, \tau' = 1}[Q].\]
  Here $x, x_{l} \in \R$ are the positions of the follower and lead vehicle, respectively.
  Variable $u \in \mathcal{U}$ is the 
  control input, where $\mathcal{U} \subseteq \R$ is some compact set, $v_{l}$ 
  is the constant velocity of the lead vehicle, and 
  $Q \equiv v_{l} \geq 0 \Land T \geq \tau$ is the domain constraint. For dynamics 
  $\alpha_{\text{acc}}$ we can model an adaptive cruise control system, as shown in 
  Figure~\ref{fig:acc}, in \dL, which 
  maintains a safe distance from the lead vehicle as follows
  \begin{equation}\label{eq:acc}
    \text{Safe} \to [(\text{ctrl};\alpha_{\text{acc}})^{*}]\text{Safe}
  \end{equation}
  where $\text{Safe} \equiv x_{l} - x \geq d_{\text{min}}$. Taking a closer look, 
  the formula Safe encodes the safety condition that says the distance between the 
  lead vehicle and the follower vehicle must be greater than or equal to the minimum 
  distance $(d_{\text{min}})$. 
  Formula (\ref{eq:acc}) can be proven safe in \dL for 
  the right implementation of the controller (ctrl). But, before one can construct such a 
  controller one would like to prove whether a control input $u \in \mathcal{U}$ 
  even exists that ensures safety, essentially asking the more general question 
  of control invariance, which is formally defined next.
\end{example}

\begin{definition}[Control Invariance]
  Given a control system $x' = \mathbf{f}(x,u)$, where $x \in \R^{n}$ and $u \in \mathcal{U} \subseteq \R^{m}$ 
  with $\mathbf{f}: \R^{n} \times \R^{m} \to \R^{n}$ locally Lipschitz, and some set $S \subseteq \R^{n}$, 
  \emph{control invariance} is the property that ensures there exists some control input $u \in \mathcal{U}$ such that 
  when starting from the set $S$, the solution of the control system remains in the set $S$ for 
  all future times in the domain of existence i.e., $\varphi(t) \in S, \ \forall t \in [0,T]$, for 
  some $T > 0$. The set $S$ is also known as a \emph{control invariant set}. 
\end{definition}

Informally, control invariance of $S$ is a property that the system admits, from
every state in $S$, at least one admissible control input that keeps the system
inside $S$ as long as the solution exists. One can observe that control 
invariance asks a more general question of ``whether a control input even exists'' such 
that a controller can be constructed for enforcing control invariance.
Control invariance properties can be modeled in \dL by existentially quantifying 
over the control input $u$ for the control system $x' = \mathbf{f}(x,u)$ i.e.,
\begin{equation}\label{eq:control-inv-dl}
  S \to \exists u \fBox[\ode[x]{\mathbf{f}(x,u)}[Q]][S].
\end{equation}
The validity of formulas of the form (\ref{eq:control-inv-dl}) proves that there always exist some control 
input that keeps $S$ invariant, implying control invariance. 
This shows that proving 
control invariance is more general than proving a control system stays invariant for a 
concrete controller implementation, like the adaptive cruise control hybrid program 
(\ref{eq:acc}). Now we move on to developing sound axioms and proof rules, 
which enable verifying control invariance in \dL, this is made possible by comparison 
invariants, which will be delved into next.

\subsection{Comparison Invariants}

We now introduce \emph{comparison invariants}, which enable proving control invariance 
by proving a functional inequality on their Lie derivative. Comparison invariants 
are based on the ideas of the comparison principles with an additional regularity 
enforcement on the comparison functions.

\begin{definition}[Class $\mathcal{G}$ Functions {\cite{kellett2014compendium}}]\label{def:class-g}
  A scalar function $\eta(\cdot)$ and a vector-valued function $\boldsymbol{\eta}(\cdot)$ 
  are called (extended) \emph{class $\mathcal{G}$} if the functions are zero at zero and 
  nondecreasing. They are characterized by the following logical formulas.
  \[\G{\eta} \defequiv \forall t \forall x \forall y (\eta(0,t) = 0 \Land (x \geq y \to \eta(x,t) \geq \eta(y,t))) \qquad \vecG{\boldsymbol{\eta}} \defequiv \bigwedge^{n}_{i = 1} \G{\eta_{i}}.\]
\end{definition}

The following simple counterexample demonstrates why extra care needs to be taken when 
handling control systems and the reason comparison invariants require the added 
regularity over just locally Lipschitz comparison functions even for the scalar case.

\begin{counterexample}\label{exp:counter-2}
  Consider the control system $x' = u$, where $x \in \R$ and $u \in \mathcal{U} = [-1,0]$, with the 
  requirement that $h(x) \geq 0$, with $h(x) := x$, stays invariant. In order to 
  show invariance of $h(x) \geq 0$, by the comparison principle from Corollary~\ref{cor:scalar-comparison}, 
  the system needs to satisfy the inequality
  \[\dot{h}(x) - \eta(h,t) \geq 0.\]
  Now choosing $\eta(h,t) := -1$, which is locally Lipschitz, but does not satisfy 
  $\eta(0,t) = 0$ to be a class $\mathcal{G}$ function as per Definition~\ref{def:class-g}. 
  We get the inequality $u + 1 \geq 0$, since $\dot{h}(x) = x' = u$. This 
  shows that every admissible control in $\mathcal{U}$ satisfies this inequality, but 
  when the system starts at $x(0) = 0$, and the controller chooses $u = -1$, we get 
  $x(t) = -t < 0$, meaning $h(x) \geq 0$ is no longer invariant, and the system 
  is unsafe. This demonstrates the need for $\eta(0,t) = 0$, to make sure trajectories 
  cannot escape at the boundary. Now choosing $\eta(h,t) := -h$, which is locally Lipschitz 
  and also satisfies $\eta(0,t) = 0$, but is decreasing and hence is not class $\mathcal{G}$. 
  Now the inequality becomes $u + h \geq 0$, although this condition still prevents 
  an immediate boundary crossing, its dependence on $h$ is reversed: increasing 
  the value of $h(x)$ decreases $\eta(h,t)$ and changes the comparison bound in 
  the opposite direction. Consequently, the ordering of the comparison condition 
  is not consistent with the ordering of the $h(x)$ values. By requiring $\eta(\cdot)$ 
  to be nondecreasing, class $\mathcal{G}$ functions rule out this reversed dependence, 
  which is useful for construction of practical controllers.
\end{counterexample}

\begin{definition}[Comparison Invariant]
  For a control system $x' = \mathbf{f}(x,u)$ with a control admissible set $\mathcal{U} \subseteq \R^{m}$, 
  and a set $S = \{x \in \R^{n} : h \succcurlyeq 0\}$, 
  where $h: \R^{n} \to \R$ is a scalar function and $\succcurlyeq$ is either 
  $\geq$ or $>$, $h$ is said to be a \emph{comparison invariant} if there exists a 
  comparison function $\eta(\cdot)$ that is locally 
  Lipschitz and class $\mathcal{G}$ such that it satisfies the inequality
  \[\forall x \in S. \exists u \in \mathcal{U}.(\Lie{\mathbf{f}(x,u)}{h} + \eta(h,t) \geq 0)\]
  implying control invariance, essentially $h \succcurlyeq 0$ stays invariant, 
  where $t \in [0, T]$ for some $T > 0$.
\end{definition}

A key requirement for comparison invariants is the existence of a comparison function 
that is not just locally Lipschitz but also of class $\mathcal{G}$.
As shown by Counterexample~\ref{exp:counter-2}, this added regularity ensures that the Lie derivative of the comparison invariant 
$h$ does not get too negative, while at the same time allowing the solution of the 
control system to, in a sense, move ``forward.'' The comparison function $\eta(\cdot)$ being 
class $\mathcal{G}$ means whenever $h = 0$, the Lie derivative becomes $\dot{h} \geq 0$. 
Semantically, this invariance argument is corroborated by the comparison principle, 
which ensures $h$ is lower bounded by the solution of the comparison system $z' = \eta(z)$, 
whenever $h(0) \geq z(0) \geq 0$. Next we axiomatize scalar comparison invariants.

\begin{theorem}[Scalar Comparison Invariants]\label{thm:scalar-ci}
  The scalar comparison invariants axiom \lref{ax:SCI} and proof rule \lref{pr:sCI} are 
  derivable in \emph{\dL}. Here $\mathcal{U}$ characterizes a compact 
  control input set. ($\succcurlyeq$ is either $\geq$ or $>$.)
  \begin{equation*}
    \begin{array}{@{}r@{\quad}l@{}}
      \axtag{ax:SCI}
      &
      \LocLip{\eta} \Land \G{\eta} \Land \fBox[\odes{x' = \mathbf{f}(x,u), \tau' = 1}[Q]][\exists u (\mathcal{U} \Land \dot{h} + \eta(h,\tau) \geq 0)] \to \\
    & (h \succcurlyeq 0 \to \exists u\fBox[\odes{x' = \mathbf{f}(x,u), \tau' = 1}[Q]][h \succcurlyeq 0]) \\ [2ex]
      \axtag{pr:sCI}
      &
      \vcenter{\hbox{%
        \AxiomC{$\vdash \LocLip{\eta} \Land \G{\eta}$}
        \AxiomC{$Q \vdash \exists u (\mathcal{U} \Land \dot{h} + \eta(h,\tau) \geq 0)$}
        \BinaryInfC{$h \succcurlyeq 0 \vdash \exists u\fBox[\odes{x' = \mathbf{f}(x,u), \tau' = 1}[Q]][h \succcurlyeq 0]$}
        \DisplayProof
      }}
    \end{array}
  \end{equation*}
\end{theorem}
\begin{proof}[Proof (Sketch)]
  The axiom \lref{ax:SCI} derives syntactically in \dL's proof calculus, the full proof 
  can be found in Appendix~\ref{appendix:comparison-invariants}. The proof utilizes  
  axiom \lref{ax:BDG} to introduce the comparison system $z' = \eta(z,\tau)$, 
  where $z$ is a fresh scalar variable and $\eta(\cdot)$ is a locally Lipschitz class $\mathcal{G}$ 
  function. Then the scalar comparison principle axiom \lref{ax:SCP} finishes the proof. 
  Proof rule \lref{pr:sCI} is also derived syntactically as follows.
  \begin{sequentproof}
    \ax[]{}{\LocLip{\eta} \Land \G{\eta}}
    \ax[]{Q}{\exists u (\mathcal{U} \Land \dot{h} + \eta(h,\tau) \geq 0)}
    \bi[andR,dW]{}{\LocLip{\eta} \Land \G{\eta} \Land \fBox[\odes{x' = \mathbf{f}(x,u), \tau' = 1}[Q]][\exists u (\mathcal{U} \Land \dot{h} + \eta(h,\tau) \geq 0)]}
    \un[ax:SCI]{h \succcurlyeq 0}{\exists u\fBox[\odes{x' = \mathbf{f}(x,u), \tau' = 1}[Q]][h \succcurlyeq 0]}
  \end{sequentproof}
  The proof utilizes axiom \lref{ax:SCI}, followed by \lref{andR} and \lref{dW} resulting 
  in the required premises.
\end{proof}

From Theorem~\ref{thm:scalar-ci}, proof rule \lref{pr:sCI} enables proving control 
invariance by reducing the box modality to two open premises with purely real arithmetic 
formulas. Importantly, when the comparison function $\eta(\cdot)$ is a polynomial 
the two resulting premises are formulas of first-order logic over the reals (FOL$_\R$) 
and are, hence, decidable by \emph{quantifier elimination} \cite{bochnak2013real}.

\begin{example}[ACC Revisited]
  Now we take another look at the adaptive cruise control system from Example~\ref{ex:acc}. 
  Recall $\alpha_{\text{acc}} \equiv \odes{x' = u, x_{l}' = v_{l}, \tau' = 1}[Q]$ with 
  $Q \equiv v_{l} \geq 0 \Land T \geq \tau$. Also recall, $\text{Safe} \equiv x_{l} - x \geq d_{\text{min}}$, 
  we can encode this safety condition as $h := x_{l} - x - d_{\text{min}}$. 
  Now we would like to prove control invariance for this system by proving the 
  following \dL formula
  \[h \geq 0 \to \exists u[\alpha_{\text{acc}}]h \geq 0.\]
  From our previous developments we know we can prove such formulas using the scalar comparison 
  invariants proof rule \lref{pr:sCI}:
  \begin{sequentproof}
    \ax[]{}{\LocLip{\eta} \Land \G{\eta}}
    \ax[]{Q}{\exists u(\mathcal{U} \Land \dot{h} + \eta(h,\tau) \geq 0)}
    \bi[pr:sCI]{h \geq 0}{\exists u\fBox[\odes{x' = u, x_{l}' = v_{l}, \tau' = 1}[Q]][h \geq 0]}
  \end{sequentproof}
  The Lie derivative of $h$ is computed as $\dot{h} = (x_{l})' - (x)' - (d_{\text{min}})' = v_{l} - u$, 
  observe $d_{\text{min}}$ is a constant. Now in order 
  to complete the proof we need to choose a class $\mathcal{G}$ function $\eta(\cdot)$, 
  a sufficient choice is the linear function $\eta(h,\tau) := \lambda h$, for some 
  positive $\lambda$. This function is a polynomial, hence locally Lipschitz, and satisfies the conditions 
  of class $\mathcal{G}$ from Definition~\ref{def:class-g}.
  \begin{sequentproof}
    \ax*[Real]{\lambda > 0}{\LocLip{\lambda h} \Land \G{\lambda h}}
    \un[cut,exL]{}{\LocLip{\eta} \Land \G{\eta}}
    \ax[]{d_{\text{min}} > 0, \lambda > 0, Q}{\exists u(\mathcal{U} \Land \dot{h} + \lambda h \geq 0)}
    \un[cut,exL]{Q}{\exists u(\mathcal{U} \Land \dot{h} + \eta(h,\tau) \geq 0)}
    \bi[pr:sCI]{h \geq 0}{\exists u\fBox[\odes{x' = u, x_{l}' = v_{l}, \tau' = 1}[Q]][h \geq 0]}
  \end{sequentproof}
  Now continuing with the right branch with the unfolded definition of $Q$, this 
  resulting branch only involves first-order logic formulas over the reals, and hence 
  is decidable \cite{bochnak2013real}. Depending on the specification of the control 
  constraint set $\mathcal{U}$ the branch can be closed by real arithmetic \lref{Real}, 
  proving control invariance.
  \begin{sequentproof}
    \ax*[Real]{d_{\text{min}} > 0, \lambda > 0, v_{l} \geq 0, T \geq \tau \geq 0}{\exists u(\mathcal{U} \Land v_{l} - u + \lambda h \geq 0)}
    \un[]{d_{\text{min}} > 0, \lambda > 0, v_{l} \geq 0, T \geq \tau \geq 0}{\exists u(\mathcal{U} \Land \dot{h} + \lambda h \geq 0)}
  \end{sequentproof}
  For the control constraint set $\mathcal{U} \equiv u_{max} \geq u \Land u \geq 0$, 
  control invariance of the system can be proven, as $u = 0$ is a valid control input 
  that maintains the safe distance, moreover, the upper bound $u_{max}$ can be computed 
  from the Lie derivative as $v_{l} + \lambda h \geq u$.
\end{example}

Now, similar to how the scalar comparison principle was lifted to the vectorial case, 
scalar comparison invariants can also be lifted to the so-called \emph{vector comparison invariants}.

\begin{definition}[Vector Comparison Invariants]
  For a control system $x' = \mathbf{f}(x,u)$ with a control admissible set $\mathcal{U} \subseteq \R^{m}$, 
  and a set $S = \{x \in \R^{n} : \mathbf{h} \succcurlyeq 0\}$, 
  where $\mathbf{h}: \R^{n} \to \R^{n}$ is a vector-valued function and $\succcurlyeq$ is either 
  $\geq$ or $>$ with component-wise ordering i.e., $h_{i} \succcurlyeq 0$ for $i = 1, \dots, n$, 
  $\mathbf{h}$ is said to be a comparison invariant if there exists a 
  vector-valued comparison function $\boldsymbol{\eta}(\cdot)$ that is locally 
  Lipschitz and class $\mathcal{G}$ such that it satisfies the following inequality
  \[\forall x \in S. \exists u \in \mathcal{U}.(\Lie{\mathbf{f}(x,u)}{\mathbf{h}} + \boldsymbol{\eta}(\mathbf{h},t) \geq 0)\]
  that implies control invariance, hence $\mathbf{h} \succcurlyeq 0$ stays invariant, 
  where $t \in [0, T]$ for some $T > 0$.
\end{definition}

As we can see from the above definition, vector comparison invariants lift scalar 
comparison invariants to $n$-dimensional vectors with component-wise ordering of the 
inequalities. Consequently, scalar comparison invariants are a special case of vector 
comparison invariants when $n = 1$. Next we axiomatize vector comparison invariants 
for vectorial control invariance proofs.

\begin{theorem}[Vector Comparison Invariants]\label{thm:vector-ci}
  The vector comparison invariants axiom \lref{ax:VCI} and proof rule \lref{pr:vCI} 
  derive syntactically in \emph{\dL}. The compact control input set is represented by 
  the formula $\mathcal{U}$. ($\succcurlyeq$ is either $\geq$ or $>$.)
  \begin{equation*}
    \begin{array}{@{}r@{\quad}l@{}}
      \axtag{ax:VCI}
      &
      \LocLip{\boldsymbol{\eta}} \Land \vecG{\boldsymbol{\eta}} \Land \fBox[\odes{x' = \mathbf{f}(x,u), \tau' = 1}[Q]][\exists u(\mathcal{U} \Land \dot{\mathbf{h}} + \boldsymbol{\eta}(\mathbf{h},\tau) \geq 0)] \to \\
      & (\mathbf{h} \succcurlyeq 0 \to \exists u\fBox[\odes{x' = \mathbf{f}(x,u), \tau' = 1}[Q]][\mathbf{h} \succcurlyeq 0]) \\ [2ex]
      \axtag{pr:vCI}
      &
      \vcenter{\hbox{%
        \AxiomC{$\vdash \LocLip{\boldsymbol{\eta}} \Land \vecG{\boldsymbol{\eta}}$}
        \AxiomC{$Q \vdash \exists u(\mathcal{U} \Land \dot{\mathbf{h}} + \boldsymbol{\eta}(\mathbf{h},\tau) \geq 0)$}
        \BinaryInfC{$\mathbf{h} \succcurlyeq 0 \vdash \exists u\fBox[\odes{x' = \mathbf{f}(x,u), \tau' = 1}[Q]][\mathbf{h} \succcurlyeq 0]$}
        \DisplayProof
      }}
    \end{array}
  \end{equation*}
\end{theorem}
\begin{proof}[Proof (Sketch)] 
  Similar to the proofs of the previous axioms, the proof of axiom \lref{ax:VCI} uses 
  axiom \lref{ax:BDG} to introduce the comparison system $z' = \boldsymbol{\eta}(z,\tau)$, 
  where $z$ is a fresh vector of variables, 
  then the proof is finished by axiom \lref{ax:VCP}. Similarly, the proof rule 
  with two open premises is derived as follows. The full proof can be found in Appendix~\ref{proof:vci}.
  \begin{sequentproof}
    \ax[]{}{\LocLip{\boldsymbol{\eta}} \Land \vecG{\boldsymbol{\eta}}}
    \ax[]{Q}{\exists u(\mathcal{U} \Land \dot{\mathbf{h}} + \boldsymbol{\eta}(\mathbf{h},\tau) \geq 0)}
    \bi[andR,dW]{}{\LocLip{\boldsymbol{\eta}} \Land \vecG{\boldsymbol{\eta}} \Land \fBox[\odes{x' = \mathbf{f}(x,u), \tau' = 1}[Q]][\exists u(\mathcal{U} \Land \dot{\mathbf{h}} + \boldsymbol{\eta}(\mathbf{h},\tau) \geq 0)]}
    \un[ax:VCI]{\mathbf{h} \succcurlyeq 0}{\exists u\fBox[\odes{x' = \mathbf{f}(x,u), \tau' = 1}[Q]][\mathbf{h} \succcurlyeq 0]}
  \end{sequentproof}
  The required premises are derived using \lref{ax:VCI}, \lref{andR}, and \lref{dW}.
\end{proof}

With this we successfully axiomatized \emph{scalar} and \emph{vector comparison invariants} in 
\dL.

Comparison invariants are an important and general class of invariants, which will 
be further showcased in their use for axiomatizing control barrier functions next.

\subsection{Control Barrier Functions}

First introduced by Ames \emph{et al.} \cite{ames2014control,ames2016control,ames2019control}, 
control barrier functions (CBFs) have risen to popularity in the past decade as a 
go-to method for synthesizing safe controllers. CBF-based quadratic programming 
as an optimal control method provides an efficient way for enforcing safety with 
formal guarantees \cite{ames2019control}.
However, these formal safety guarantees 
are only enforced by the synthesized controller if the CBF used to synthesize it 
was a valid CBF in the first place, and verifying which is a nontrivial task \cite{clark2024semi}.

\begin{definition}[Class $\mathcal{K}$ Functions {\cite{kellett2014compendium}}]\label{def:class-k}
  A function is called an (extended) class $\mathcal{K}$ function 
  if it is zero at zero and strictly increasing.
  Class $\mathcal{K}$ functions are encoded as the following.
  \[\K{\eta} \defequiv \forall t \forall x \forall y(\eta(0,t) = 0 \Land (x > y \to \eta(x,t) > \eta(y,t))) \qquad \vecK{\boldsymbol{\eta}} \defequiv \bigwedge^{n}_{i = 1} \K{\eta_{i}}\]
  Where $\eta(\cdot)$ and $\boldsymbol{\eta}(\cdot)$ are scalar and vector-valued functions, respectively.
\end{definition}

Observe that class $\mathcal{K}$ functions have a stronger 
requirement of strictly increasing as compared to class $\mathcal{G}$, which are 
nondecreasing, meaning every class $\mathcal{K}$ function is also class $\mathcal{G}$, 
but not the other way around.

\begin{definition}[Control Barrier Function \cite{ames2016control}]\label{def:cbf}
  For a sublevel set $S = \{x \in \R^{n} : h(x) \geq 0\}$, where $h: \R^{n} \to \R$ is 
  a continuously differentiable scalar function, $h(x)$ is called a \emph{control barrier function} 
  for the control system $x' = \mathbf{f}(x,u)$ with the control admissible set $\mathcal{U}$ 
  if there exists a class $\mathcal{K}$ function $\gamma(\cdot)$ 
  such that the following inequality holds:
  \[\forall x \in S. \exists u \in \mathcal{U}. (\Lie{\mathbf{f}(x,u)}{h(x)} + \gamma(h(x)) \geq 0)\]
  implying $h(x) \geq 0$ stays invariant along $x' = \mathbf{f}(x,u)$.
\end{definition}

\begin{remark}
  Definition~\ref{def:cbf} gives the standard definition of CBFs as found in the 
  literature \cite{ames2014control,ames2016control,ames2019control,clark2024semi}, 
  where the class $\mathcal{K}$ function $\gamma(\cdot)$ is the special time-independent 
  case of Definition~\ref{def:class-k}. In order to reason about CBFs in \dL, however, 
  we require the use of Definition~\ref{def:class-k} with the special time-tracking 
  variable $\tau$ in order to facilitate proofs, as previously mentioned in Remark~\ref{rem:time}.
\end{remark}

We can already see that CBFs are essentially (scalar) comparison invariants with 
the added regularity that the comparison function $\eta(\cdot)$ is class $\mathcal{K}$, 
rather than the more general class $\mathcal{G}$. The key distinction is that 
class $\mathcal{K}$ functions are strictly increasing, so the allowable system 
behavior changes predictably with the CBF value. In practice, this gives CBFs 
additional structure for balancing safety with other control objectives, such as 
stability, tracking, or performance.

\begin{corollary}[Control Barrier Functions]\label{cor:cbf}
  The following axiom and proof rule derive syntactically in \emph{\dL} and can be 
  utilized to prove the validity of scalar control barrier functions. Here $\mathcal{U}$ 
  characterizes a compact set.
  \begin{equation*}
    \begin{array}{@{}r@{\quad}l@{}}
      \axtag{ax:CBF}
      &
      \LocLip{\eta} \Land \K{\eta} \Land \fBox[\odes{x' = \mathbf{f}(x,u), \tau' = 1}[Q]][\exists u (\mathcal{U} \Land \dot{h} + \eta(h,\tau) \geq 0)] \to \\
    & (h \geq 0 \to \exists u\fBox[\odes{x' = \mathbf{f}(x,u), \tau' = 1}[Q]][h \geq 0]) \\ [2ex]
      \axtag{pr:cbf}
      &
      \vcenter{\hbox{%
        \AxiomC{$\vdash \LocLip{\eta} \Land \K{\eta}$}
        \AxiomC{$Q \vdash \exists u (\mathcal{U} \Land \dot{h} + \eta(h,\tau) \geq 0)$}
        \BinaryInfC{$h \geq 0 \vdash \exists u\fBox[\odes{x' = \mathbf{f}(x,u), \tau' = 1}[Q]][h \geq 0]$}
        \DisplayProof
      }}
    \end{array}
  \end{equation*}
\end{corollary}
\begin{proof}[Proof (Sketch)]
  Axiom \lref{ax:CBF} is proven syntactically in \dL using \lref{ax:SCI}. The full proof 
  can be found in Appendix~\ref{appendix:cbf}. The proof rule \lref{pr:cbf} is derived as follows.
  \begin{sequentproof}
    \ax[]{}{\LocLip{\eta} \Land \K{\eta}}
    \ax[]{Q}{\exists u (\mathcal{U} \Land \dot{h} + \eta(h,\tau) \geq 0)}
    \bi[andR,dW]{}{\LocLip{\eta} \Land \K{\eta} \Land \fBox[\ode{x' = \mathbf{f}(x,u), \tau' = 1}[Q]][\exists u (\mathcal{U} \Land \dot{h} + \eta(h,\tau) \geq 0)]}
    \un[ax:CBF]{h \geq 0}{\exists u\fBox[\odes{x' = \mathbf{f}(x,u), \tau' = 1}[Q]][h \geq 0]}
  \end{sequentproof}
  The required two open premises are derived using axiom \lref{ax:CBF}, followed 
  by \lref{andR} and \lref{dW}.
\end{proof}

The above corollary shows how CBFs can be axiomatized in \dL, and happen to be a special 
instance of comparison invariants. The resulting axiom 
\lref{ax:CBF} and proof rule \lref{pr:cbf} enable proving control invariance 
using a valid CBF and conversely verify the validity of candidate CBFs.

As per Definition~\ref{def:cbf}, standard CBFs are scalar functions, in many instances, 
however, these are not enough to enforce safety, especially when multiple safety 
constraints need to be satisfied at the same time. To this end, 
the so-called \emph{vector control barrier functions} extend the standard scalar 
CBFs to vector-valued functions \cite{allibhoy2023control,ren2023vector}. Similar to 
scalar CBFs, these vector CBFs enforce safety by ensuring $\mathbf{h}(x) \geq 0$, 
where the inequality is component-wise i.e., $h_{i}(x) \geq 0$ for $i = 1, \dots, n$, 
stays invariant. The following corollary axiomatizes vector CBFs in \dL as a special 
instance of vector comparison invariants.

\begin{corollary}[Vector Control Barrier Functions]\label{cor:vcbf}
  The following axiom and proof rule can be used to verify vector control barrier 
  functions, and are derivable in \emph{\dL}. Here the formula $\mathcal{U}$ characterizes 
  a compact set.
  \begin{equation*}
    \begin{array}{@{}r@{\quad}l@{}}
      \axtag{ax:VCBF}
      &
      \LocLip{\boldsymbol{\eta}} \Land \vecK{\boldsymbol{\eta}} \Land \fBox[\odes{x' = \mathbf{f}(x,u), \tau' = 1}[Q]][\exists u(\mathcal{U} \Land \dot{\mathbf{h}} + \boldsymbol{\eta}(\mathbf{h},\tau) \geq 0)] \to \\
      & (\mathbf{h} \geq 0 \to \exists u\fBox[\odes{x' = \mathbf{f}(x,u), \tau' = 1}[Q]][\mathbf{h} \geq 0]) \\ [2ex]
      \axtag{pr:vcbf}
      &
      \vcenter{\hbox{%
        \AxiomC{$\vdash \LocLip{\boldsymbol{\eta}} \Land \vecK{\boldsymbol{\eta}}$}
        \AxiomC{$Q \vdash \exists u(\mathcal{U} \Land \dot{\mathbf{h}} + \boldsymbol{\eta}(\mathbf{h},\tau) \geq 0)$}
        \BinaryInfC{$\mathbf{h} \geq 0 \vdash \exists u\fBox[\odes{x' = \mathbf{f}(x,u), \tau' = 1}[Q]][\mathbf{h} \geq 0]$}
        \DisplayProof
      }}
    \end{array}
  \end{equation*}
\end{corollary}
\begin{proof}[Proof (Sketch)]
  Both the axiom \lref{ax:VCBF} and proof rule \lref{pr:vcbf} are derived syntactically 
  in \dL. In particular, the proof of axiom \lref{ax:VCBF} utilizes the vector 
  comparison invariants axiom \lref{ax:VCI}. The complete proof can be found in 
  Appendix~\ref{proof:vcbf}. The proof rule derives from axiom \lref{ax:VCBF} as 
  follows.
  \begin{sequentproof}
    \ax[]{}{\LocLip{\boldsymbol{\eta}} \Land \vecK{\boldsymbol{\eta}}}
    \ax[]{Q}{\exists u(\mathcal{U} \Land \dot{\mathbf{h}} + \boldsymbol{\eta}(\mathbf{h},\tau))}
    \bi[andR,dW]{}{\LocLip{\boldsymbol{\eta}} \Land \vecK{\boldsymbol{\eta}} \Land \fBox[\odes{x' = \mathbf{f}(x,u), \tau' = 1}[Q]][\exists u(\mathcal{U} \Land \dot{\mathbf{h}} + \boldsymbol{\eta}(\mathbf{h},\tau))]}
    \un[ax:VCBF]{\mathbf{h} \geq 0}{\exists u \fBox[\odes{x' = \mathbf{f}(x,u), \tau' = 1}[Q]][\mathbf{h} \geq 0]}
  \end{sequentproof}
\end{proof}

This axiomatization shows one can prove \emph{control invariance} properties in \dL 
using a valid control barrier function with the right choice of comparison function. 
Consequently, it is also possible to verify a candidate control barrier function in 
\dL for a particular control system. In the next section we will see more special 
instances of comparison invariants.

\section{Special Instances}

We have already seen in the previous section how the axiomatization of control barrier 
functions, both scalar and vector, was made possible using comparison invariants. 
This section gives some more special cases of comparison invariants including \emph{Darboux 
invariants} and \emph{differential invariants}, further bolstering their versatility.

\subsection{Darboux Invariants}

Here we look at an important class of invariants known as \emph{Darboux invariants} 
as they are inspired by Darboux polynomials \cite[Section~3]{DBLP:journals/jacm/PlatzerT20}. 
At a high-level any polynomial $p$ that satisfies the polynomial inequality $\dot{p} \geq gp$ 
for some cofactor polynomial $g$ ensures $p \succcurlyeq 0$ stays invariant. 

Semantically this invariance property is ensured by Gr\"{o}nwall's lemma \cite[\S29.VI]{DBLP:journals/jacm/PlatzerT20,walter2013ordinary}. 
It will be shown that the same invariance observation can be drawn using the comparison 
principle, specifically, comparison invariants. Moreover, the two existing \dL axioms 
\lref{ax:DBX_inequal} and \lref{ax:VDBX} internalizing Darboux invariants are derived 
syntactically using comparison invariants.

\begin{corollary}[Scalar Darboux {\cite[Lemma~3.1]{DBLP:journals/jacm/PlatzerT20}}]\label{cor:scalar-darboux}
  The existing \emph{\dL} axiom \lref{ax:DBX_inequal} for scalar Darboux invariants 
  is derivable using axiom \lref{ax:SCI}. ($\succcurlyeq$ is either $\geq$ or $>$.)
  \[\axtag{ax:DBX_inequal} \quad \fBox[\odes{x' = \mathbf{f}(x,u), \tau' = 1}[Q]][(p)' \geq gp] \to (p \succcurlyeq 0 \to \exists u\fBox[\odes{x' = \mathbf{f}(x,u), \tau' = 1}[Q]][p \succcurlyeq 0])\]
\end{corollary}
\begin{proof}[Proof (Sketch)]
  The axiom is derived when the comparison function is defined as $\eta(p,\tau) = g \cdot p$ 
  for some cofactor polynomial $g = g(x,\tau)$. The full derivation can be found in Appendix~\ref{appendix:special-instances}.
\end{proof}

The corollary above showed the derivation of scalar Darboux invariants using scalar 
comparison invariants. Similarly, vectorial Darboux invariants can also be derived 
using vector comparison invariants as shown in the following corollary.

\begin{corollary}[Vectorial Darboux {\cite[Lemma~4.1]{DBLP:journals/jacm/PlatzerT20}}]\label{cor:vector-darboux}
  The existing \emph{\dL} axiom \lref{ax:VDBX} for vectorial Darboux 
  invariants is derivable using axiom \lref{ax:VCI}.
  \[\axtag{ax:VDBX} \quad \fBox[\odes{x' = \mathbf{f}(x,u), \tau' = 1}[Q]][(\mathbf{p})' = G\mathbf{p}] \to (\mathbf{p} = 0 \to \exists u\fBox[\odes{x' = \mathbf{f}(x,u), \tau' = 1}[Q]][\mathbf{p} = 0])\]
\end{corollary}
\begin{proof}[Proof (Sketch)]
  The axiom is derived by first rewriting $\mathbf{p} = 0 \Liff -\mathbf{p} \geq 0 \Land \mathbf{p} \geq 0$, 
  after which axiom \lref{ax:VCI} is applied with comparison functions, $\boldsymbol{\eta}(\mathbf{p},\tau) = G\cdot\mathbf{p}$. 
  Where $G$ is an essentially positive matrix or \emph{Metzler matrix} \cite{DBLP:conf/fm/SogokonGTP18,walter2013ordinary}. 
  The full proof can be found in Appendix~\ref{appendix:special-instances}.
\end{proof}

As we saw from Corollary~\ref{cor:scalar-darboux} and Corollary~\ref{cor:vector-darboux}, 
comparison invariants generalize Darboux invariants. Next we will see another 
special instance of comparison invariants, namely, differential invariants.

\subsection{Differential Invariants}

Another simple consequence of (scalar) comparison invariants are differential invariants. 
Differential invariants relate invariance of scalar inequalities of the form $p \succcurlyeq 0$ to 
their Lie derivative inequality $\dot{p} \geq 0$ over the entire domain $Q$. The following 
shows the derivation of the differential invariants proof rule \lref{dI_inequal} 
tailored to control systems.

\begin{corollary}[Differential Invariants \cite{Platzer18}]\label{cor:dI}
  The existing differential invariants proof rule \lref{dI_inequal} is a special 
  instance of comparison invariants and is derivable using rule \lref{pr:sCI}. 
  ($\succcurlyeq$ is either $\geq$ or $>$.)
  \begin{prooftree}
    \AxiomC{$Q \vdash \dot{p} \geq 0$}
    \LeftLabel{\axtag{dI_inequal} }
    \UnaryInfC{$p \succcurlyeq 0 \vdash \exists u \fBox[\odes{x' = \mathbf{f}(x,u), \tau' = 1}[Q]][p \succcurlyeq 0]$}
  \end{prooftree}
\end{corollary}
\begin{proof}[Proof (Sketch)]
  The proof is straight forward when the comparison function is the zero function, 
  i.e., $\eta(p,\tau) = 0 \cdot p$. The full derivation can be found in Appendix~\ref{appendix:special-instances}.
\end{proof}

These special instances showcase the versatility of comparison invariants as the 
central unifying idea behind various proof techniques for invariance verification. 
It also shows that the chosen comparison systems in the comparison principle axioms 
\lref{ax:SCP} and \lref{ax:VCP} are very general, and can relate various invariance 
properties with the right choice of the scalar and vector comparison functions $\eta(\cdot)$ and $\boldsymbol{\eta}(\cdot)$, 
respectively.

\section{Related Work}
\emph{Invariance in Dynamical Systems.}
Checking whether a given set is invariant under the flow of an 
ODE has been a central interest in the field of 
dynamical systems \cite{blanchini2008set,walter2013ordinary,mejstrik2012some,sastry2013nonlinear},  
going all the way back to the theorem of Nagumo, which 
gives conditions for proving (positive) invariance of closed sets \cite{mejstrik2012some,blanchini2008set}. 
The comparison principle, a prominent technique for relating solutions of two 
differential equations using differential inequalities, provides an alternative 
method for proving invariance of both open and closed sets \cite{walter2013ordinary}. 
More (relatively) recently barrier certificates, first introduced by Prajna \emph{et al.} 
\cite{prajna2004safety}, gained popularity as a Lyapunov-like technique for verifying 
invariance \cite{DBLP:conf/cdc/PanjaP25}. Vector barrier certificates introduced by Sogokon \emph{et al.} \cite{DBLP:conf/fm/SogokonGTP18}, 
generalizing scalar barrier certificates, primarily utilized the comparison 
principle as the basis.

\emph{Control Invariance.}
A similar question of whether a set is invariant under the flow of a control system 
has also been a central concern in control theory \cite{blanchini2008set,chicone2006ordinary,sastry2013nonlinear}. 
In the field of reachability analysis geometric set representations 
like zonotopes and ellipsoids have been utilized to characterize control invariant 
sets \cite{schafer2023scalable,blanchini2008set}. First introduced by Ames \emph{et al.} 
\cite{ames2014control}, control barrier functions 
have gained immense popularity in the past decade for being an efficient means for 
synthesizing safety enforcing controllers \cite{ames2016control,ames2019control,xiao2021high}. 
Recently, Clark \cite{clark2024semi} gave a semialgebraic framework 
using semidefinite programming for verifying control barrier functions.

\emph{Deductive Verification.}
Unlike the numerical methods \cite{schafer2023scalable,blanchini2008set,clark2024semi}, 
deductive verification methods utilize a symbolic and logical approach 
to proving invariance of dynamical systems. To this end Taly \emph{et al.} \cite{taly2009deductive}, 
and Ghorbal \emph{et al.} \cite{DBLP:journals/cl/GhorbalSP17} have shown ODE invariance 
can be proven using a hierarchy of sound but incomplete proof rules. In the logic 
\dL, completeness of algebraic and semialgebraic invariance has been shown by 
Platzer and Tan \cite{DBLP:conf/lics/PlatzerT18}, with later extensions to Noetherian functions \cite{DBLP:journals/jacm/PlatzerT20}. 
Synthesis of control invariants, specifically control envelopes, have been successfully 
undertaken by Kabra \emph{et al.} \cite{kabra2024cesar}, using differential game 
logic (\textsf{dGL}). Recently Hellwig \emph{et al.} \cite{hellwig2025zonotopes}, 
showed the verification of zonotope-based control envelopes in \dL. Control invariance 
has also been studied by Platzer \cite{DBLP:journals/tocl/Platzer17} 
from the lens of differential hybrid games.

\section{Conclusion}
In this paper we presented an axiomatization for verifying control invariance in 
differential dynamic logic (\dL). This was made possible by first axiomatizing the \emph{scalar} and 
\emph{vector comparison principles}, which served as the central proof principles 
underlying our proofs. We introduced \emph{scalar} 
and \emph{vector comparison invariants} leading to a dedicated set of axioms and 
proof rules that reduce proving control invariance to proving a functional inequality 
on its Lie derivative. Finally, we axiomatized control barrier functions, which 
along with Darboux and differential invariants were shown to be special instances of comparison invariants. 

For future work we will be implementing the presented axioms and proof rules 
in the theorem prover \keymaera \cite{DBLP:conf/cade/FultonMQVP15}. Since 
all the presented axioms and proof rules were derived syntactically in \dL, they 
can be implemented as tactics on top of a soundness-critical prover kernel \cite{DBLP:conf/cade/FultonMQVP15}. 
We will further be looking at deriving completeness results for proving control invariance 
in \dL.

\paragraph{Acknowledgements.} We thank the anonymous reviewers for their helpful 
feedback. We also thank Adrian Kulmburg for his feedback on an earlier version of this work 
including his suggestion of utilizing the fact that locally Lipschitz functions are 
bounded on a compact domain in the proofs. This work was funded by the Deutsche 
Forschungsgemeinschaft (DFG, German Research Foundation) – SFB 1608 – 501798263.

%%
%% Bibliography
%%

%% Please use bibtex, 

\bibliographystyle{plainurl}
\bibliography{reference}

\ifprintappendix

\appendix

\section{\dL Proof Calculus}

The lemmas listed in this Appendix give a brief list of \dL's sound axioms and proof rules for 
reference. We primarily look at the axioms and proof rules 
for reasoning about differential equations. The complete proof calculus 
for hybrid systems along with its soundness proofs can be found elsewhere \cite{DBLP:conf/lics/Platzer12b,DBLP:journals/jar/Platzer17,Platzer18,DBLP:journals/jacm/PlatzerT20}.

\subsection{Base Calculus}\label{appendix:base-calculus}

The following lemma states the base axioms and proof rules of \dL, which are standard 
for modal and dynamic logics.

\begin{lemma}[Base Axioms and Proof Rules]\label{lemma:base-calculus}
  \begin{align*}
    \axtag{ax:K} \quad & \lquote{ax:K} && \axtag{ax:V} \quad \lquote{ax:V} \quad (\lside{ax:V}) \\
    \axtag{ax:DiaDual} \quad & \lquote{ax:DiaDual} 
    && \axtag{pr:G} \hspace{4pt}
    \vcenter{\hbox{%
    \AxiomC{$\vdash \phi$}
    \UnaryInfC{$\Gamma \vdash [\alpha]\phi$}
    \DisplayProof
    }}
  \end{align*}
\end{lemma}
\begin{proof}
  The soundness proofs of these axioms and proof rules can be found elsewhere \cite{DBLP:journals/jar/Platzer17,Platzer18}.
\end{proof}

Lemma~\ref{lemma:differential-axioms} lists the \emph{differential equation axioms} of \dL. 
These sound axioms enable deductive verification of ODE properties, which serve as 
an essential tool for the analytic completeness of ODE invariance in \dL \cite{DBLP:journals/jacm/PlatzerT20}.

\begin{lemma}[Differential axioms]\label{lemma:differential-axioms}
  The following axioms enable sound reasoning of differential equations.
  \begin{align*}
    \axtag{ax:DW} \quad & \lquote{ax:DW} \\
    \axtag{ax:DX} \quad & \lquote{ax:DX} \qquad (\lside{ax:DX}) \\
    \axtag{ax:DMP} \quad & \lquote{ax:DMP} \\
    \axtag{ax:DI_equal} \quad & \lquote{ax:DI_equal} \\
    \axtag{ax:DI_inequal} \quad & \lquote{ax:DI_inequal} \\
    \axtag{ax:DC} \quad & \lquote{ax:DC} \\
  \end{align*}
\end{lemma}
\begin{proof}
  Soundness proofs of differential axioms of \dL are presented elsewhere \cite{DBLP:journals/jar/Platzer17,Platzer18,DBLP:journals/jacm/PlatzerT20}.
\end{proof}

\subsection{Derived Axioms and Rules}

Lemma~\ref{lem:derived-axioms} lists some standard derived \dL axioms and proof rules that are 
derived from the previously stated base axioms. These will come in handy throughout 
the proofs and derivations presented in this paper.

\begin{lemma}[Derived Axioms and Proof Rules]\label{lem:derived-axioms}
  \begin{align*}
    & \axtag{dW} \hspace{4pt}
    \vcenter{\hbox{%
    \AxiomC{$Q \vdash \phi$}
    \UnaryInfC{$\Gamma \vdash \fBox[\ode[x]{f(x)}[Q]][\phi]$}
    \DisplayProof
    }} 
    && \axtag{ax:box_and} \quad \lquote{ax:box_and} \\
    & \axtag{pr:M[']} \hspace{4pt}
    \vcenter{\hbox{%
    \AxiomC{$Q, \psi \vdash \phi$}
    \AxiomC{$\Gamma \vdash \fBox[\ode[x]{f(x)}[Q]][\psi]$}
    \BinaryInfC{$\Gamma \vdash \fBox[\ode[x]{f(x)}[Q]][\phi]$}
    \DisplayProof
    }}
  \end{align*}
\end{lemma}
\begin{proof}
  The proof of these derived axioms and proof rules can be found elsewhere \cite{Platzer18,DBLP:journals/jacm/PlatzerT20}.
\end{proof}

\section{Proofs}

This appendix provides the complete proofs of the presented axioms and proof rules. All 
the derivations are conducted syntactically in the sound proof calculus of \dL.

\subsection{Comparison Principle}\label{appendix:comparison}
Here we provide the proofs of the axioms \lref{ax:SCP} and \lref{ax:VCP}, which 
internalize the \emph{scalar} and \emph{vector comparison principles} in \dL, respectively, 
and serve as the basis for the rest of the presented axioms and proof rules.

\begin{proof}[Proof (Lemma~\ref{lem:scp})]
  The axiom \lref{ax:SCP} is proven syntactically in \dL's proof calculus. The 
  following abbreviations are used throughout the proof.
  \begin{align*}
    \Gamma &\equiv h \geq z \Land \LocLip{\eta} \\
    Q_{\tau} & \equiv Q \Land \tau \geq 0 \\
    \alpha_{x\tau} &\equiv \odes{x' = \mathbf{f}(x), \tau' = 1} \\
    \alpha_{x\tau z} &\equiv \odes{x' = \mathbf{f}(x), \tau' = 1, z' = \eta(z,\tau)}
  \end{align*}
  Recall the following from Definition~\ref{def:local-lip}, which encodes the local Lipschitz 
  condition in real arithmetic, 
  where $B_{r}(a) \equiv t \geq 0 \Land \norm{x - a}^{2} < r^{2} \Land \norm{y - a}^{2} < r^{2}$.
  \[\LocLip{\eta} \equiv \forall \delta \forall a \exists r \exists k \forall t \forall x \forall y (t \leq \delta^{2} \Land B_{r}(a) \to \norm{\eta(x,t) - \eta(y,t)}^{2} \leq k^{2}\norm{x - y}^{2})\]

  After some propositional rewriting the proof begins with the addition of $\tau \geq 0$ 
  in the domain constraint by an application of \lref{ax:DC}.
  \begin{sequentproof}
    \ax[]{\Gamma, \fBox[\ode{\alpha_{x\tau}}[Q]][\dot{h} \geq \eta(h,\tau)]}{\fBox[\ode{\alpha_{x\tau z}}[Q_{\tau}]][h \geq z]}
    \un[ax:DC,dI_inequal]{\Gamma, \fBox[\ode{\alpha_{x\tau}}[Q]][\dot{h} \geq \eta(h,\tau)]}{\fBox[\ode{\alpha_{x\tau z}}[Q]][h \geq z]}
  \end{sequentproof}
  The resulting premise from \lref{ax:DC} is proved by an application of \lref{dI_inequal}, 
  because $Q \vdash \tau' \geq 0$ is provable as $\tau' = 1$.
  
  Continuing from there, the comparison system $z' = \eta(z,\tau)$ is introduced in the 
  antecedent by an application of \lref{ax:BDG}, resulting in two premises. The 
  left premise is abbreviated as \textcircled{a} and right premise as \textcircled{b}, 
  which will be proven separately.
  \begin{sequentproof}
    \ax[]{\Gamma, \fBox[\ode{\alpha_{x\tau z}}[Q]][\dot{h} \geq \eta(h,\tau)]}{\textcircled{a}}
    \ax[]{\Gamma}{\textcircled{b}}
    \bi[ax:BDG]{\Gamma, \fBox[\ode{\alpha_{x\tau}}[Q]][\dot{h} \geq \eta(h,\tau)]}{\fBox[\ode{\alpha_{x\tau z}}[Q_{\tau}]][h \geq z]}
  \end{sequentproof}

  Now continuing with the right premise \textcircled{b}. We need to prove that 
  the solution of the comparison system $z' = \eta(z,\tau)$ is bounded from above by  
  $c\tau$, where $\tau$ is the solution of the ODE $\tau' = 1$ and $c$ is some real constant. 
  Basically this implies that the solution 
  of the comparison system does not blow-up in the (compact) domain of existence. 
  In order to show this we first begin by applying axiom \lref{ax:DI_inequal} followed 
  by axiom \lref{()'}.
  \begin{sequentproof}
    \ax[]{\Gamma}{\fBox[\odes{x' = \mathbf{f}(x), \tau' = 1, z' = \eta(z,\tau)}[Q_{\tau}]][2zz' \leq c\tau']}
    \un[()']{\Gamma}{\fBox[\odes{x' = \mathbf{f}(x), \tau' = 1, z' = \eta(z,\tau)}[Q_{\tau}]][(\norm{z}^{2})' \leq (c\tau)']}
    \un[ax:DI_inequal]{\Gamma}{\fBox[\odes{x' = \mathbf{f}(x), \tau' = 1, z' = \eta(z,\tau)}[Q_{\tau}]][\norm{z}^{2} \leq c\tau]}
  \end{sequentproof}
  Continuing from there an application of axioms \lref{DE} and \lref{[:=]} replaces 
  the prime variables. After which an application of \lref{dW} removes the box modality. 
  Observe that the formula $\Gamma$ is kept soundly in the context as it is quantified.
  \begin{sequentproof}
    \ax[]{\LocLip{\eta}, Q, \tau \geq 0}{2z\cdot \eta(z,\tau) \leq c}
    \un[dW]{\Gamma}{\fBox[\odes{x' = \mathbf{f}(x), \tau' = 1, z' = \eta(z,\tau)}[Q_{\tau}]][2z\cdot\eta(z,\tau) \leq c]}
    \un[DE,[:=]]{\Gamma}{\fBox[\odes{x' = \mathbf{f}(x), \tau' = 1, z' = \eta(z,\tau)}[Q_{\tau}]][2zz' \leq c\tau']}
  \end{sequentproof}
  Now we are left to prove the remaining real arithmetic. To do so first we need to 
  unfold the definition of $\LocLip{\eta}$ with $B_{r}(a) \equiv t \geq 0 \Land \norm{x - a}^{2} < r^{2} \Land \norm{y - a}^{2} < r^{2}$ as follows.
  \[\forall \delta \forall a \exists r \exists k \forall t \forall x \forall y (t \leq \delta^{2} \Land B_{r}(a) \to \norm{\eta(x,t) - \eta(y,t)}^{2} \leq k^{2}\norm{x - y}^{2})\]
  Now instantiating the quantifiers as, $\delta \geq 0$, $a = 0$, $r > 0$, $k \geq 0$, 
  $t = \tau$, $x = z$, and $y = 0$, we get
  \[\tau \leq \delta^{2} \Land \norm{z}^{2} < r^{2} \to \norm{\eta(z,\tau) - \eta(0, \tau)}^{2} \leq k^{2} \norm{z}^{2}\]
  The premise on the left of the implication is provable in real arithmetic. Since 
  $z$ and $\eta(\cdot)$ are scalars we can rewrite the right of the above 
  implication without the norms, as following.
  \begin{equation}\label{eq:eta-inequality}
    (\eta(z,\tau) - \eta(0,\tau))^{2} \leq k^{2}z^{2}
  \end{equation}
  Now by rewriting the following inequality we can derive the required inequality 
  on $2z\cdot\eta(z,\tau)$ as follows.
  \begin{equation}\label{eq:z-tau-inequality}
    \begin{aligned}
      0 &\leq (z - \eta(z,\tau))^{2} \\
      2z\cdot\eta(z,\tau) &\leq z^{2} + \eta(z,\tau)^{2}
    \end{aligned}
  \end{equation}
  This shows in order to prove the given arithmetic we need to give an upper bound 
  on $z^{2}$ and $\eta(z,\tau)^{2}$, respectively. We already know $\norm{z}^{2} < r^{2}$, 
  which implies $z^{2} < r^{2}$, since $z$ is a scalar. Now we are left with determining 
  a bound on $\eta(z,\tau)^{2}$.
  Continuing from (\ref{eq:eta-inequality}) we can rearrange terms as following. 
  \begin{align*}
    \eta(z,\tau)^{2} &\leq k^{2}z^{2} + 2\eta(z,\tau)\cdot\eta(0,\tau) - \eta(0,\tau)^{2} + \eta(0,\tau)^{2} - \eta(0,\tau)^{2} \\
    &\leq k^{2}z^{2} + 2\eta(0,\tau)\cdot(\eta(z,\tau) - \eta(0,\tau)) + \eta(0,\tau)^{2} \\
    &\leq k^{2}z^{2} + ((\eta(z,\tau) - \eta(0,\tau)) +\eta(0,\tau))^{2} - (\eta(z,\tau) - \eta(0,\tau))^{2}
  \end{align*}
  First using (\ref{eq:eta-inequality}) we can replace some terms such that they 
  cancel out. Then we use the standard inequality $(a + b)^{2} \leq 2a^{2} + 2b^{2}$ 
  to get the following upper bound on $\eta(z,\tau)^{2}$.
  \begin{align*}
    \eta(z,\tau)^{2} &\leq \cancel{k^{2}z^{2}} - \cancel{k^{2}z^{2}} + ((\eta(z,\tau) - \eta(0,\tau)) +\eta(0,\tau))^{2} \\
    &\leq 2(\eta(z,\tau) - \eta(0,\tau))^{2} + 2\eta(0,\tau)^{2} \\
    &\leq 2k^{2}z^{2} + 2\eta(0,\tau)^{2}
  \end{align*}
  Recall from Section~\ref{syntax:terms}, according to the term syntax of \dL, fixed 
  function symbols like $\eta(\cdot)$ are continuous in their respective arguments. 
  Since we know $\eta(\cdot,\tau)$ is locally Lipschitz and continuous in $\tau$, therefore 
  it is bounded on a compact domain. From the context of the sequent we know $\tau \geq 0$ 
  and $\tau \leq \delta^{2}$, which means the domain of $\tau$ i.e., the interval 
  $[0,\delta^{2}]$ is compact. This implies there exists some $M \geq 0$, such that 
  $\eta(0,\tau)^{2} \leq M^{2}$ is true. This implies the following upper bound 
  on $\eta(z,\tau)^{2}$.
  \[\eta(z,\tau)^{2} \leq 2k^{2}z^{2} + 2M^{2}\]
  With that we have all the ingredients for cooking up our upper bound. Finally, 
  going back to inequality (\ref{eq:z-tau-inequality}) 
  we have
  \begin{align*}
    2z\cdot\eta(z,\tau) &\leq z^{2} + 2k^{2}z^{2} + 2M^{2} \\
    &\leq z^{2}(1 + 2k^{2}) + 2M^{2} \\
    &\leq r^{2}(1 + 2k^{2}) + 2M^{2}
  \end{align*}
  Now we can first cut in $M \geq 0$, whose square acts as the upper bound on $\eta(0,\tau)^{2}$, 
  and then we can cut in $c = r^{2}(1 + 2k^{2}) + 2M^{2}$. The branch closes with 
  the application of \lref{Real}.
  \begin{sequentproof}
    \ax*[Real]{\LocLip{\eta}, Q, M \geq 0, c = r^{2}(1 + 2k^{2}) + 2M^{2}, \tau \geq 0}{2z\cdot \eta(z,\tau) \leq c}
    \un[cut,exL]{\LocLip{\eta}, Q, \tau \geq 0}{2z\cdot \eta(z,\tau) \leq c}
  \end{sequentproof}

  Coming back to the left premise \textcircled{a} from the application of \lref{ax:BDG}, the 
  proof continuous with the addition of $\dot{h} \geq \eta(h,\tau)$ to the domain constraint 
  of the succedent by an application of \lref{ax:DC}. Next, \lref{pr:M[']} monotonically 
  rewrites the postcondition in the antecedent to $h - z \geq 0$.
  \begin{sequentproof}[align]
    \ax[]{\Gamma, \fBox[\ode{\alpha_{x\tau z}}[Q]][\dot{h} \geq \eta(h,\tau)]}{\fBox[\ode{\alpha_{x\tau z}}[Q_{\tau} \Land \dot{h} \geq \eta(h,\tau)]][h - z \geq 0]}
    \un[pr:M[']]{\Gamma, \fBox[\ode{\alpha_{x\tau z}}[Q]][\dot{h} \geq \eta(h,\tau)]}{\fBox[\ode{\alpha_{x\tau z}}[Q_{\tau} \Land \dot{h} \geq \eta(h,\tau)]][h \geq z]}
    \un[ax:DC]{\Gamma, \fBox[\ode{\alpha_{x\tau z}}[Q]][\dot{h} \geq \eta(h,\tau)]}{\fBox[\ode{\alpha_{x\tau z}}[Q_{\tau}]][h \geq z]}
  \end{sequentproof}
  Now we cut in the provable formula $\exp(k\tau)\cdot(h-z) \geq 0$ using \lref{cut} and \lref{Real}. 
  After that the postcondition in the succedent is monotonically replaced by $\exp(k\tau)\cdot(h-z) \geq 0$ 
  using \lref{pr:M[']}.
  \begin{sequentproof}
    \ax[]{\Gamma, \exp(k\tau)\cdot(h-z) \geq 0}{\fBox[\ode{\alpha_{x\tau z}}[Q_{\tau} \Land \dot{h} \geq \eta(h,\tau)]][\exp(k\tau)\cdot(h-z) \geq 0]}
    \un[pr:M['],Real]{\Gamma, \exp(k\tau)\cdot(h-z) \geq 0, \fBox[\ode{\alpha_{x\tau z}}[Q]][\dot{h} \geq \eta(h,\tau)]}{\fBox[\ode{\alpha_{x\tau z}}[Q_{\tau} \Land \dot{h} \geq \eta(h,\tau)]][h - z \geq 0]}
    \un[cut,Real]{\Gamma, \fBox[\ode{\alpha_{x\tau z}}[Q]][\dot{h} \geq \eta(h,\tau)]}{\fBox[\ode{\alpha_{x\tau z}}[Q_{\tau} \Land \dot{h} \geq \eta(h,\tau)]][h - z \geq 0]}
  \end{sequentproof}
  Now we proceed to proving this new invariant $\exp(k\tau)\cdot(h-z) \geq 0$. Where $\exp(\cdot)$ is 
  the real exponential function, $k$ is the Lipschitz constant, and $\tau$ is the 
  variable for tracking passage of time. The insight 
  here is that by proving this new invariant we will essentially prove $\exp(k\tau)\cdot(h-z)$ 
  is monotonically increasing, which would mean $h - z \geq 0$ is true for all time 
  as long as the solution of the ODEs exists. We continue by applying axiom \lref{ax:DI_inequal}.
  \begin{sequentproof}
    \ax[]{\Gamma}{\fBox[\ode{\alpha_{x\tau z}}[Q_{\tau} \Land \dot{h} \geq \eta(h,\tau)]][(\exp(k\tau)\cdot(h-z))' \geq 0]}
    \un[ax:DI_inequal]{\Gamma, \exp(k\tau)\cdot(h-z) \geq 0}{\fBox[\ode{\alpha_{x\tau z}}[Q_{\tau} \Land \dot{h} \geq \eta(h,\tau)]][\exp(k\tau)\cdot(h-z) \geq 0]}
  \end{sequentproof}
  
  All the resulting differential variables appearing in the postcondition after 
  applying \lref{()'} axiom are replaced by appropriate terms from the context of 
  the ODE $\alpha_{x\tau z}$ by \lref{DE} and \lref{[:=]}.
  \begin{sequentproof}
    \ax[]{\Gamma}{\fBox[\ode{\alpha_{x\tau z}}[Q_{\tau} \Land \dot{h} \geq \eta(h,\tau)]][\exp(k\tau)\cdot(k(h-z) + \dot{h} - \eta(z,\tau)) \geq 0]}
    \un[DE,[:=]]{\Gamma}{\fBox[\ode{\alpha_{x\tau z}}[Q_{\tau} \Land \dot{h} \geq \eta(h,\tau)]][(\exp(k\tau))'\cdot(h-z) + \exp(k\tau) \cdot (h' - z') \geq 0]}
    \un[()']{\Gamma}{\fBox[\ode{\alpha_{x\tau z}}[Q_{\tau} \Land \dot{h} \geq \eta(h,\tau)]][(\exp(k\tau)\cdot(h-z))' \geq 0]}
  \end{sequentproof}
  Continuing from there, we first unfold $\Gamma \equiv h \geq z \Land \LocLip{\eta}$ 
  then apply \lref{dW} to remove the box modality, again \LocLip{$\eta$} 
  is kept soundly in the context since it is quantified.
  \begin{sequentproof}
    \ax[]{\LocLip{\eta}, Q, \tau \geq 0, \dot{h} - \eta(h,\tau) \geq 0}{\exp(k\tau)\cdot(k(h-z) + \dot{h} - \eta(z,\tau)) \geq 0}
    \un[Real]{\LocLip{\eta}, Q, \tau \geq 0, \dot{h} - \eta(h,\tau) \geq 0}{\exp(k\tau)\cdot(k(h-z) + \dot{h} - \eta(z,\tau)) \geq 0}
    \un[dW]{h \geq z \Land \LocLip{\eta}}{\fBox[\ode{\alpha_{x\tau z}}[Q_{\tau} \Land \dot{h} \geq \eta(h,\tau)]][\exp(k\tau)\cdot(k(h-z) + \dot{h} - \eta(z,\tau)) \geq 0]}
  \end{sequentproof}
  Now observe that we have $\dot{h} \geq \eta(h,\tau)$ in the assumption, which implies the following
  \[\exp(k\tau)\cdot(k(h-z) + \dot{h} - \eta(z,\tau)) \geq \exp(k\tau)\cdot(k(h-z) + \eta(h,\tau) - \eta(z,\tau)) \geq 0\]
  We know that $\exp(k\tau) \geq 0$ always, this leaves us with proving $k(h-z) + \eta(h,\tau) - \eta(z,\tau) \geq 0$. 
  Now we first unfold \LocLip{$\eta$} and instantiate the quantifiers appropriately, 
  in particular $r > 0$, $k \geq 0$, $x = z$, and $y = h$.
  \[B_{r}(a) \to \norm{\eta(z,\tau) - \eta(h,\tau)}^{2} \leq k^{2}\norm{z - h}^{2}\]
  $B_{r}(a)$ is provable for some arbitrary $a$, implying $h$ and $z$ are in the 
  open ball of radius $r$. Continuing with the succedent.
  \begin{align*}
    \norm{\eta(z,\tau) - \eta(h,\tau)}^{2} &\leq k^{2}\norm{z - h}^{2} \\
    \norm{\eta(z,\tau) - \eta(h,\tau)} &\leq k\norm{z - h} \\
    \norm{\eta(h,\tau) - \eta(z,\tau)} &\geq -k\norm{h - z} \\
    \norm{\eta(h,\tau) - \eta(z,\tau)} + k\norm{h - z} &\geq 0
  \end{align*}
  Since $h$, $z$, and $\eta(\cdot)$ are scalar functions, hence we can rewrite it as follows
  \[k(h-z) + \eta(h,\tau) - \eta(z,\tau) \geq 0\]
  hence proving the required inequality, the branch closes by \lref{allL}, 
  \lref{exL}, \lref{allL}, and \lref{Real}.
  \begin{sequentproof}[align]
    \ax*[Real]{\LocLip{\eta}, Q, \tau \geq 0, \dot{h} - \eta(h,\tau) \geq 0}{\exp(k\tau)\cdot(k(h-z) + \dot{h} - \eta(z,\tau)) \geq 0}
    \un[allL]{\LocLip{\eta}, Q, \tau \geq 0, \dot{h} - \eta(h,\tau) \geq 0}{\exp(k\tau)\cdot(k(h-z) + \dot{h} - \eta(z,\tau)) \geq 0}
    \un[allL,exL]{\LocLip{\eta}, Q, \tau \geq 0, \dot{h} - \eta(h,\tau) \geq 0}{\exp(k\tau)\cdot(k(h-z) + \dot{h} - \eta(z,\tau)) \geq 0}
  \end{sequentproof}
\end{proof}

\begin{proof}[Proof (Lemma~\ref{lem:vcp})]\label{proof:vcp}
  Axiom \lref{ax:VCP} is proven syntactically using sound axioms of \dL. The following 
  abbreviations are used throughout the proof.
  \begin{align*}
    \Gamma &\equiv \mathbf{h} \geq z \Land \LocLip{\boldsymbol{\eta}} \Land \QuasiMon{\boldsymbol{\eta}} \\
    \alpha_{x\tau} &\equiv \odes{x' = \mathbf{f}(x), \tau' = 1} \\
    \alpha_{x\tau z} &\equiv \odes{x' = \mathbf{f}(x), \tau' = 1, z' = \boldsymbol{\eta}(z,\tau)}
  \end{align*}
  Recall \LocLip{$\boldsymbol{\eta}$} and \QuasiMon{$\boldsymbol{\eta}$} are unfolded as 
  following as per Definition~\ref{def:local-lip} and Definition~\ref{def:quasi-mon}, 
  respectively.
  \begin{align*}
    \LocLip{\eta} &\equiv \forall \delta \forall a \exists r \exists k \forall t \forall x \forall y (t \leq \delta^{2} \Land B_{r} \to \norm{\eta(x,t) - \eta(y,t)}^{2} \leq k^{2}\norm{x - y}^{2}) \\
    \QuasiMon{\boldsymbol{\eta}} &\equiv \forall t \forall x \forall y \left(x \geq y \to \bigwedge^{n}_{i=1}\left(x_{i} = y_{i} \to \eta_{i}(x,t) \geq \eta_{i}(y,t)\right)\right)
  \end{align*}
  Where $B_{r}(a) \equiv t \geq 0 \Land \norm{x - a}^{2} < r^{2} \Land \norm{y - a}^{2} < r^{2}$.

  After propositional rewriting the proof begins by adding $\tau \geq 0$ to the domain 
  constraint with the resulting premise from \lref{ax:DC} being proved by \lref{dI_inequal} 
  as $Q \vdash \tau' \geq 0$ is provable as $\tau' = 1$.
  \begin{sequentproof}
    \ax[]{\Gamma, \fBox[\ode{\alpha_{x\tau}}[Q]][\dot{\mathbf{h}} \geq \boldsymbol{\eta}(\mathbf{h},\tau)]}{\fBox[\ode{\alpha_{x\tau z}}[Q_{\tau}]][\mathbf{h} \geq z]}
    \un[ax:DC,dI_inequal]{\Gamma, \fBox[\ode{\alpha_{x\tau}}[Q]][\dot{\mathbf{h}} \geq \boldsymbol{\eta}(\mathbf{h},\tau)]}{\fBox[\ode{\alpha_{x\tau z}}[Q]][\mathbf{h} \geq z]}
  \end{sequentproof}

  Next we add the comparison system $z' = \boldsymbol{\eta}(z,\tau)$ to the antecedent 
  by an application of \lref{ax:BDG}, resulting in two premises with the left and right 
  premise abbreviated as \textcircled{a} and \textcircled{b}, respectively.
  \begin{sequentproof}
    \ax[]{\Gamma, \fBox[\ode{\alpha_{x\tau}}[Q]][\dot{\mathbf{h}} \geq \boldsymbol{\eta}(\mathbf{h},\tau)]}{\textcircled{a}}
    \ax[]{\Gamma}{\textcircled{b}}
    \bi[ax:BDG]{\Gamma, \fBox[\ode{\alpha_{x\tau}}[Q]][\dot{\mathbf{h}} \geq \boldsymbol{\eta}(\mathbf{h},\tau)]}{\fBox[\ode{\alpha_{x\tau z}}[Q_{\tau}]][\mathbf{h} \geq z]}
  \end{sequentproof}
  
  Continuing with the right premise \textcircled{b}. This open premise requires us 
  to prove that the solution of the newly introduced comparison system $z' = \boldsymbol{\eta}(z,\tau)$ 
  is upper bounded by the solution of the ODE $\tau' = 1$. Essentially saying that 
  the solution $z$ does not blow up in the compact domain of existence. This is 
  shown by applying axiom \lref{ax:DI_inequal} followed by \lref{dW}.
  \begin{sequentproof}
    \ax[]{\LocLip{\boldsymbol{\eta}}, \QuasiMon{\boldsymbol{\eta}}, Q, \tau \geq 0}{2z\cdot \boldsymbol{\eta}(z,\tau) \leq c}
    \un[dW]{\Gamma}{\fBox[\odes{x' = \mathbf{f}(x), \tau' = 1, z' = \boldsymbol{\eta}(z,\tau)}[Q_{\tau}]][2z\cdot \boldsymbol{\eta}(z,\tau) \leq c]}
    \un[ax:DI_inequal]{\Gamma}{\fBox[\odes{x' = \mathbf{f}(x), \tau' = 1, z' = \boldsymbol{\eta}(z,\tau)}[Q_{\tau}]][\norm{z}^{2} \leq c\tau]}
  \end{sequentproof}
  The resulting real arithmetic is proven as follows. First we unfold the definition of 
  $\LocLip{\boldsymbol{\eta}}$ with $B_{r}(a) \equiv t \geq 0 \Land \norm{x - a}^{2} < r^{2} \Land \norm{y - a}^{2} < r^{2}$.
  \[\forall \delta \forall a \exists r \exists k \forall t \forall x \forall y (t \leq \delta^{2} \Land B_{r}(a) \to \norm{\boldsymbol{\eta}(x,t) - \boldsymbol{\eta}(y,t)}^{2} \leq k^{2}\norm{x - y}^{2})\]
  Now instantiating the quantifiers as, $\delta \geq 0$, $a = 0$, $r > 0$, $k \geq 0$, 
  $t = \tau$, $x = z$, and $y = 0$, we get
  \[\tau \leq \delta^{2} \Land \norm{z}^{2} < r^{2} \to \norm{\boldsymbol{\eta}(z,\tau) - \boldsymbol{\eta}(0, \tau)}^{2} \leq k^{2} \norm{z}^{2}\]
  The premise on the left of the implication is provable in real arithmetic. Now using 
  Young's inequality i.e., $ab \leq \frac{a^{p}}{p} + \frac{b^{q}}{q}$, for $p = q = 2$, we get the 
  following relation.
  \[2\norm{z} \norm{\boldsymbol{\eta}(z,\tau)} \leq \norm{z}^{2} + \norm{\boldsymbol{\eta}(z,\tau)}^{2}\]
  Also using the standard inequality $(a + b)^{2} \leq 2a^{2} + 2b^{2}$ after rewriting 
  $\boldsymbol{\eta}(z,\tau) = (\boldsymbol{\eta}(z,\tau) - \boldsymbol{\eta}(0,\tau)) + \boldsymbol{\eta}(0,\tau)$ we get 
  the following relation.
  \begin{align*}
    \norm{\boldsymbol{\eta}(z,\tau)}^{2} &\leq 2\norm{\boldsymbol{\eta}(z,\tau) - \boldsymbol{\eta}(0,\tau)}^{2} + 2\norm{\boldsymbol{\eta}(0,\tau)}^{2} \\
    \norm{\boldsymbol{\eta}(z,\tau)}^{2} &\leq 2k^{2}\norm{z}^{2} + 2\norm{\boldsymbol{\eta}(0,\tau)}^{2}
  \end{align*}
  Now recall from Section~\ref{syntax:terms} the term syntax of \dL ensures that the 
  function $\boldsymbol{\eta}(\cdot, \tau)$ is continuous in all its arguments, and 
  we also know $\boldsymbol{\eta}(\cdot, \tau)$ is locally Lipschitz, which implies it 
  is bounded on a compact domain. Now from the context we have $\tau \geq 0$ and $\tau \leq \delta^{2}$, 
  which means the interval $[0, \delta^{2}]$ is compact. Hence, there exists a real 
  $M \geq 0$ such that $\norm{\boldsymbol{\eta}(0, \tau)}^{2} \leq M^{2}$. This gives us 
  the following inequality.
  \[\norm{\boldsymbol{\eta}(z,\tau)}^{2} \leq 2k^{2}\norm{z}^{2} + 2M^{2}\]
  We already know $\norm{z}^{2} \leq r^{2}$, combining this with the above inequalities we get,
  \begin{align*}
    2\norm{z} \norm{\boldsymbol{\eta}(z,\tau)} &\leq \norm{z}^{2} + 2k^{2}\norm{z}^{2} + 2M^{2} \\
    &\leq r^{2}(1 + 2k^{2}) + 2M^{2}
  \end{align*}
  Now using the Cauchy-Schwarz inequality i.e., $|v \cdot w| \leq \norm{v}\norm{w}$, we get an upper bound 
  on our required inequality.
  \[2z \cdot \boldsymbol{\eta}(z,\tau) \leq 2\norm{z} \norm{\boldsymbol{\eta}(z,\tau)} \leq r^{2}(1 + 2k^{2}) + 2M^{2}\]
  By cutting in $M \geq 0$ and $c = r^{2}(1 + 2k^{2}) + 2M^{2}$ and instantiating 
  all the quantifiers as showed above the branch closes with \lref{Real}.
  \begin{sequentproof}
    \ax*[Real]{\LocLip{\boldsymbol{\eta}}, \QuasiMon{\boldsymbol{\eta}}, Q, \tau \geq 0, M \geq 0, c = r^{2}(1 + 2k^{2}) + 2M^{2}}{2z\cdot \boldsymbol{\eta}(z,\tau) \leq c}
    \un[cut,Real]{\LocLip{\boldsymbol{\eta}}, \QuasiMon{\boldsymbol{\eta}}, Q, \tau \geq 0}{2z\cdot \boldsymbol{\eta}(z,\tau) \leq c}
  \end{sequentproof}

  Continuing with the open premise \textcircled{a}. We first add $\dot{\mathbf{h}} \geq \boldsymbol{\eta}(\mathbf{h},\tau)$ 
  to the domain constraint in the succedent. Then we rewrite the postcondition in 
  the succedent to $\mathbf{h} - z \geq 0$ using \lref{pr:M[']}.
  \begin{sequentproof}
    \ax[]{\Gamma, \fBox[\ode{\alpha_{x\tau z}}[Q]][\dot{\mathbf{h}} \geq \boldsymbol{\eta}(\mathbf{h},\tau)]}{\fBox[\ode{\alpha_{x\tau z}}[Q_{\tau} \Land \dot{\mathbf{h}} \geq \boldsymbol{\eta}(\mathbf{h},\tau)]][\mathbf{h} - z \geq 0]}
    \un[pr:M[']]{\Gamma, \fBox[\ode{\alpha_{x\tau z}}[Q]][\dot{\mathbf{h}} \geq \boldsymbol{\eta}(\mathbf{h},\tau)]}{\fBox[\ode{\alpha_{x\tau z}}[Q_{\tau} \Land \dot{\mathbf{h}} \geq \boldsymbol{\eta}(\mathbf{h},\tau)]][\mathbf{h} \geq z]}
    \un[ax:DC]{\Gamma, \fBox[\ode{\alpha_{x\tau z}}[Q]][\dot{\mathbf{h}} \geq \boldsymbol{\eta}(\mathbf{h},\tau)]}{\fBox[\ode{\alpha_{x\tau z}}[Q_{\tau}]][\mathbf{h} \geq z]}
  \end{sequentproof}
  Next we cut in the provable formula $\exp(k\tau)\cdot(\mathbf{h} - z) \geq 0$. 
  Where $\exp(\cdot)$ is the real exponential function and $k$ is the Lipschitz constant.
  Then the postcondition in the succedent is replaced by $\exp(k\tau)\cdot(\mathbf{h} - z) \geq 0$ 
  using \lref{pr:M[']}, which acts as the new invariant.
  \begin{sequentproof}
    \ax[]{\Gamma, \exp(k\tau)\cdot(\mathbf{h} - z) \geq 0}{\fBox[\ode{\alpha_{x\tau z}}[Q_{\tau} \Land \dot{\mathbf{h}} \geq \boldsymbol{\eta}(\mathbf{h},\tau)]][\exp(k\tau)\cdot(\mathbf{h} - z) \geq 0]}
    \un[pr:M[']]{\Gamma, \exp(k\tau)\cdot(\mathbf{h} - z) \geq 0, \fBox[\ode{\alpha_{x\tau z}}[Q]][\dot{\mathbf{h}} \geq \boldsymbol{\eta}(\mathbf{h},\tau)]}{\fBox[\ode{\alpha_{x\tau z}}[Q_{\tau} \Land \dot{\mathbf{h}} \geq \boldsymbol{\eta}(\mathbf{h},\tau)]][\mathbf{h} - z \geq 0]}
    \un[cut,Real]{\Gamma, \fBox[\ode{\alpha_{x\tau z}}[Q]][\dot{\mathbf{h}} \geq \boldsymbol{\eta}(\mathbf{h},\tau)]}{\fBox[\ode{\alpha_{x\tau z}}[Q_{\tau} \Land \dot{\mathbf{h}} \geq \boldsymbol{\eta}(\mathbf{h},\tau)]][\mathbf{h} - z \geq 0]}
  \end{sequentproof}
  Since all the vector inequality $\mathbf{h} - z \geq 0$ is component-wise and 
  $\exp(k\tau)$ is a scalar, the real arithmetic equivalence $\exp(k\tau)\cdot(\mathbf{h} - z) \geq 0 \Liff \bigwedge^{n}_{i=1}(\exp(k\tau)\cdot(h_{i} - z_{i}) \geq 0)$, 
  where each $h_{i}$ and $z_{i}$ are scalars. The postcondition in the succedent is rewritten 
  to this equivalence using \lref{pr:M[']}.
  \begin{sequentproof}
    \ax[]{\Gamma, \bigwedge^{n}_{i=1}(\exp(k\tau)\cdot(h_{i} - z_{i}) \geq 0)}{\fBox[\ode{\alpha_{x\tau z}}[Q_{\tau} \Land \dot{\mathbf{h}} \geq \boldsymbol{\eta}(\mathbf{h},\tau)]][\bigwedge^{n}_{i=1}(\exp(k\tau)\cdot(h_{i} - z_{i}) \geq 0)]}
    \un[pr:M['],Real]{\Gamma, \exp(k\tau)\cdot(\mathbf{h} - z) \geq 0}{\fBox[\ode{\alpha_{x\tau z}}[Q_{\tau} \Land \dot{\mathbf{h}} \geq \boldsymbol{\eta}(\mathbf{h},\tau)]][\exp(k\tau)\cdot(\mathbf{h} - z) \geq 0]}
  \end{sequentproof}
  Now we decompose the conjunctions in the postcondition using \lref{ax:box_and}, 
  which results in conjunctions of box modalities.
  \begin{sequentproof}
    \ax[]{\Gamma, \bigwedge^{n}_{i=1}(\exp(k\tau)\cdot(h_{i} - z_{i}) \geq 0)}{\bigwedge^{n}_{i=1}(\fBox[\ode{\alpha_{x\tau z}}[Q_{\tau} \Land \dot{\mathbf{h}} \geq \boldsymbol{\eta}(\mathbf{h},\tau)]][\exp(k\tau)\cdot(h_{i} - z_{i}) \geq 0])}
    \un[ax:box_and]{\Gamma, \bigwedge^{n}_{i=1}(\exp(k\tau)\cdot(h_{i} - z_{i}) \geq 0)}{\fBox[\ode{\alpha_{x\tau z}}[Q_{\tau} \Land \dot{\mathbf{h}} \geq \boldsymbol{\eta}(\mathbf{h},\tau)]][\bigwedge^{n}_{i=1}(\exp(k\tau)\cdot(h_{i} - z_{i}) \geq 0)]}
  \end{sequentproof}
  Now with the application of \lref{andR} we get $n$ stacked open premises, each of them 
  consisting a box modality with one of the scalar components of the original vector 
  inequality.
  \begin{sequentproof}
    \ax[]{\Gamma, \bigwedge^{n}_{i=1}(\exp(k\tau)\cdot(h_{i} - z_{i}) \geq 0)}{\fBox[\ode{\alpha_{x\tau z}}[Q_{\tau} \Land \dot{\mathbf{h}} \geq \boldsymbol{\eta}(\mathbf{h},\tau)]][\exp(k\tau)\cdot(h_{1} - z_{1}) \geq 0]}
    \noLine
    \UnaryInfC{$\vdots$}
    \noLine
    \un[]{\Gamma, \bigwedge^{n}_{i=1}(\exp(k\tau)\cdot(h_{i} - z_{i}) \geq 0)}{\fBox[\ode{\alpha_{x\tau z}}[Q_{\tau} \Land \dot{\mathbf{h}} \geq \boldsymbol{\eta}(\mathbf{h},\tau)]][\exp(k\tau)\cdot(h_{n} - z_{n}) \geq 0]}
    \un[andR]{\Gamma, \bigwedge^{n}_{i=1}(\exp(k\tau)\cdot(h_{i} - z_{i}) \geq 0)}{\bigwedge^{n}_{i=1}(\fBox[\ode{\alpha_{x\tau z}}[Q_{\tau} \Land \dot{\mathbf{h}} \geq \boldsymbol{\eta}(\mathbf{h},\tau)]][\exp(k\tau)\cdot(h_{i} - z_{i}) \geq 0])}
  \end{sequentproof}
  From here we will only focus on the $i$-th component, as the same result can be 
  followed for all the $n$ premises. Since we have a scalar inequality in the postcondition 
  we can continue by applying axiom \lref{ax:DI_inequal}.
  \begin{sequentproof}
    \ax[]{\Gamma}{\fBox[\ode{\alpha_{x\tau z}}[Q_{\tau} \Land \dot{h} \geq \boldsymbol{\eta}(\mathbf{h},\tau)]][(\exp(k\tau)\cdot(h_{i} - z_{i}))' \geq 0]}
    \un[ax:DI_inequal]{\Gamma, \bigwedge^{n}_{i=1}(\exp(k\tau)\cdot(h_{i} - z_{i}) \geq 0)}{\fBox[\ode{\alpha_{x\tau z}}[Q_{\tau} \Land \dot{h} \geq \boldsymbol{\eta}(\mathbf{h},\tau)]][\exp(k\tau)\cdot(h_{i} - z_{i}) \geq 0]}
  \end{sequentproof}
  
  All the resulting differential variables appearing in the postcondition after 
  applying \lref{()'} axiom are replaced by appropriate terms from the $i$-th component of 
  the ODE $\alpha_{x\tau z}$ by \lref{DE} and \lref{[:=]}.
  \begin{sequentproof}
    \ax[]{\Gamma}{\fBox[\ode{\alpha_{x\tau z}}[Q_{\tau} \Land \dot{h} \geq \boldsymbol{\eta}(\mathbf{h},\tau)]][\exp(k\tau)\cdot(k(h_{i}-z_{i}) + \dot{h_{i}} - \eta_{i}(z,\tau)) \geq 0]}
    \un[DE,[:=]]{\Gamma}{\fBox[\ode{\alpha_{x\tau z}}[Q_{\tau} \Land \dot{h} \geq \boldsymbol{\eta}(\mathbf{h},\tau)]][(\exp(k\tau))'\cdot(h_{i}-z_{i}) + \exp(k\tau) \cdot (h_{i}' - z_{i}') \geq 0]}
    \un[()']{\Gamma}{\fBox[\ode{\alpha_{x\tau z}}[Q_{\tau} \Land \dot{h} \geq \boldsymbol{\eta}(\mathbf{h},\tau)]][(\exp(k\tau)\cdot(h_{i}-z_{i}))' \geq 0]}
  \end{sequentproof}
  Continuing from there, we first unfold $\Gamma \equiv \mathbf{h} \geq z \Land \LocLip{\boldsymbol{\eta}} \Land \QuasiMon{\boldsymbol{\eta}}$ 
  then apply \lref{dW} to remove the box modality, $\LocLip{\boldsymbol{\eta}}$ and 
  $\QuasiMon{\boldsymbol{\eta}}$ are kept soundly in the context as they are quantified.
  \begin{sequentproof}
    \ax[]{\LocLip{\boldsymbol{\eta}}, \QuasiMon{\boldsymbol{\eta}}, Q, \tau \geq 0, \dot{\mathbf{h}} - \boldsymbol{\eta}(\mathbf{h},\tau) \geq 0}{\exp(k\tau)\cdot(k(h_{i}-z_{i}) + \dot{h_{i}} - \eta_{i}(z,\tau)) \geq 0}
    \un[dW]{\Gamma}{\fBox[\ode{\alpha_{x\tau z}}[Q_{\tau} \Land \dot{h} \geq \eta(h,\tau)]][\exp(k\tau)\cdot(k(h_{i}-z_{i}) + \dot{h_{i}} - \eta_{i}(z,\tau)) \geq 0]}
  \end{sequentproof}
  First we rewrite $\dot{\mathbf{h}} - \boldsymbol{\eta}(\mathbf{h},\tau) \geq 0 \Liff \bigwedge^{n}_{i=1}(\dot{h_{i}} - \eta_{i}(\mathbf{h},\tau) \geq 0)$, since 
  the vector inequality holds component-wise.
  Now observe that we have $\dot{h_{i}} \geq \eta_{i}(\mathbf{h},\tau)$ in the assumption, which implies the following
  \[\exp(k\tau)\cdot(k(h_{i}-z_{i}) + \dot{h_{i}} - \eta_{i}(z,\tau)) \geq \exp(k\tau)\cdot(k(h_{i}-z_{i}) + \eta_{i}(\mathbf{h},\tau) - \eta_{i}(z,\tau)) \geq 0\]
  We know that $\exp(k\tau) \geq 0$ always, this leaves us with proving $k(h_{i}-z_{i}) + \eta_{i}(\mathbf{h},\tau) - \eta_{i}(z,\tau) \geq 0$.
  We continue proving this inequality by first constructing the following vector, 
  \[\tilde{z} = (h_{1}, \dots, h_{i-1}, z_{i}, h_{i+1}, \dots, h_{n})^{\top}\]
  which says the vectors $\tilde{z}$ and $z$ agree on the $i$-th component, i.e., $\tilde{z}_{i} = z_{i}$, 
  and the vectors $\tilde{z}$ and $\mathbf{h}$ agree on every component except $i$-th, 
  i.e., $\tilde{z}_{j} = h_{j}$ for all $j \neq i$. Unfolding $\QuasiMon{\boldsymbol{\eta}}$, 
  we have by quasimonotonicity,
  \[z \leq \tilde{z} \to \bigwedge^{n}_{i=1}\left(z_{i} = \tilde{z}_{i} \to \eta_{i}(z,\tau) \leq \eta_{i}(\tilde{z},\tau)\right)\]
  we can rewrite the right side of the inner implication as, 
  \begin{align*}
    -\eta_{i}(z,\tau) &\geq -\eta_{i}(\tilde{z},\tau) \\
    \eta_{i}(\mathbf{h},\tau) -\eta_{i}(z,\tau) &\geq \eta_{i}(\mathbf{h},\tau) -\eta_{i}(\tilde{z},\tau)
  \end{align*}
  Now unfolding the $\LocLip{\boldsymbol{\eta}}$ with $a = \mathbf{h}$, $r > 0$, $k \geq 0$, 
  $t = \tau$, $x = \mathbf{h}$, and $y = \tilde{z}$, we get from $B_{r}(a)$
  \[\norm{\mathbf{h} - \mathbf{h}}^{2} < r^{2} \Land \norm{\tilde{z} - \mathbf{h}}^{2} < r^{2}\]
  Recall that the vectors $\tilde{z}$ and $\mathbf{h}$ agree on all components except 
  the $i$-th. Hence, we get this relation $\norm{\tilde{z} - h}^{2} = (h_{i} - z_{i})^{2} < r^{2}$. 
  From local Lipschitz we get,
  \[\norm{\boldsymbol{\eta}(\mathbf{h},\tau) - \boldsymbol{\eta}(\tilde{z},\tau)}^{2} \leq k^{2}\norm{\mathbf{h} - \tilde{z}}^{2}\]
  We know the $i$-th component is bounded by the norm,
  \[(\eta_{i}(\mathbf{h},\tau) - \eta_{i}(\tilde{z},\tau))^{2} \leq \norm{\boldsymbol{\eta}(\mathbf{h},\tau) - \boldsymbol{\eta}(\tilde{z},\tau)}^{2}\]
  hence we get the following inequality.
  \begin{align*}
    (\eta_{i}(\mathbf{h},\tau) - \eta_{i}(\tilde{z},\tau))^{2} &\leq k^{2}(h_{i} - z_{i})^{2} \\
    \eta_{i}(\mathbf{h},\tau) - \eta_{i}(\tilde{z},\tau) &\geq -k(h_{i} - z_{i})
  \end{align*}
  Now combining all the relations i.e., $\eta_{i}(\mathbf{h},\tau) -\eta_{i}(z,\tau) \geq \eta_{i}(\mathbf{h},\tau) -\eta_{i}(\tilde{z},\tau)$, 
  and $\eta_{i}(\mathbf{h},\tau) - \eta_{i}(\tilde{z},\tau) \geq -k(h_{i} - z_{i})$, we 
  get the following
  \[k(h_{i} - z_{i}) + \eta_{i}(\mathbf{h},\tau) - \eta_{i}(z,\tau) \geq 0\]
  which is the required inequality. The branch closes with the application of \lref{Real}.
  \begin{sequentproof}
    \ax*[Real]{\LocLip{\boldsymbol{\eta}}, \QuasiMon{\boldsymbol{\eta}}, Q, \tau \geq 0, \dot{\mathbf{h}} - \boldsymbol{\eta}(\mathbf{h},\tau) \geq 0}{\exp(k\tau)\cdot(k(h_{i}-z_{i}) + \dot{h_{i}} - \eta_{i}(z,\tau)) \geq 0}
  \end{sequentproof}
  This finishes the proof of the $i$-th component, which is one of the $n$ open 
  premises. By following this procedure $n$-times all the open premises are closed finishing 
  the proof of axiom \lref{ax:VCP}.
\end{proof}

\subsection{Control Invariance}

All the axioms and proof rules presented in this paper, which enable sound reasoning 
of control invariance in \dL are proved here.

\subsubsection{Comparison Invariants}\label{appendix:comparison-invariants}
\begin{proof}[Proof (Theorem~\ref{thm:scalar-ci})]
  The following abbreviations are used throughout the proof.
  \begin{align*}
    \Gamma &\equiv \LocLip{\eta} \Land \G{\eta} \Land h \succcurlyeq 0, \quad Q_{\tau} \equiv Q \Land \tau \geq 0, \quad Q_{\tau z} \equiv Q \Land \tau \geq 0 \Land z \succcurlyeq 0 \\
    \alpha_{x\tau} &\equiv x' = \mathbf{f}(x,u), \tau' = 1, \qquad \alpha_{x\tau z} \equiv \odes{x' = \mathbf{f}(x,u), \tau' = 1, z' = \eta(z,\tau)} \\
    \G{\eta} &\equiv \forall t \forall x \forall y(\eta(0,t) = 0 \Land  (x \geq y \to \eta(x,t) \geq \eta(y,t)))
  \end{align*}

  \paragraph{(\lref{ax:SCI})} After propositional rewriting the proof starts by resolving the existential quantifier 
  in the succedent using \lref{exR}, such that the resulting right-hand side of the 
  ODE $\mathbf{f}(x,u,\tau)$ is still locally Lipschitz after replacing $u$ with 
  some Skolem function (also a polynomial). After that \lref{ax:DC} and \lref{dI_inequal} 
  add $\tau \geq 0$ to the domain constraint of the succedent.
  \begin{sequentproof}
    \ax[]{\Gamma, \fBox[\ode{\alpha_{x\tau}}[Q_{\tau}]][\exists u(\mathcal{U} \Land \dot{h} + \eta(h,\tau) \geq 0)]}{\fBox[\ode{\alpha_{x\tau}}[Q_{\tau}]][h \succcurlyeq 0]}
    \un[DC,dI_inequal]{\Gamma, \fBox[\ode{\alpha_{x\tau}}[Q]][\exists u(\mathcal{U} \Land \dot{h} + \eta(h,\tau) \geq 0)]}{\fBox[\ode{\alpha_{x\tau}}[Q]][h \succcurlyeq 0]}
    \un[exR]{\Gamma, \fBox[\ode{\alpha_{x\tau}}[Q]][\exists u(\mathcal{U} \Land \dot{h} + \eta(h,\tau) \geq 0)]}{\exists u \fBox[\ode{\alpha_{x\tau}}[Q]][h \succcurlyeq 0]}
  \end{sequentproof}
  From there, axiom \lref{ax:BDG} adds the comparison system $z' = -\eta(z,\tau)$ 
  to the succedent, resulting in two separate premises abbreviated as \textcircled{a} 
  and \textcircled{b}.
  \begin{sequentproof}
    \ax[]{\Gamma, \fBox[\ode{\alpha_{x\tau}}[Q_{\tau}]][\exists u(\mathcal{U} \Land \dot{h} + \eta(h,\tau) \geq 0)]}{\textcircled{a}}
    \ax[]{\Gamma}{\textcircled{b}}
    \bi[ax:BDG]{\Gamma, \fBox[\ode{\alpha_{x\tau}}[Q_{\tau}]][\exists u(\mathcal{U} \Land \dot{h} + \eta(h,\tau) \geq 0)]}{\fBox[\ode{\alpha_{x\tau}}[Q_{\tau}]][h \succcurlyeq 0]}
  \end{sequentproof}

  Continuing with \textcircled{b}. This open premise says we have to show the solution 
  of the introduced comparison system $z' = \eta(z,\tau)$ is bounded from above 
  by the solution of the ODE $\tau' = 1$. Following in similar steps from the 
  proof of axiom \lref{ax:SCP}, we apply axiom \lref{ax:DI_inequal} followed by 
  \lref{()'}.
  \begin{sequentproof}
    \ax[]{\Gamma}{\fBox[\odes{x' = \mathbf{f}(x,u), \tau' = 1, z' = \eta(z,\tau)}[Q_{\tau}]][2zz' \leq c\tau']}
    \un[ax:DI_inequal,()']{\Gamma}{\fBox[\odes{x' = \mathbf{f}(x,u), \tau' = 1, z' = \eta(z,\tau)}[Q_{\tau}]][\norm{z}^{2} \leq c\tau]}
  \end{sequentproof}
  After that with the applications of \lref{DE}, \lref{[:=]}, and \lref{dW}, we 
  are left with proving the real arithmetic inequality $2z\cdot\eta(z,\tau) \leq c$.
  \begin{sequentproof}
    \ax[]{\LocLip{\eta}, \G{\eta}, Q, \tau \geq 0}{-2z\cdot\eta(z,\tau) \leq c}
    \un[dW]{\Gamma}{\fBox[\odes{x' = \mathbf{f}(x,u), \tau' = 1, z' = \eta(z,\tau)}[Q_{\tau}]][2z\cdot\eta(z,\tau) \leq c]}
    \un[DE,[:=]]{\Gamma}{\fBox[\odes{x' = \mathbf{f}(x,u), \tau' = 1, z' = \eta(z,\tau)}[Q_{\tau}]][2zz' \leq c\tau']}
  \end{sequentproof}
  In order to prove the resulting real arithmetic, we first have to unfold the definition 
  of \LocLip{$\eta$} as follows.
  \[\forall \delta \forall a \exists r \exists k \forall t \forall x \forall y (t \leq \delta^{2} \Land B_{r}(a) \to \norm{\eta(x,t) - \eta(y,t)}^{2} \leq k^{2}\norm{x - y}^{2})\]
  Here $B_{r}(a) \equiv t \geq 0 \Land \norm{x - a}^{2} < r^{2} \Land \norm{y - a}^{2} < r^{2}$, 
  similarly we unfold the definition of $\G{\eta}$ as following.
  \[\forall t \forall x \forall y(\eta(0,t) = 0 \Land  (x \leq y \to \eta(x,t) \leq \eta(y,t)))\]
  Next instantiating the quantifiers as $\delta \geq 0$, $a = 0$, $r > 0$, $k \geq 0$, 
  $t = \tau$, $x = z$, and $y = 0$ we get the following from the definition of local 
  Lipschitz.
  \[\tau \leq \delta^{2} \Land \norm{z}^{2} < r^{2} \to \norm{\eta(z,\tau) - \eta(0,\tau)}^{2} \leq k^{2}\norm{z}^{2}\]
  Now from the definition of $\G{\eta}$ we know $\eta(0,\tau) = 0$, which implies 
  the following.
  \[\tau \leq \delta^{2} \Land \norm{z}^{2} < r^{2} \to \norm{\eta(z,\tau)}^{2} \leq k^{2}\norm{z}^{2}\]
  The left of the implication is provable in real arithmetic. The terms right of 
  the implication can be rewritten without norms as $z$ and $\eta(\cdot)$ are scalars.
  \[\eta(z,\tau)^{2} \leq k^{2}z^{2}\]
  Now after expanding and then rearranging the terms in the following inequality we 
  get.
  \begin{align*}
    0 &\leq (z - \eta(z,\tau))^{2} \\
    2z\cdot\eta(z,\tau) &\leq z^{2} + \eta(z,\tau)^{2}
  \end{align*}
  We also already know that $\norm{z}^{2} < r^{2}$, rewriting it equivalently without 
  the norm we get $z^{2} < r^{2}$. With that we have our upper bounds in $z^{2}$ and 
  $\eta(z,\tau)^{2}$, which are combined as follows to get the desired upper bound.
  \begin{align*}
    2z\cdot\eta(z,\tau) &\leq z^{2} + \eta(z,\tau)^{2} \\
    &\leq z^{2} + k^{2}z^{2} \\
    &\leq r^{2}(1 + k^{2})
  \end{align*}
  The branch closes by cutting in $c = r^{2}(1 + k^{2})$ followed by \lref{Real}.
  \begin{sequentproof}
    \ax*[Real]{\LocLip{\eta}, \G{\eta}, Q, \tau \geq 0, c = r^{2}(1 + k^{2})}{2z\cdot\eta(z,\tau) \leq c}
    \un[cut,exL]{\LocLip{\eta}, \G{\eta}, Q, \tau \geq 0}{2z\cdot\eta(z,\tau) \leq c}
  \end{sequentproof}

  Now continuing with the open left premise \textcircled{a} from the application of \lref{ax:BDG}. 
  We first add $z \succcurlyeq 0$ ($\succcurlyeq$ is either $\geq$ or $>$) to the domain constraint of the succedent 
  using \lref{ax:DC}, resulting in two premises \textcircled{c} and \textcircled{d}.
  \begin{sequentproof}
    \ax[]{\Gamma, \fBox[\ode{\alpha_{x\tau}}[Q_{\tau}]][\exists u(\mathcal{U} \Land \dot{h} + \eta(h,\tau) \geq 0)]}{\textcircled{c}}
    \ax[]{\Gamma, \fBox[\ode{\alpha_{x\tau}}[Q_{\tau}]][\exists u(\mathcal{U} \Land \dot{h} + \eta(h,\tau) \geq 0)]}{\textcircled{d}}
    \bi[ax:DC]{\Gamma, \fBox[\ode{\alpha_{x\tau}}[Q_{\tau}]][\exists u(\mathcal{U} \Land \dot{h} + \eta(h,\tau) \geq 0)]}{\fBox[\ode{\alpha_{x\tau z}}[Q_{\tau}]][h \succcurlyeq 0]}
  \end{sequentproof}

  Bringing our attention to the open premise \textcircled{c}. An application of 
  \lref{pr:M[']} monotonically rewrites the postcondition in the succedent to 
  $h \geq z$. The resulting premise $Q, \tau \geq 0, z \succcurlyeq 0, h \geq z \vdash h \succcurlyeq 0$ 
  is provable in real arithmetic.
  \begin{sequentproof}
    \ax[]{\Gamma, \fBox[\ode{\alpha_{x\tau}}[Q_{\tau}]][\exists u(\mathcal{U} \Land \dot{h} + \eta(h,\tau) \geq 0)]}{\fBox[\ode{\alpha_{x\tau z}}[Q_{\tau z}]][h \geq z]}
    \un[pr:M['],Real]{\Gamma, \fBox[\ode{\alpha_{x\tau}}[Q_{\tau}]][\exists u(\mathcal{U} \Land \dot{h} + \eta(h,\tau) \geq 0)]}{\fBox[\ode{\alpha_{x\tau z}}[Q_{\tau z}]][h \succcurlyeq 0]}
  \end{sequentproof}
  From there, the formula $h \geq z$ is cut into the antecedent by first cutting in 
  $\exists z(z = h)$, which combined with $h \succcurlyeq 0$ 
  implies $h \geq z$. After that \lref{pr:M[']} rewrites the postcondition in the 
  antecedent to $\mathcal{U} \Land \dot{h} + \eta(h,\tau) \geq 0$, which is provable 
  by resolving the existential quantifier.
  \begin{sequentproof}
    \ax[]{\Gamma, h \geq z, \fBox[\ode{\alpha_{x\tau}}[Q_{\tau}]][(\mathcal{U} \Land \dot{h} + \eta(h,\tau) \geq 0)]}{\fBox[\ode{\alpha_{x\tau z}}[Q_{\tau z}]][h \geq z]}
    \un[pr:M[']]{\Gamma, h \geq z, \fBox[\ode{\alpha_{x\tau}}[Q_{\tau}]][\exists u(\mathcal{U} \Land \dot{h} + \eta(h,\tau) \geq 0)]}{\fBox[\ode{\alpha_{x\tau z}}[Q_{\tau z}]][h \geq z]}
    \un[cut,Real]{\Gamma, \fBox[\ode{\alpha_{x\tau}}[Q_{\tau}]][\exists u(\mathcal{U} \Land \dot{h} + \eta(h,\tau) \geq 0)]}{\fBox[\ode{\alpha_{x\tau z}}[Q_{\tau z}]][h \geq z]}
  \end{sequentproof}
  Next the conjunction in the postcondition of the antecedent is decomposed using 
  \lref{ax:box_and}. The formula $\fBox[\ode{\alpha_{x\tau}}[Q_{\tau}]][\mathcal{U}]$ 
  is dropped from the context by weakening as it is not required from here on. 
  After that an application of \lref{ax:DC} adds $z \succcurlyeq 0$ to the domain 
  constraint of the antecedent to match the succedent, resulting in two open premises 
  labeled \textcircled{e} and \textcircled{f}.
  \begin{sequentproof}
    \ax[]{\Gamma, h \geq z, \fBox[\ode{\alpha_{x\tau}}[Q_{\tau z}]][\dot{h} + \eta(h,\tau) \geq 0]}{\textcircled{e}}
    \ax[]{\Gamma, h \geq z}{\textcircled{f}}
    \bi[ax:DC]{\Gamma, h \geq z, \fBox[\ode{\alpha_{x\tau}}[Q_{\tau}]][\mathcal{U}], \fBox[\ode{\alpha_{x\tau}}[Q_{\tau}]][\dot{h} + \eta(h,\tau) \geq 0]}{\fBox[\ode{\alpha_{x\tau z}}[Q_{\tau z}]][h \geq z]}
    \un[ax:box_and]{\Gamma, h \geq z, \fBox[\ode{\alpha_{x\tau}}[Q_{\tau}]][(\mathcal{U} \Land \dot{h} + \eta(h,\tau) \geq 0)]}{\fBox[\ode{\alpha_{x\tau z}}[Q_{\tau z}]][h \geq z]}
  \end{sequentproof}
  
  Continuing with the left premise \textcircled{e}. The branch is closed by a straightforward 
  application of the scalar comparison principle axiom \lref{ax:SCP}.
  \begin{sequentproof}
    \ax*[ax:SCP]{\LocLip{\eta}, \G{\eta}, h \geq z, \fBox[\ode{\alpha_{x\tau}}[Q_{\tau z}]][\dot{h} + \eta(h,\tau) \geq 0]}{\fBox[\ode{\alpha_{x\tau z}}[Q_{\tau z}]][h \geq z]}
    \un[]{\Gamma, h \geq z, \fBox[\ode{\alpha_{x\tau}}[Q_{\tau z}]][\dot{h} + \eta(h,\tau) \geq 0]}{\fBox[\ode{\alpha_{x\tau z}}[Q_{\tau z}]][h \geq z]}
  \end{sequentproof}

  Getting back to the open right premise \textcircled{d}, \lref{pr:M[']} rewrites 
  the postcondition of the box modality in the succedent to $\exp(k\tau)\cdot z \succcurlyeq 0$. Here 
  $\exp(\cdot)$ is the real exponential function, $k$ is the Lipschitz constant, 
  and $\tau$ is the time variable. After that an application of \lref{cut} and \lref{Real} 
  adds $\exists z(h = z)$ to the antecedent.
  \begin{sequentproof}[align]
    \ax[]{\Gamma, \exists z(h = z), \fBox[\ode{\alpha_{x\tau}}[Q_{\tau}]][\exists u(\mathcal{U} \Land \dot{h} + \eta(h,\tau) \geq 0)]}{\fBox[\ode{\alpha_{x\tau z}}[Q_{\tau}]][\exp(k\tau)\cdot z \succcurlyeq 0]}
    \un[cut,Real]{\Gamma, \fBox[\ode{\alpha_{x\tau}}[Q_{\tau}]][\exists u(\mathcal{U} \Land \dot{h} + \eta(h,\tau) \geq 0)]}{\fBox[\ode{\alpha_{x\tau z}}[Q_{\tau}]][\exp(k\tau)\cdot z \succcurlyeq 0]}
    \un[pr:M[']]{\Gamma, \fBox[\ode{\alpha_{x\tau}}[Q_{\tau}]][\exists u(\mathcal{U} \Land \dot{h} + \eta(h,\tau) \geq 0)]}{\fBox[\ode{\alpha_{x\tau z}}[Q_{\tau}]][z \succcurlyeq 0]}
  \end{sequentproof}
  Following from there, another \lref{cut} adds $\exp(k\tau)\cdot z \succcurlyeq 0$ 
  to the antecedent. Now we can apply \lref{ax:DI_inequal} followed by \lref{()'}.
  \begin{sequentproof}[align]
    \ax[]{\Gamma, \exists z(h = z)}{\fBox[\ode{\alpha_{x\tau z}}[Q_{\tau}]][\exp(k\tau)\cdot (k\tau'z + z') \geq 0]}
    \un[()']{\Gamma, \exists z(h = z)}{\fBox[\ode{\alpha_{x\tau z}}[Q_{\tau}]][(\exp(k\tau)\cdot z)' \geq 0]}
    \un[ax:DI_inequal]{\Gamma, \exists z(h = z), \exp(k\tau)\cdot z \succcurlyeq 0}{\fBox[\ode{\alpha_{x\tau z}}[Q_{\tau}]][\exp(k\tau)\cdot z \succcurlyeq 0]}
    \un[cut,Real]{\Gamma, \exists z(h = z)}{\fBox[\ode{\alpha_{x\tau z}}[Q_{\tau}]][\exp(k\tau)\cdot z \succcurlyeq 0]}
  \end{sequentproof}
  Now after applying axioms \lref{DE}, \lref{[:=]}, and then weakening \lref{dW}, 
  we are left proving the resulting real arithmetic. Again all the quantified formulas 
  are soundly kept into context.
  \begin{sequentproof}
    \ax[]{\LocLip{\eta}, \G{\eta}, \exists z(h = z), Q, \tau \geq 0}{\exp(k\tau)\cdot (kz + \eta(z,\tau)) \geq 0}
    \un[dW]{\Gamma, \exists z(h = z)}{\fBox[\ode{\alpha_{x\tau z}}[Q_{\tau}]][\exp(k\tau)\cdot (kz + \eta(z,\tau)) \geq 0]}
    \un[DE,[:=]]{\Gamma, \exists z(h = z)}{\fBox[\ode{\alpha_{x\tau z}}[Q_{\tau}]][\exp(k\tau)\cdot (k\tau'z + z') \geq 0]}
  \end{sequentproof}
  We know $\exp(k\eta) \geq 0$, so it suffices to prove $kz + \eta(z,\tau) \geq 0$.
  Now unfolding $\LocLip{\eta}$ and $\G{\eta}$ with the quantifiers instantiated as 
  $\delta \geq 0$, $a = 0$, $r > 0$, $k \geq 0$, 
  $t = \tau$, $x = z$, and $y = 0$ we get the following implication.
  \[\tau \leq \delta^{2} \Land \norm{z}^{2} < r^{2} \to \norm{\eta(z,\tau) - \eta(0,\tau)}^{2} \leq k^{2}\norm{z}^{2}\]
  From the definition of $\G{\eta}$ we know $\eta(0,\tau) = 0$, which implies 
  the following.
  \[\tau \leq \delta^{2} \Land \norm{z}^{2} < r^{2} \to \norm{\eta(z,\tau)}^{2} \leq k^{2}\norm{z}^{2}\]
  The terms right of the implication can be rewritten without the squared norms as $z$ and $\eta(\cdot)$ are scalars.
  \begin{align*}
    \eta(z,\tau) &\geq -kz \\
    kz + \eta(z,\tau) &\geq 0
  \end{align*}
  The branch closes with an application of \lref{Real}.
  \begin{sequentproof}
    \ax*[Real]{\LocLip{\eta}, \G{\eta}, \exists z(h = z)}{\exp(k\tau)\cdot (kz + \eta(z,\tau)) \geq 0}
  \end{sequentproof}

  Finally, coming back to the open premise \textcircled{f}, the proof of this branch 
  follows similarly to \textcircled{d}. First, \lref{pr:M[']} introduces the new 
  invariant $\exp(k\tau)\cdot z \succcurlyeq 0$, where $\exp(\cdot)$ is the real 
  exponential function. Then an application of \lref{ax:DI_inequal} followed by 
  \lref{dW} results in real arithmetic proof obligations.
  \begin{sequentproof}
    \ax[]{\LocLip{\eta}, \G{\eta}, Q, \tau \geq 0}{\exp(k\tau)\cdot (kz + \eta(z,\tau)) \geq 0}
    \un[ax:DI_inequal,dW]{\Gamma, h \geq z, \exp(k\tau)\cdot z \succcurlyeq 0}{\fBox[\ode{\alpha_{x\tau z}}[Q_{\tau}]][\exp(k\tau)\cdot z \succcurlyeq 0]}
    \un[pr:M['],cut]{\Gamma, h \geq z}{\fBox[\ode{\alpha_{x\tau z}}[Q_{\tau}]][z \succcurlyeq 0]}
  \end{sequentproof}
  Similar as before proving $kz - \eta(z,\tau) \geq 0$ suffices to prove the given 
  inequality. Unfolding $\LocLip{\eta}$ and $\G{\eta}$ with the quantifiers 
  instantiated as $\delta \geq 0$, $a = 0$, $r > 0$, $k \geq 0$, 
  $t = \tau$, $x = z$, and $y = 0$ we get the following implication.
  \[\tau \leq \delta^{2} \Land \norm{z}^{2} < r^{2} \to \norm{\eta(z,\tau) - \eta(0,\tau)}^{2} \leq k^{2}\norm{z}^{2}\]
  Then combining with the definition of $\G{\eta}$, we get the desired inequality 
  $kz + \eta(z,\tau) \geq 0$.
  \begin{sequentproof}
    \ax*[Real]{\LocLip{\eta}, \G{\eta}, Q, \tau \geq 0}{\exp(k\tau)\cdot (kz + \eta(z,\tau)) \geq 0}
  \end{sequentproof}
  The branch is closed by an application of \lref{Real} finishing the proof of the 
  scalar comparison invariants axiom \lref{ax:SCI}.

  \paragraph{(\lref{pr:sCI})} For the sake of completeness we give the derivation 
  of the proof rule, as presented in the main text in the proof sketch of Theorem~\ref{thm:scalar-ci} 
  here again. Proof rule \lref{pr:sCI} is also derived syntactically as follows.
  \begin{sequentproof}
    \ax[]{}{\LocLip{\eta} \Land \G{\eta}}
    \ax[]{Q}{\exists u (\mathcal{U} \Land \dot{h} + \eta(h,\tau) \geq 0)}
    \bi[andR,dW]{}{\LocLip{\eta} \Land \G{\eta} \Land \fBox[\odes{x' = \mathbf{f}(x,u,\tau), \tau' = 1}[Q]][\exists u (\mathcal{U} \Land \dot{h} + \eta(h,\tau) \geq 0)]}
    \un[ax:SCI]{h \succcurlyeq 0}{\exists u\fBox[\odes{x' = \mathbf{f}(x,u,\tau), \tau' = 1}[Q]][h \succcurlyeq 0]}
  \end{sequentproof}
  The proof utilizes axiom \lref{ax:SCI}, followed by \lref{andR} and \lref{dW} resulting 
  in the required premises.
\end{proof}

\begin{proof}[Proof (Theorem~\ref{thm:vector-ci})]\label{proof:vci}
   The following abbreviations are used throughout the proof.
  \begin{align*}
    \Gamma &\equiv \LocLip{\boldsymbol{\eta}} \Land \vecG{\boldsymbol{\eta}} \Land \mathbf{h} \succcurlyeq 0, \quad Q_{\tau} \equiv Q \Land \tau \geq 0, \quad Q_{\tau z} \equiv Q \Land \tau \geq 0 \Land z \succcurlyeq 0 \\
    \alpha_{x\tau} &\equiv x' = \mathbf{f}(x,u), \tau' = 1, \qquad \alpha_{x\tau z} \equiv \odes{x' = \mathbf{f}(x,u), \tau' = 1, z' = \boldsymbol{\eta}(z,\tau)} \\
    \G{\eta} &\equiv \forall t \forall x \forall y(\eta(0,t) = 0 \Land  (x \geq y \to \eta(x,t) \geq \eta(y,t))) \qquad \vecG{\boldsymbol{\eta}} \defequiv \bigwedge^{n}_{i = 1} \G{\eta_{i}}
  \end{align*}

  \paragraph{(\lref{ax:VCI})} After some propositional rewriting the proof begins by resolving the existential quantifier 
  in the succedent using \lref{exR}, such that the resulting right-hand side of the 
  ODE $x' = \mathbf{f}(x,u,\tau)$ is still locally Lipschitz after replacing $u$ with 
  some Skolem function (also a polynomial). After that \lref{ax:DC} and \lref{dI_inequal} 
  add $\tau \geq 0$ to the domain constraint of the succedent.
  \begin{sequentproof}
    \ax[]{\Gamma, \fBox[\ode{\alpha_{x\tau}}[Q_{\tau}]][\exists u(\mathcal{U} \Land \dot{\mathbf{h}} + \boldsymbol{\eta}(\mathbf{h},\tau) \geq 0)]}{\fBox[\ode{\alpha_{x\tau}}[Q_{\tau}]][\mathbf{h} \succcurlyeq 0]}
    \un[DC,dI_inequal]{\Gamma, \fBox[\ode{\alpha_{x\tau}}[Q]][\exists u(\mathcal{U} \Land \dot{\mathbf{h}} + \boldsymbol{\eta}(\mathbf{h},\tau) \geq 0)]}{\fBox[\ode{\alpha_{x\tau}}[Q]][\mathbf{h} \succcurlyeq 0]}
    \un[exR]{\Gamma, \fBox[\ode{\alpha_{x\tau}}[Q]][\exists u(\mathcal{U} \Land \dot{\mathbf{h}} + \boldsymbol{\eta}(\mathbf{h},\tau) \geq 0)]}{\exists u \fBox[\ode{\alpha_{x\tau}}[Q]][\mathbf{h} \succcurlyeq 0]}
  \end{sequentproof}
  From there, axiom \lref{ax:BDG} adds the comparison system $z' = \boldsymbol{\eta}(z,\tau)$ 
  to the succedent, resulting in two separate premises abbreviated as \textcircled{a} 
  and \textcircled{b}.
  \begin{sequentproof}
    \ax[]{\Gamma, \fBox[\ode{\alpha_{x\tau}}[Q_{\tau}]][\exists u(\mathcal{U} \Land \dot{\mathbf{h}} + \boldsymbol{\eta}(\mathbf{h},\tau) \geq 0)]}{\textcircled{a}}
    \ax[]{\Gamma}{\textcircled{b}}
    \bi[ax:BDG]{\Gamma, \fBox[\ode{\alpha_{x\tau}}[Q_{\tau}]][\exists u(\mathcal{U} \Land \dot{\mathbf{h}} + \boldsymbol{\eta}(\mathbf{h},\tau) \geq 0)]}{\fBox[\ode{\alpha_{x\tau}}[Q_{\tau}]][\mathbf{h} \succcurlyeq 0]}
  \end{sequentproof}

  Continuing with \textcircled{b}. This premise says the solution 
  of the introduced comparison system $z' = \boldsymbol{\eta}(z,\tau)$ is bounded from above 
  by the solution of the ODE $\tau' = 1$. Following in similar steps from the 
  proof of axiom \lref{ax:VCP}, we apply axiom \lref{ax:DI_inequal} followed by 
  \lref{()'}.
  \begin{sequentproof}
    \ax[]{\Gamma}{\fBox[\odes{x' = \mathbf{f}(x,u), \tau' = 1, z' = \boldsymbol{\eta}(z,\tau)}[Q_{\tau}]][2z\cdot z' \leq c\tau']}
    \un[ax:DI_inequal,()']{\Gamma}{\fBox[\odes{x' = \mathbf{f}(x,u), \tau' = 1, z' = \boldsymbol{\eta}(z,\tau)}[Q_{\tau}]][\norm{z}^{2} \leq c\tau]}
  \end{sequentproof}
  After that with the applications of \lref{DE}, \lref{[:=]}, and \lref{dW}, we 
  are left with proving the real arithmetic inequality $-2z\cdot\boldsymbol{\eta}(z,\tau) \leq c$.
  \begin{sequentproof}
    \ax[]{\LocLip{\boldsymbol{\eta}}, \vecG{\boldsymbol{\eta}}, Q, \tau \geq 0}{2z\cdot \boldsymbol{\eta}(z,\tau) \leq c}
    \un[dW]{\Gamma}{\fBox[\odes{x' = \mathbf{f}(x,u), \tau' = 1, z' = \boldsymbol{\eta}(z,\tau)}[Q_{\tau}]][2z\cdot \boldsymbol{\eta}(z,\tau) \leq c]}
    \un[DE,[:=]]{\Gamma}{\fBox[\odes{x' = \mathbf{f}(x,u), \tau' = 1, z' = \boldsymbol{\eta}(z,\tau)}[Q_{\tau}]][2z\cdot z' \leq c\tau']}
  \end{sequentproof}
  In order to prove the resulting real arithmetic, we first have to unfold the definition 
  of \LocLip{$\boldsymbol{\eta}$} 
  with $B_{r}(a) \equiv t \geq 0 \Land \norm{x - a}^{2} < r^{2} \Land \norm{y - a}^{2} < r^{2}$ as follows.
  \[\forall \delta \forall a \exists r \exists k \forall t \forall x \forall y (t \leq \delta^{2} \Land B_{r}(a) \to \norm{\boldsymbol{\eta}(x,t) - \boldsymbol{\eta}(y,t)}^{2} \leq k^{2}\norm{x - y}^{2})\]
  Now instantiating the quantifiers as, $\delta \geq 0$, $a = 0$, $r > 0$, $k \geq 0$, 
  $t = \tau$, $x = z$, and $y = 0$, we get
  \[\tau \leq \delta^{2} \Land \norm{z}^{2} < r^{2} \to \norm{\boldsymbol{\eta}(z,\tau) - \boldsymbol{\eta}(0, \tau)}^{2} \leq k^{2} \norm{z}^{2}\]
  The premise on the left of the implication is provable in real arithmetic. Now we 
  unfold $\vecG{\boldsymbol{\eta}}$ as $\bigwedge^{n}_{i = 1} \G{\eta_{i}}$, and from 
  $\G{\eta_{i}}$ we know $\eta_{i}(0,\tau) = 0$, hence $\boldsymbol{\eta}(0,\tau) = \mathbf{0}$, 
  where $\mathbf{0}$ is the zero vector. Following from local Lipschitz we get the following.
  \[\norm{\boldsymbol{\eta}(z,\tau)}^{2} \leq k^{2}\norm{z}^{2}\]
  Now using 
  Young's inequality i.e., $ab \leq \frac{a^{p}}{p} + \frac{b^{q}}{q}$, for $p = q = 2$, we get the 
  following relation.
  \[2\norm{z} \norm{\boldsymbol{\eta}(z,\tau)} \leq \norm{z}^{2} + \norm{\boldsymbol{\eta}(z,\tau)}^{2}\]
  We already know $\norm{z}^{2} \leq r^{2}$, combining this with the above inequalities we get,
  \begin{align*}
    2\norm{z} \norm{\boldsymbol{\eta}(z,\tau)} &\leq \norm{z}^{2} + k^{2}\norm{z}^{2}\\
    &\leq r^{2}(1 + k^{2})
  \end{align*}
  Now using the Cauchy-Schwarz inequality i.e., $|v \cdot w| \leq \norm{v}\norm{w}$, we get an upper bound 
  on our required inequality.
  \[2z \cdot \boldsymbol{\eta}(z,\tau) \leq 2\norm{z} \norm{\boldsymbol{\eta}(z,\tau)} \leq r^{2}(1 + k^{2})\]
  By cutting in $c = r^{2}(1 + k^{2})$ and instantiating 
  all the quantifiers as showed above the branch closes with \lref{Real}.
  \begin{sequentproof}
    \ax*[Real]{\LocLip{\boldsymbol{\eta}}, \QuasiMon{\boldsymbol{\eta}}, Q, \tau \geq 0, c = r^{2}(1 + k^{2})}{2z\cdot \boldsymbol{\eta}(z,\tau) \leq c}
    \un[cut,Real]{\LocLip{\boldsymbol{\eta}}, \QuasiMon{\boldsymbol{\eta}}, Q, \tau \geq 0}{2z\cdot \boldsymbol{\eta}(z,\tau) \leq c}
  \end{sequentproof}
  Now continuing with the open left premise \textcircled{a} from the application of \lref{ax:BDG}. 
  We first add $z \succcurlyeq 0$ ($\succcurlyeq$ is either $\geq$ or $>$) to the domain constraint of the succedent 
  using \lref{ax:DC}, resulting in two premises \textcircled{c} and \textcircled{d}.
  \begin{sequentproof}
    \ax[]{\Gamma, \fBox[\ode{\alpha_{x\tau}}[Q_{\tau}]][\exists u(\mathcal{U} \Land \dot{\mathbf{h}} + \boldsymbol{\eta}(\mathbf{h},\tau) \geq 0)]}{\textcircled{c}}
    \ax[]{\Gamma, \fBox[\ode{\alpha_{x\tau}}[Q_{\tau}]][\exists u(\mathcal{U} \Land \dot{\mathbf{h}} + \boldsymbol{\eta}(\mathbf{h},\tau) \geq 0)]}{\textcircled{d}}
    \bi[ax:DC]{\Gamma, \fBox[\ode{\alpha_{x\tau}}[Q_{\tau}]][\exists u(\mathcal{U} \Land \dot{\mathbf{h}} + \boldsymbol{\eta}(\mathbf{h},\tau) \geq 0)]}{\fBox[\ode{\alpha_{x\tau z}}[Q_{\tau}]][\mathbf{h} \succcurlyeq 0]}
  \end{sequentproof}

  Bringing our attention to the open premise \textcircled{c}. An application of 
  \lref{pr:M[']} monotonically rewrites the postcondition in the succedent to 
  $\mathbf{h} \geq z$. The resulting premise $Q, \tau \geq 0, z \succcurlyeq 0, \mathbf{h} \geq z \vdash \mathbf{h} \succcurlyeq 0$ 
  is provable in real arithmetic.
  \begin{sequentproof}
    \ax[]{\Gamma, \fBox[\ode{\alpha_{x\tau}}[Q_{\tau}]][\exists u(\mathcal{U} \Land \dot{\mathbf{h}} + \boldsymbol{\eta}(\mathbf{h},\tau) \geq 0)]}{\fBox[\ode{\alpha_{x\tau z}}[Q_{\tau z}]][\mathbf{h} \geq z]}
    \un[pr:M['],Real]{\Gamma, \fBox[\ode{\alpha_{x\tau}}[Q_{\tau}]][\exists u(\mathcal{U} \Land \dot{\mathbf{h}} + \boldsymbol{\eta}(\mathbf{h},\tau) \geq 0)]}{\fBox[\ode{\alpha_{x\tau z}}[Q_{\tau z}]][\mathbf{h} \succcurlyeq 0]}
  \end{sequentproof}
  From there, the formula $\mathbf{h} \geq z$ is cut into the antecedent by first cutting in 
  $\exists z(z = \mathbf{h})$, which combined with $\mathbf{h} \succcurlyeq 0$ 
  implies $\mathbf{h} \geq z$. After which \lref{pr:M[']} rewrites the postcondition in the 
  antecedent to $\mathcal{U} \Land \dot{\mathbf{h}} + \boldsymbol{\eta}(\mathbf{h},\tau) \geq 0$, which is provable 
  by resolving the existential quantifier.
  \begin{sequentproof}
    \ax[]{\Gamma, \mathbf{h} \geq z, \fBox[\ode{\alpha_{x\tau}}[Q_{\tau}]][(\mathcal{U} \Land \dot{\mathbf{h}} + \boldsymbol{\eta}(\mathbf{h},\tau) \geq 0)]}{\fBox[\ode{\alpha_{x\tau z}}[Q_{\tau z}]][\mathbf{h} \geq z]}
    \un[pr:M[']]{\Gamma, \mathbf{h} \geq z, \fBox[\ode{\alpha_{x\tau}}[Q_{\tau}]][\exists u(\mathcal{U} \Land \dot{\mathbf{h}} + \boldsymbol{\eta}(\mathbf{h},\tau) \geq 0)]}{\fBox[\ode{\alpha_{x\tau z}}[Q_{\tau z}]][\mathbf{h} \geq z]}
    \un[cut,Real]{\Gamma, \fBox[\ode{\alpha_{x\tau}}[Q_{\tau}]][\exists u(\mathcal{U} \Land \dot{\mathbf{h}} + \boldsymbol{\eta}(\mathbf{h},\tau) \geq 0)]}{\fBox[\ode{\alpha_{x\tau z}}[Q_{\tau z}]][\mathbf{h} \geq z]}
  \end{sequentproof}
  Next the conjunction in the postcondition of the antecedent is decomposed using 
  \lref{ax:box_and}. The formula $\fBox[\ode{\alpha_{x\tau}}[Q_{\tau}]][\mathcal{U}]$ 
  is dropped from the context by weakening as it is not required from here on. 
  After that an application of \lref{ax:DC} adds $z \succcurlyeq 0$ to the domain 
  constraint of the antecedent to match the succedent, resulting in two open premises 
  labeled \textcircled{e} and \textcircled{f}.
  \begin{sequentproof}
    \ax[]{\Gamma, \mathbf{h} \geq z, \fBox[\ode{\alpha_{x\tau}}[Q_{\tau z}]][\dot{\mathbf{h}} + \boldsymbol{\eta}(\mathbf{h},\tau) \geq 0]}{\textcircled{e}}
    \ax[]{\Gamma, \mathbf{h} \geq z}{\textcircled{f}}
    \bi[ax:DC]{\Gamma, \mathbf{h} \geq z, \fBox[\ode{\alpha_{x\tau}}[Q_{\tau}]][\mathcal{U}], \fBox[\ode{\alpha_{x\tau}}[Q_{\tau}]][\dot{\mathbf{h}} + \boldsymbol{\eta}(\mathbf{h},\tau) \geq 0]}{\fBox[\ode{\alpha_{x\tau z}}[Q_{\tau z}]][\mathbf{h} \geq z]}
    \un[ax:box_and]{\Gamma, \mathbf{h} \geq z, \fBox[\ode{\alpha_{x\tau}}[Q_{\tau}]][(\mathcal{U} \Land \dot{\mathbf{h}} + \boldsymbol{\eta}(\mathbf{h},\tau) \geq 0)]}{\fBox[\ode{\alpha_{x\tau z}}[Q_{\tau z}]][\mathbf{h} \geq z]}
  \end{sequentproof}
  
  Continuing with the left premise \textcircled{e}. The branch is closed after applying 
  the vector comparison principle axiom \lref{ax:VCP}. Observe the fact that if a 
  function $\boldsymbol{\eta}(\cdot)$ is class $\mathcal{G}$ then it implies $\boldsymbol{\eta}(\cdot)$ 
  is also quasimonotone increasing, hence proving the resulting premise $\vecG{\boldsymbol{\eta}} \vdash \QuasiMon{\boldsymbol{\eta}}$.
  \begin{sequentproof}
    \ax*[ax:VCP]{\LocLip{\eta}, \vecG{\eta}, \mathbf{h} \geq z, \fBox[\ode{\alpha_{x\tau}}[Q_{\tau z}]][\dot{\mathbf{h}} + \boldsymbol{\eta}(\mathbf{h},\tau) \geq 0]}{\fBox[\ode{\alpha_{x\tau z}}[Q_{\tau z}]][\mathbf{h} \geq z]}
    \un[]{\Gamma, \mathbf{h} \geq z, \fBox[\ode{\alpha_{x\tau}}[Q_{\tau z}]][\dot{\mathbf{h}} + \boldsymbol{\eta}(\mathbf{h},\tau) \geq 0]}{\fBox[\ode{\alpha_{x\tau z}}[Q_{\tau z}]][\mathbf{h} \geq z]}
  \end{sequentproof}

  Getting back to the open right premise \textcircled{d}, \lref{pr:M[']} rewrites 
  the postcondition of the box modality in the succedent to $\exp(k\tau)\cdot z \succcurlyeq 0$. Here 
  $\exp(\cdot)$ is the real exponential function, $k$ is the Lipschitz constant, 
  and $\tau$ is the time variable. After that an application of \lref{cut} and \lref{Real} 
  adds $\exists z(\mathbf{h} = z)$ to the antecedent.
  \begin{sequentproof}[align]
    \ax[]{\Gamma, \exists z(\mathbf{h} = z), \fBox[\ode{\alpha_{x\tau}}[Q_{\tau}]][\exists u(\mathcal{U} \Land \dot{\mathbf{h}} + \boldsymbol{\eta}(\mathbf{h},\tau) \geq 0)]}{\fBox[\ode{\alpha_{x\tau z}}[Q_{\tau}]][\exp(k\tau)\cdot z \succcurlyeq 0]}
    \un[cut,Real]{\Gamma, \fBox[\ode{\alpha_{x\tau}}[Q_{\tau}]][\exists u(\mathcal{U} \Land \dot{\mathbf{h}} + \boldsymbol{\eta}(\mathbf{h},\tau) \geq 0)]}{\fBox[\ode{\alpha_{x\tau z}}[Q_{\tau}]][\exp(k\tau)\cdot z \succcurlyeq 0]}
    \un[pr:M[']]{\Gamma, \fBox[\ode{\alpha_{x\tau}}[Q_{\tau}]][\exists u(\mathcal{U} \Land \dot{\mathbf{h}} + \boldsymbol{\eta}(\mathbf{h},\tau) \geq 0)]}{\fBox[\ode{\alpha_{x\tau z}}[Q_{\tau}]][z \succcurlyeq 0]}
  \end{sequentproof}
  Next a \lref{cut} adds $\exp(k\tau)\cdot z \succcurlyeq 0$ to the antecedent.
  \begin{sequentproof}
    \ax[]{\Gamma, \exp(k\tau)\cdot z \succcurlyeq 0, \fBox[\ode{\alpha_{x\tau}}[Q_{\tau}]][\exists u(\mathcal{U} \Land \dot{\mathbf{h}} + \boldsymbol{\eta}(\mathbf{h},\tau) \geq 0)]}{\fBox[\ode{\alpha_{x\tau z}}[Q_{\tau}]][\exp(k\tau)\cdot z \succcurlyeq 0]}
    \un[cut,Real]{\Gamma, \exists z(\mathbf{h} = z), \fBox[\ode{\alpha_{x\tau}}[Q_{\tau}]][\exists u(\mathcal{U} \Land \dot{\mathbf{h}} + \boldsymbol{\eta}(\mathbf{h},\tau) \geq 0)]}{\fBox[\ode{\alpha_{x\tau z}}[Q_{\tau}]][\exp(k\tau)\cdot z \succcurlyeq 0]}
  \end{sequentproof}
  Now observe that all the vector inequalities are component-wise, hence the real arithmetic equivalence $\exp(k\tau)\cdot z \succcurlyeq 0 \Liff \bigwedge^{n}_{i=1}(\exp(k\tau)\cdot z_{i} \succcurlyeq 0)$ is valid, 
  where each $z_{i}$ is a scalar. The postcondition in the succedent is rewritten 
  to this equivalence using \lref{pr:M[']}.
  \begin{sequentproof}
    \ax[]{\Gamma, \bigwedge^{n}_{i=1}(\exp(k\tau)\cdot z_{i} \succcurlyeq 0)}{\fBox[\ode{\alpha_{x\tau z}}[Q_{\tau} \Land \dot{\mathbf{h}} \geq -\boldsymbol{\eta}(\mathbf{h},\tau)]][\bigwedge^{n}_{i=1}(\exp(k\tau)\cdot z_{i} \succcurlyeq 0)]}
    \un[pr:M['],Real]{\Gamma, \exp(k\tau)\cdot z \succcurlyeq 0}{\fBox[\ode{\alpha_{x\tau z}}[Q_{\tau} \Land \dot{\mathbf{h}} \geq -\boldsymbol{\eta}(\mathbf{h},\tau)]][\exp(k\tau)\cdot z \succcurlyeq 0]}
  \end{sequentproof}
  Now we decompose the conjunctions in the postcondition using \lref{ax:box_and}, 
  which results in conjunctions of box modalities.
  \begin{sequentproof}
    \ax[]{\Gamma, \bigwedge^{n}_{i=1}(\exp(k\tau)\cdot z_{i} \succcurlyeq 0)}{\bigwedge^{n}_{i=1}(\fBox[\ode{\alpha_{x\tau z}}[Q_{\tau} \Land \dot{\mathbf{h}} \geq -\boldsymbol{\eta}(\mathbf{h},\tau)]][\exp(k\tau)\cdot z_{i} \succcurlyeq 0])}
    \un[ax:box_and]{\Gamma, \bigwedge^{n}_{i=1}(\exp(k\tau)\cdot z_{i} \succcurlyeq 0)}{\fBox[\ode{\alpha_{x\tau z}}[Q_{\tau} \Land \dot{\mathbf{h}} \geq -\boldsymbol{\eta}(\mathbf{h},\tau)]][\bigwedge^{n}_{i=1}(\exp(k\tau)\cdot z_{i} \succcurlyeq 0)]}
  \end{sequentproof}
  Now with the application of \lref{andR} we get $n$ stacked open premises, each of them 
  consisting a box modality with one of the scalar components of the original vector 
  inequality.
  \begin{sequentproof}
    \ax[]{\Gamma, \bigwedge^{n}_{i=1}(\exp(k\tau)\cdot z_{i} \succcurlyeq 0)}{\fBox[\ode{\alpha_{x\tau z}}[Q_{\tau} \Land \dot{\mathbf{h}} \geq -\boldsymbol{\eta}(\mathbf{h},\tau)]][\exp(k\tau)\cdot z_{1} \succcurlyeq 0]}
    \noLine
    \UnaryInfC{$\vdots$}
    \noLine
    \un[]{\Gamma, \bigwedge^{n}_{i=1}(\exp(k\tau)\cdot z_{i} \succcurlyeq 0)}{\fBox[\ode{\alpha_{x\tau z}}[Q_{\tau} \Land \dot{\mathbf{h}} \geq -\boldsymbol{\eta}(\mathbf{h},\tau)]][\exp(k\tau)\cdot z_{n} \succcurlyeq 0]}
    \un[andR]{\Gamma, \bigwedge^{n}_{i=1}(\exp(k\tau)\cdot z_{i} \succcurlyeq 0)}{\bigwedge^{n}_{i=1}(\fBox[\ode{\alpha_{x\tau z}}[Q_{\tau} \Land \dot{\mathbf{h}} \geq -\boldsymbol{\eta}(\mathbf{h},\tau)]][\exp(k\tau)\cdot z_{i} \succcurlyeq 0])}
  \end{sequentproof}
  From here we will only look at the $i$-th premise, as the same result can be 
  followed for all the $n$ premises. Since we have a scalar inequality in the postcondition 
  we can continue by applying axiom \lref{ax:DI_inequal} followed by \lref{dW}.
  \begin{sequentproof}
    \ax[]{\LocLip{\boldsymbol{\eta}}, \QuasiMon{\boldsymbol{\eta}}, Q, \tau \geq 0, \dot{\mathbf{h}} + \boldsymbol{\eta}(\mathbf{h},\tau) \geq 0}{\exp(k\tau)\cdot(kz_{i} + \eta_{i}(z,\tau)) \geq 0}
    \un[dW]{\Gamma, \bigwedge^{n}_{i=1}(\exp(k\tau)\cdot z_{i} \succcurlyeq 0)}{\fBox[\ode{\alpha_{x\tau z}}[Q_{\tau} \Land \dot{\mathbf{h}} \geq -\boldsymbol{\eta}(\mathbf{h},\tau)]][\exp(k\tau)\cdot (kz_{i} + \eta_{i}(z,\tau)) \geq 0]}
    \un[ax:DI_inequal]{\Gamma, \bigwedge^{n}_{i=1}(\exp(k\tau)\cdot z_{i} \succcurlyeq 0)}{\fBox[\ode{\alpha_{x\tau z}}[Q_{\tau} \Land \dot{\mathbf{h}} \geq -\boldsymbol{\eta}(\mathbf{h},\tau)]][\exp(k\tau)\cdot z_{i} \succcurlyeq 0]}
  \end{sequentproof}
  Now continuing with the proof of the resulting real arithmetic. We first construct the following 
  vector.
  \[\tilde{z} = (z_{1}, \dots, z_{i-1}, 0, z_{i+1}, \dots, z_{n})^{\top}\]
  The vector $\tilde{z}$ agrees with $z$ on every component except the $i$-th, the 
  $i$-th component, however, agrees with the $i$-th component of the zero vector $\mathbf{0}$ (where all the components are 0).
  Now from the definition of $\vecG{\boldsymbol{\eta}}$ we know,
  \[\eta_{i}(0,\tau) = 0 \Land  (0 \leq \tilde{z} \to \eta_{i}(0,\tau) \leq \eta_{i}(\tilde{z},\tau))\] 
  from which we get $\eta_{i}(\tilde{z},\tau) \geq 0$ that can be rewritten as,
  \begin{align*}
    0 &\geq -\eta_{i}(\tilde{z},\tau) \\
    \eta_{i}(z,\tau) &\geq \eta_{i}(z,\tau) - \eta_{i}(\tilde{z},\tau)
  \end{align*}
  From the local Lipschitz condition with the quantifiers instantiated as $a = z$, 
  $x = z$, $y = \tilde{z}$, $r > 0$, and $k \geq 0$ we get,
  \[\norm{\tilde{z} - z}^{2} < r^{2} \to \norm{\boldsymbol{\eta}(z,\tau) - \boldsymbol{\eta}(\tilde{z},\tau)}^{2} \leq k^{2}\norm{z - \tilde{z}}^{2}\]
  Now observe that $\tilde{z}_{j} = z_{j}$ for all $j \neq i$, this means $\norm{\tilde{z} - z}^{2} = (0 - z_{i})^{2} < r^{2}$. 
  We also know that each component is bounded by the norm, hence we get,
  \begin{align*}
    (\eta_{i}(z,\tau) - \eta_{i}(\tilde{z},\tau))^{2} \leq \norm{\boldsymbol{\eta}(z,\tau) - \boldsymbol{\eta}(\tilde{z},\tau)}^{2} &\leq k^{2}\norm{z - \tilde{z}}^{2} \\
    (\eta_{i}(z,\tau) - \eta_{i}(\tilde{z},\tau))^{2} &\leq k^{2}(0 - z_{i})^{2}
  \end{align*} 
  Combining the previous inequalies we get the desired relation.
  \[kz_{i} + \eta_{i}(z,\tau) \geq 0\]
  \begin{sequentproof}
    \ax*[Real]{\LocLip{\boldsymbol{\eta}}, \QuasiMon{\boldsymbol{\eta}}, Q, \tau \geq 0, \dot{\mathbf{h}} + \boldsymbol{\eta}(\mathbf{h},\tau) \geq 0}{\exp(k\tau)\cdot(kz_{i} + \eta_{i}(z,\tau)) \geq 0}
  \end{sequentproof}
  The branch closes by an application of \lref{Real}. This reasoning repeated $n$ times 
  to close all the $n$ open premises. Finally coming to the open premise \textcircled{f}. 
  This branch closes by following the same reasoning as that of premise \textcircled{d}, as they 
  are proving the same postcondition for the same set of ODEs, i.e., $\alpha_{x\tau}$. 
  \begin{sequentproof}
    \ax*[pr:M['],ax:box_and,ax:DI_inequal,Real]{\Gamma, \mathbf{h} \geq z}{\fBox[\ode{\alpha_{x\tau}}[Q_{\tau}]][z \succcurlyeq 0]}
  \end{sequentproof}
  This finishes the proof of the vector comparison invariants axiom \lref{ax:VCI}.

  \paragraph{(\lref{pr:vCI})} For the sake of completeness we also give the derivation 
  of the proof rule \lref{pr:vCI}.
  \begin{sequentproof}
    \ax[]{}{\LocLip{\boldsymbol{\eta}} \Land \vecG{\boldsymbol{\eta}}}
    \ax[]{Q}{\exists u(\mathcal{U} \Land \dot{\mathbf{h}} + \boldsymbol{\eta}(\mathbf{h},\tau) \geq 0)}
    \bi[andR,dW]{}{\LocLip{\boldsymbol{\eta}} \Land \vecG{\boldsymbol{\eta}} \Land \fBox[\odes{x' = \mathbf{f}(x,u,\tau), \tau' = 1}[Q]][\exists u(\mathcal{U} \Land \dot{\mathbf{h}} + \boldsymbol{\eta}(\mathbf{h},\tau) \geq 0)]}
    \un[ax:VCI]{\mathbf{h} \succcurlyeq 0}{\exists u\fBox[\odes{x' = \mathbf{f}(x,u,\tau), \tau' = 1}[Q]][\mathbf{h} \succcurlyeq 0]}
  \end{sequentproof}
  The required premises are derived using \lref{ax:VCI}, \lref{andR}, and \lref{dW}.

\end{proof}

\subsubsection{Control Barrier Functions}\label{appendix:cbf}
The proofs of the axioms and proof rules of the (scalar) control barrier function and 
vector control barrier functions are given here.

\begin{proof}[Proof (Corollary~\ref{cor:cbf})]
  The proof utilizes the following abbreviations. 
  \begin{align*}
    \Gamma &\equiv \LocLip{\eta} \Land \K{\eta} \Land h \geq 0 \\
    \alpha_{x\tau} &\equiv x' = \mathbf{f}(x,u,\tau), \tau' = 1 \\
    \K{\eta} &\equiv \forall t \forall x \forall y(\eta(0,t) = 0 \Land  (x > y \to \eta(x,t) > \eta(y,t))) \\
    \G{\eta} &\equiv \forall t \forall x \forall y(\eta(0,t) = 0 \Land  (x \geq y \to \eta(x,t) \geq \eta(y,t)))
  \end{align*}
  \paragraph{(\lref{ax:CBF})} The proof of the axiom is as follows.
  Observe that a class $\mathcal{K}$ function implies that the function is also 
  class $\mathcal{G}$, i.e., $\K{\eta} \to \G{\eta}$, which is unfolded as following.
  \[\forall t \forall x \forall y(\eta(0,t) = 0 \Land (x > y \to \eta(x,t) > \eta(y,t))) \to \forall t \forall x \forall y(\eta(0,t) = 0 \Land (x \geq y \to \eta(x,t) \geq \eta(y,t)))\]
  This implication is provable in real arithmetic, hence the formula $\G{\eta}$ can 
  be soundly cut into the sequent. This can also be seen from the fact that all strictly increasing 
  functions are also non-decreasing. The postcondition in the antecedent is 
  monotonically replaced by $\mathcal{U} \Land \dot{h} + \eta(h,\tau) \geq 0$ using \lref{pr:M[']}.
  \begin{sequentproof}
    \ax[]{\Gamma, \G{\eta}, \fBox[\ode{\alpha_{x\tau}}[Q]][(\mathcal{U} \Land \dot{h} + \eta(h,\tau) \geq 0)]}{\exists u\fBox[\ode{\alpha_{x\tau}}[Q]][h \geq 0]}
    \un[pr:M[']]{\Gamma, \G{\eta}, \fBox[\ode{\alpha_{x\tau}}[Q]][\exists u (\mathcal{U} \Land \dot{h} + \eta(h,\tau) \geq 0)]}{\exists u\fBox[\ode{\alpha_{x\tau}}[Q]][h \geq 0]}
    \un[cut,Real]{\Gamma, \fBox[\ode{\alpha_{x\tau}}[Q]][\exists u (\mathcal{U} \Land \dot{h} + \eta(h,\tau) \geq 0)]}{\exists u\fBox[\ode{\alpha_{x\tau}}[Q]][h \geq 0]}
  \end{sequentproof}
  Continuing from there, axiom \lref{ax:box_and} decomposes the conjunction in the 
  postcondition of the antecedent into two separate box modalities with the 
  postcondition separated. After this the axiom \lref{ax:SCI} can be applied, which 
  finishes the proof.
  \begin{sequentproof}
    \ax*[ax:SCI]{\Gamma, \G{\eta}, \fBox[\ode{\alpha_{x\tau}}[Q]][\mathcal{U}], \fBox[\ode{\alpha_{x\tau}}[Q]][\dot{h} + \eta(h,\tau) \geq 0]}{\exists u\fBox[\ode{\alpha_{x\tau}}[Q]][h \geq 0]}
    \un[ax:box_and]{\Gamma, \G{\eta}, \fBox[\ode{\alpha_{x\tau}}[Q]][(\mathcal{U} \Land \dot{h} + \eta(h,\tau) \geq 0)]}{\exists u\fBox[\ode{\alpha_{x\tau}}[Q]][h \geq 0]}
  \end{sequentproof}
  \paragraph{(\lref{pr:cbf})}
  The proof rule is derived as follows.
  \begin{sequentproof}
    \ax[]{}{\LocLip{\eta} \Land \K{\eta}}
    \ax[]{Q}{\exists u (\mathcal{U} \Land \dot{h} + \eta(h,\tau) \geq 0)}
    \bi[andR,dW]{}{\LocLip{\eta} \Land \K{\eta} \Land \fBox[\ode{x' = \mathbf{f}(x,u), \tau' = 1}[Q]][\exists u (\mathcal{U} \Land \dot{h} + \eta(h,\eta) \geq 0)]}
    \un[ax:CBF]{h \geq 0}{\exists u\fBox[\odes{x' = \mathbf{f}(x,u), \tau' = 1}[Q]][h \geq 0]}
  \end{sequentproof}
  The required two open premises are derived using axiom \lref{ax:CBF}, followed 
  by \lref{andR} and \lref{dW}.
\end{proof}

\begin{proof}[Proof (Corollary~\ref{cor:vcbf})]\label{proof:vcbf}
  The proof of the vectorial CBF axiom and proof rule utilize the following abbreviations.
  \begin{align*}
    \Gamma &\equiv \LocLip{\boldsymbol{\eta}} \Land \vecK{\boldsymbol{\eta}} \Land h \geq 0 \qquad &&\alpha_{x\tau} \equiv x' = \mathbf{f}(x,u), \tau' = 1 \\
    \vecK{\boldsymbol{\eta}} &\equiv \bigwedge^{n}_{i = 1} \K{\eta_{i}} \qquad &&\vecG{\boldsymbol{\eta}} \equiv \bigwedge^{n}_{i = 1} \G{\eta_{i}} \\
  \end{align*}
  \paragraph{(\lref{ax:VCBF})} The proof of the axiom begins by first cutting in 
  $\vecG{\boldsymbol{\eta}}$ into the antecedent. This cut is sound because, 
  strictly increasing functions are also non-decreasing, which means 
  $\vecK{\boldsymbol{\eta}} \to \vecG{\boldsymbol{\eta}}$, and after unfolding the 
  two definitions the following real arithmetic formula is provable.
  \[\bigwedge^{n}_{i = 1} \K{\eta_{i}} \to \bigwedge^{n}_{i = 1} \G{\eta_{i}}\]
  After which the postcondition in the antecedent is rewritten monotonically to 
  $\mathcal{U} \Land \dot{\mathbf{h}} + \boldsymbol{\eta}(\mathbf{h}) \geq 0$ using 
  \lref{pr:M[']}, as the resulting premise with the existential quantifier can be 
  resolved using \lref{exL}.
  \begin{sequentproof}
    \ax[]{\Gamma, \vecG{\boldsymbol{\eta}}, \fBox[\ode{\alpha_{x\tau}}[Q]][(\mathcal{U} \Land \dot{\mathbf{h}} + \boldsymbol{\eta}(\mathbf{h},\tau) \geq 0)]}{\exists u \fBox[\ode{\alpha_{x\tau}}[Q]][\mathbf{h} \geq 0]}
    \un[pr:M[']]{\Gamma, \vecG{\boldsymbol{\eta}}, \fBox[\ode{\alpha_{x\tau}}[Q]][\exists u(\mathcal{U} \Land \dot{\mathbf{h}} + \boldsymbol{\eta}(\mathbf{h},\tau) \geq 0)]}{\exists u \fBox[\ode{\alpha_{x\tau}}[Q]][\mathbf{h} \geq 0]}
    \un[cut,Real]{\Gamma, \fBox[\ode{\alpha_{x\tau}}[Q]][\exists u(\mathcal{U} \Land \dot{\mathbf{h}} + \boldsymbol{\eta}(\mathbf{h},\tau) \geq 0)]}{\exists u \fBox[\ode{\alpha_{x\tau}}[Q]][\mathbf{h} \geq 0]}
  \end{sequentproof}
  Continuing from there, we first separate the conjunction in the postcondition 
  of the antecedent into $[\ode{\alpha_{x\tau}}[Q]]\mathcal{U}$ and $[\ode{\alpha_{x\tau}}[Q]]\dot{\mathbf{h}} + \boldsymbol{\eta}(\mathbf{h},\tau) \geq 0$. 
  After this with a straightforward application of axiom \lref{ax:VCI} finishes the 
  proof.
  \begin{sequentproof}
    \ax*[ax:VCI]{\Gamma, \vecG{\boldsymbol{\eta}}, \fBox[\ode{\alpha_{x\tau}}[Q]][\mathcal{U}], \fBox[\ode{\alpha_{x\tau}}[Q]][\dot{\mathbf{h}} + \boldsymbol{\eta}(\mathbf{h},\tau) \geq 0]}{\exists u \fBox[\ode{\alpha_{x\tau}}[Q]][\mathbf{h} \geq 0]}
    \un[ax:box_and]{\Gamma, \vecG{\boldsymbol{\eta}}, \fBox[\ode{\alpha_{x\tau}}[Q]][(\mathcal{U} \Land \dot{\mathbf{h}} + \boldsymbol{\eta}(\mathbf{h},\tau) \geq 0)]}{\exists u \fBox[\ode{\alpha_{x\tau}}[Q]][\mathbf{h} \geq 0]}
  \end{sequentproof}

  \paragraph{(\lref{pr:vcbf})} The proof rule derives from axiom \lref{ax:VCBF} as 
  follows.
  \begin{sequentproof}
    \ax[]{}{\LocLip{\boldsymbol{\eta}} \Land \vecK{\boldsymbol{\eta}}}
    \ax[]{Q}{\exists u(\mathcal{U} \Land \dot{\mathbf{h}} + \boldsymbol{\eta}(\mathbf{h},\tau))}
    \bi[andR,dW]{}{\LocLip{\boldsymbol{\eta}} \Land \vecK{\boldsymbol{\eta}} \Land \fBox[\odes{x' = \mathbf{f}(x,u), \tau' = 1}[Q]][\exists u(\mathcal{U} \Land \dot{\mathbf{h}} + \boldsymbol{\eta}(\mathbf{h},\tau))]}
    \un[ax:VCBF]{\mathbf{h} \geq 0}{\exists u \fBox[\odes{x' = \mathbf{f}(x,u), \tau' = 1}[Q]][\mathbf{h} \geq 0]}
  \end{sequentproof}
\end{proof}

\subsection{Special Instances}\label{appendix:special-instances}

This section gives the proofs for the special instances of comparison invariants, 
which include, scalar and vector Darboux invariants \cite{DBLP:journals/jacm/PlatzerT20}, 
and differential invariants \cite{DBLP:journals/jar/Platzer17,DBLP:journals/jacm/PlatzerT20}.

\paragraph{Darboux Invariants}
\begin{proof}[Proof (Corollary~\ref{cor:scalar-darboux})]
  We prove the axiom \lref{ax:DBX_inequal} for a parameterized system of ODEs. The proof 
  begins by a straightforward application of \lref{ax:SCI}, the resulting left premise 
  closes, because we choose the comparison function $\eta(p,\tau) = gp$, for a cofactor 
  polynomial $g$, which happens to be nondecreasing, hence satisfying $\LocLip{\eta} \Land \G{\eta}$. 
  This is because all polynomials are locally Lipschitz and nondecreasing along 
  with implies class $\mathcal{G}$.
  \begin{sequentproof}
    \ax*[Real]{}{\LocLip{g} \Land \G{g}}
    \ax[]{\fBox[\odes{x' = \mathbf{f}(x,u), \tau' = 1}[Q]][(p)' \geq gp]}{\textcircled{a}}
    \bi[ax:SCI]{\fBox[\odes{x' = \mathbf{f}(x,u), \tau' = 1}[Q]][(p)' \geq gp], p \succcurlyeq 0}{\exists u\fBox[\odes{x' = \mathbf{f}(x,u), \tau' = 1}[Q]][p \succcurlyeq 0]}
  \end{sequentproof}
  Continuing with \textcircled{a}, we first rewrite the postcondition without the 
  existential quantifier using \lref{pr:M[']}, after setting $\mathcal{U} \equiv true$, 
  we get desired formula as the antecedent.
  \begin{sequentproof}
    \ax*[]{\fBox[\odes{x' = \mathbf{f}(x,u), \tau' = 1}[Q]][\dot{p} \geq gp]}{\fBox[\odes{x' = \mathbf{f}(x,u), \tau' = 1}[Q]][\dot{p} \geq gp]}
    \un[()',Real]{\fBox[\odes{x' = \mathbf{f}(x,u), \tau' = 1}[Q]][(p)' \geq gp]}{\fBox[\odes{x' = \mathbf{f}(x,u), \tau' = 1}[Q]][\mathcal{U} \Land \dot{p} \geq gp]}
    \un[pr:M['],exL]{\fBox[\odes{x' = \mathbf{f}(x,u), \tau' = 1}[Q]][(p)' \geq gp]}{\fBox[\odes{x' = \mathbf{f}(x,u), \tau' = 1}[Q]][\exists u(\mathcal{U} \Land \dot{p} \geq gp)]}
  \end{sequentproof}
\end{proof}

\begin{proof}[Proof (Corollary~\ref{cor:vector-darboux})]
  The existing vectorial Darboux invariants axiom \lref{ax:VDBX} is derivable from 
  axiom \lref{ax:VCI}. The proof uses the abbreviation $\alpha_{x\tau} \equiv \odes{x' = \mathbf{f}(x, u, \tau), \tau' = 1}$.
  \begin{sequentproof}
    \ax*[ax:box_and,ax:VCI]{\mathbf{p} \geq 0 \Land -\mathbf{p} \geq 0, \fBox[\ode{\alpha_{x\tau}}[Q]][((\mathbf{p})' - G\mathbf{p} \geq 0 \Land -(\mathbf{p})' + G\mathbf{p} \geq 0)]}{\exists u\fBox[\ode{\alpha_{x\tau}}[Q]][(\mathbf{p} \geq 0 \Land -\mathbf{p} \geq 0)]}
    \un[pr:M['],Real]{\mathbf{p} = 0, \fBox[\ode{\alpha_{x\tau}}[Q]][(\mathbf{p})' = G\mathbf{p}]}{\exists u\fBox[\ode{\alpha_{x\tau}}[Q]][\mathbf{p} = 0]}
  \end{sequentproof}
  The key insight here is the equivalence $\mathbf{p} = 0 \Liff (\mathbf{p} \geq 0 \Land -\mathbf{p} \geq 0)$, 
  which converts the equality into a conjunction of inequalies. After this an application 
  of \lref{ax:box_and} decomposes the conjunction into separate box modalities, followed 
  by the application of \lref{ax:VCI}. Observe that $\boldsymbol{\eta}(\mathbf{p},\tau) = G \cdot \mathbf{p}$. 
  Here the square matrix $G$ is chosen such that it matches the dimension of $\mathbf{p}$, i.e., 
  if $\mathbf{p}$ is an $n$-dimensional vector then $G$ is $n\times n$ square matrix. 
  More importantly $G$ has to be such that the product $G \cdot \mathbf{p}$ is locally Lipschitz 
  and moreover it also has to be class $\mathcal{G}$. The former is satisfied for all 
  real polynomial entries i.e., any standard \dL terms, for the latter, however, 
  $G$ has to be an essentially positive matrix or a \emph{Metzler matrix}, which means all the 
  off-diagonal elements are nonnegative \cite[\S10.XII]{walter2013ordinary}.
\end{proof}
\paragraph{Differential Invariants}
\begin{proof}[Proof (Corollary~\ref{cor:dI})]
  The existing \dL proof rule \lref{dI_inequal} is derived as follows using proof 
  rule \lref{pr:sCI}.
  \begin{sequentproof}
    \ax*[Real]{}{\LocLip{\eta} \Land \G{\eta}}
    \ax[]{Q}{\dot{p} \geq 0}
    \un[exR,Real]{Q}{\exists u(\mathcal{U} \Land \dot{p} + \eta(p) \geq 0)}
    \bi[pr:sCI]{p \succcurlyeq 0}{\exists u\fBox[\odes{x' = \mathbf{f}(x,u), \tau' = 1}[Q]][p \succcurlyeq 0]}
  \end{sequentproof}
  The desired premise is obtained after setting $\mathcal{U} \equiv true$ and 
  $\eta(p,\tau) = 0 \cdot p$, which is the zero function, hence is locally Lipschitz and 
  nondecreasing (implying $\G{\eta}$).
\end{proof}

\fi

\end{document}